\documentclass[11pt]{article}

\usepackage{pdfpages}
\usepackage{lscape}
\usepackage{longtable}
\usepackage{amsmath,amsfonts,amsthm,amssymb,bbm,nicefrac,mathtools}
\usepackage{authblk}
\usepackage{graphicx}
\usepackage{mathrsfs}
\usepackage{multirow}
\usepackage{blkarray, bigstrut} %
\usepackage{soul,xcolor, color, colortbl}
\usepackage{multicol, tcolorbox}
\usepackage{tikz}
\usepackage{tikz-cd}
\usetikzlibrary{matrix,arrows,positioning,calc,shapes}
\usetikzlibrary{arrows.meta}
\usepackage{booktabs}

\usepackage{dsfont}
\usepackage{subcaption}
\usepackage[cmtip,arrow]{xy}
\usepackage{pb-diagram,pb-xy}
\usepackage{multirow}
\usepackage{phoenician}
\usepackage[normalem]{ulem}
\usepackage{array}
\usepackage{parskip}

\usepackage{hyperref}
\usepackage{wrapfig}
\usepackage{natbib}
\usepackage{url}
\usepackage{algorithm}
\usepackage[noend]{algpseudocode}
\usepackage{varwidth}

\usepackage[shortlabels]{enumitem}

\usetikzlibrary{shapes, arrows, calc, arrows.meta, fit, positioning}

\usetikzlibrary{patterns,patterns.meta}

\usepackage{listings}

\definecolor{codegreen}{rgb}{0,0.6,0}
\definecolor{codegray}{rgb}{0.5,0.5,0.5}
\definecolor{codepurple}{rgb}{0.58,0,0.82}
\definecolor{backcolour}{rgb}{0.95,0.95,0.92}

\newtheorem{thm}{Theorem}[section]
\newtheorem{lem}[thm]{Lemma}
\newtheorem{prop}[thm]{Proposition}

\theoremstyle{definition}
\newtheorem{defn}[thm]{Definition}
\newtheorem{rem}[thm]{Remark}
\newtheorem{ex}[thm]{Example}

\newcommand{\setof}[1]{\left\{ {#1}\right\}}
\newcommand{\setdef}[2]{\left\{{#1}\,\left|\,{#2}\right.\right\}}

\newcommand{\bq}{{\bf q}}

\newcommand{\bw}{{\bf w}}

\newcommand{\bI}{{\bf I}}

\newcommand{\B}{{\mathbb{B}}}

\newcommand{\I}{{\mathbb{I}}}

\newcommand{\R}{{\mathbb{R}}}

\newcommand{\Z}{{\mathbb{Z}}}

\newcommand{\bzero}{{\bf 0}}
\newcommand{\bone}{{\bf 1}}

\newcommand{\cD}{\mathcal{D}}
\newcommand{\cE}{{\mathcal E}}
\newcommand{\cF}{{\mathcal F}}
\newcommand{\cG}{{\mathcal G}}

\newcommand{\cL}{{\mathcal L}}
\newcommand{\cM}{{\mathcal M}}

\newcommand{\cP}{{\mathcal P}}

\newcommand{\cT}{{\mathcal T}}

\newcommand{\cX}{{\mathcal X}}

\newcommand{\cZ}{{\mathcal Z}}

\newcommand{\Top}{{\mathrm{Top}}}
\newcommand{\TP}{{\mathrm{TP}}}

\newcommand{\Targets}{\mathbf{T}}
\newcommand{\Sources}{\mathbf{S}}

\newcommand{\PG}{\mathsf{PG}}

\newcommand{\ol}[1]{\overline{#1}}

\newcommand{\bbB}{\mathbb{B}}

\newcommand{\mvmap}{\rightrightarrows}

\newcommand{\CG}{\mathsf{CG}}
\newcommand{\RC}{\mathsf{RC}}
\newcommand{\MG}{\mathsf{MG}}

\newcommand{\SCC}{\mathsf{SCC}}
\newcommand{\GRC}{\mathsf{GRC}}

\definecolor{proofscolor}{rgb}{0,0.2,1}
\definecolor{myblue}{rgb}{0,0.2,.7}
\definecolor{mygreen}{rgb}{0.1, .9, 0.4 }
\definecolor{mycyan}{rgb}{0, 1, .9 }

\newcommand{\cPE}{{\mathcal{P}_\mathsf{E}}}
\newcommand{\cEE}{{\mathcal{E}_\mathsf{E}}}
\newcommand{\PGE}{\mathsf{PG}_\mathsf{E}}
\newcommand{\snail}{\mathsf{Snail1}}
\newcommand{\tgfb}{\mathsf{TGFb}}
\newcommand{\miRtwo}{\mathsf{miR200}}
\newcommand{\miRthree}{\mathsf{miR34a}}
\newcommand{\zeb}{\mathsf{Zeb1}}
\newcommand{\ovol}{\mathsf{Ovol2}}
\newcommand{\EMT}{\mathsf{EMT}}

\title{Asynchronous Boolean models coarsely sample continuous dynamics of regulatory networks}

\author[1]{Breschine Cummins}
\author[,2]{Marcio Gameiro\thanks{Corresponding author: gameiro@math.rutgers.edu}}
\author[1]{Tom\'{a}\v{s} Gedeon}
\author[2]{Konstantin Mischaikow}
\author[3]{Bernardo Rivas}

\affil[1]{Department of Mathematical Sciences, Montana State University, Bozeman, MT, USA}
\affil[2]{Department of Mathematics, Rutgers University, Piscataway, NJ, USA}
\affil[3]{Department of Mathematics and Statistics, University of Toledo, Toledo, OH, USA}

\begin{document}

\maketitle

\begin{abstract}
Asynchronous monotone Boolean models are widely used to characterize the dynamics of gene regulatory networks. However, their coarse state discretization limits their ability to capture complex continuous dynamics and continuous parameter dependencies. In this paper, we present a rigorous mathematical framework that embeds asynchronous monotone Boolean models into a broader class of multilevel combinatorial models, which in turn embed into the Dynamic Signatures Generated by Regulatory Networks (DSGRN) methodology. We define the DSGRN parameter graph, which encodes the notion of parameter adjacency and is used to map Boolean functions to specific nodes within the DSGRN parameter space. We prove that these multilevel discrete update functions act as a multilevel refinement of asynchronous monotone Boolean models. We demonstrate that this class of Boolean models systematically underestimates network dynamics by missing crucial intermediate behaviors such as higher-order multistability and stable periodic orbits. We show that the DSGRN framework captures a strictly richer set of dynamics consistent with ordinary differential equations (ODEs), providing a mathematically rigorous and computationally viable bridge between discrete and continuous network modeling.
\end{abstract}


\section{Introduction}

The aphorism ``all models are wrong, but some are useful'' is clearly relevant in the domain of systems biology, and more specifically in attempts to characterize the time-dependent behavior of gene regulatory networks. The most fundamental model in this setting is that of an \emph{interaction network} in which the structure of the interactions between the genes, proteins, and other regulatory elements is presented as a signed, directed graph with nodes representing the regulatory elements and edges the regulation. The signs represent up- or down-regulation with an implied monotonicity of the response of the target with respect to abundance of the regulator. However, an interaction network does not contain sufficient information to specify dynamic behavior. 

To model the dynamics of an interaction network,  the modeler needs to ask whether a discrete or a continuous approximation of the state of the system and a discrete or continuous evolution  provides useful information about system function. 
On one end of the spectrum are   monotone Boolean models that  use discrete time and the coarsest possible discretization of the state, while on the other ordinary differential equations assume that the dynamics is continuous in both time and state. Between these two modeling approaches lies a gradation of combinatorial models. In this paper we consider  three combinatorial models that embed into each other: asynchronous monotone Boolean, asynchronous monotone multilevel, and Dynamic Signatures Generated by Regulatory Networks (DSGRN)~\cite{Cummins16,Gedeon20,Cummins2017b,Cummins17,gameiro:gedeon:kepley:mischaikow}. We use the term \textit{multilevel} map for a map with a discrete range with more than two values and reserve the term \textit{multivalued} map for maps that assign to one $x$ more than one value.
These models have differing levels of expressiveness, computational cost, and fidelity to the dynamics of continuous systems. The most detailed  model, DSGRN,  allows systematic incorporation of asynchronous monotone Boolean models within the class of  multilevel discrete models, expanding the range of observed dynamics at a computational cost comparable to that of  monotone Boolean models. Furthermore, there is a  rigorous procedure to identify the DSGRN characterization of dynamics with that of continuous models~\cite{rook_field:24}.

B. Goodwin  is credited with introducing the first explicit mathematical models, taking the form of systems of ordinary differential equations (ODEs), that give rise to well-defined dynamics and precise results with respect to both quantities and times \cite{goodwin}.
Though the models presented in \cite{goodwin} are, by today's standards, extremely simple, the essential challenges are already evident: what form of explicit nonlinearities should be chosen, what are appropriate parameter values, and which numerical simulations are needed to identify the dynamics.

These three challenges are avoided in the Boolean models introduced by S. Kauffman \cite{kauffman} where the values associated with the nodes of the regulatory network are either \textsf{ON} or \textsf{OFF}.
The combinatorial nature of these models allows for extremely efficient computations and meaningful interpretability in biological contexts.
As a consequence there has been significant effort and development of Boolean modeling~\cite{Thomas1973,deJong2004, Thomas95, Thieffry99, Jaoude16, albert:othmer,Glass1975,CH2011,chaouiya2016,chatain2020,Albert2017,Assmann2014,Thieffry2016,Zhang2008,aledo2024,pauleve2019,Pauleve_MP,laub_2016,kadelka_laub_2017,kadelka_2026}. 
These models provide direct interpretability in biological contexts~\cite{plaugher_2021,zanudo_albert_2015,Leukemia_2025}. On the computational side, several specialized algorithms allow computation of the steady states~\cite{SAT,melkman2013,Mori2022}, reachability between sets of states, and computations of attractors~\cite{tamura2009,dubrova2012,Mori2022} in Boolean models of networks with dozens of genes.

Monotone Boolean models (MBM) are an important subclass of Boolean models where the update function at each node is a \emph{monotone Boolean function} (MBF), (Definition~\ref{defn:MBF}). There are many  MBMs compatible with the same interaction network and they can be readily enumerated~\cite{Gedeon24b,MBF}.

In spite of many advantages, there are at least three significant limitations to Boolean models~\cite{chaves:albert:sontag}.
\begin{itemize}
\item 
First, there exist systems for which the measurements of the quantities associated to nodes appear to take on more than two values.
\item 
Second, there is no natural way to embed the Boolean values of \textsf{ON} and \textsf{OFF} into the real values $[0,\infty)$ associated with experimental measurements.
\item 
Third, for the purposes of understanding the impact of mutations or epigenetic factors or for control of phenotypes, it is important to understand the impact of parameters on the behavior of the network. 
\end{itemize}
One approach to dealing with the first limitation is to replace the Boolean model with a combinatorial model wherein the nodes are allowed to take on a finite set of values. 
The extension of Boolean models to finite multilevel models does not address the fact that there is no natural embedding of these combinatorial models into real-valued systems. 
The most natural embedding is built component-wise from an embedding of a finite set of integers into an interval, where each integer is mapped to a subinterval.
The ordering and adjacency of the subintervals is captured by a cell complex \cite{lefschetz}.
A cell complex is a finite combinatorial representation of a continuous phase space, where the connectedness is encoded in the contiguity relation between the cells of the complex~(see Section~\ref{sec:math}).
This representation of discrete dynamics is used by the framework called  Dynamic Signatures Generated by Regulatory Networks (DSGRN).
For the purpose of this introduction it is sufficient to know that there are explicit, verifiable conditions under which computations performed using the cell complex are valid in the continuous setting of ODEs. One of these conditions is that the cell complex is sufficiently large with the number of levels associated to each node related to the number of out-edges from each node. 
Thus, DSGRN provides a mathematically well-defined means of transitioning between combinatorial and continuous models for the dynamics of gene regulatory networks.
For the details of this transition the reader is referred to~\cite{rook_field:24}.

The focus of this paper is on the third limitation, that of parameterization. 
The efficacy of combinatorial models arises from the fact that they can only take on finitely many states, and thus their parameterization is inherently finite.
In Section~\ref{sec:math} we introduce the DSGRN parameter graph.
The importance of this structure is that:
\begin{enumerate}
\item Each node of the DSGRN parameter graph represents a multilevel combinatorial model.
\item Asynchronous monotone Boolean models can be identified with a subset of parameter  graph nodes and their dynamics can be compared with the dynamics of their neighboring asynchronous monotone multilevel models.
\item It provides a framework by which  multilevel combinatorial models can be identified with ODE models~\cite{rook_field:24}.
\end{enumerate}

As alluded to in the title we will exclusively use the asynchronous update of both Boolean and multilevel models. Our primary motivation comes from the assumption of underlying continuous network dynamics, which implies that  the transitions should only  happen between neighboring states. 
Other updates, most notably the \textit{most permissive update} (MP) introduced in \cite{Pauleve_MP,pauleve_sene}, are very useful for their computational efficiency \cite{Pauleve_MP}. Since the most permissive update provably overestimates asynchronous transitions, the number of attractors under MP will be at most equal to the number of attractors under asynchronous update~\cite{Pauleve_MP}.

We show on a series of examples that asynchronous monotone Boolean functions often miss dynamics that is captured by the broader class of asynchronous multilevel functions. For example, in Section~\ref{sec:more_examples} we present a three-node network in Figure~\ref{fig:EMT_network_intro}(Left) where asynchronous monotone Boolean models detect bistability and tristability of fixed points, but DSGRN detects instances of monostability up through 6-stability, as well as the existence of periodic orbits.  This greatly enlarges the  array of potential dynamics and  has implications for networks of biological interest. As an example of a biologically relevant network, in Section~\ref{sec:more_examples} we consider the simplified epithelial-mesenchymal network important to cancer ~\cite{Hari_EMT_2020,Sullivan_EMT_2023,xin:cummins:gedeon} (see Figure~\ref{fig:EMT_network_intro}(Right)).  For parameters that we examined, DSGRN shows existence of $7$-stability with $4$ so-called intermediate states in addition to mesenchymal and epithelial states. This  multistable dynamics may be underestimated by Boolean models.

The DSGRN software efficiently enumerates and computes recurrent sets of dynamics of multilevel maps for a given network~\cite{DSGRN_repo}. Since asynchronous monotone Boolean models miss potential dynamics, there is no reason not to use the multilevel discrete models when network size permits since these  provably describe network dynamics of all ODEs with sufficiently steep nonlinearities~\cite{rook_field:24}.

\begin{figure*}[!htb]
\centering
\includegraphics[width=0.7\linewidth]{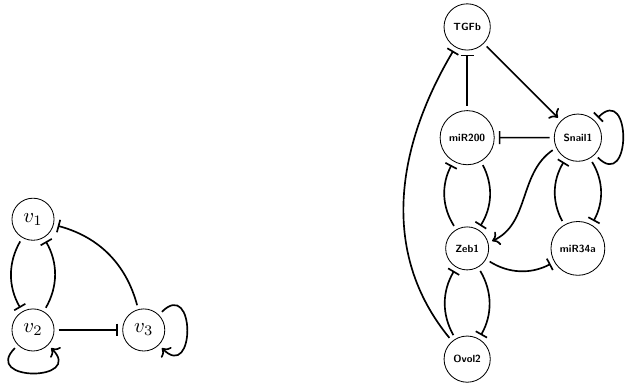}
\caption{(Left) Example 3-node network. (Right) Regulatory network for reversible epithelial-to-mesenchymal transition.}
\label{fig:EMT_network_intro}
\end{figure*}

We organize the paper as follows. We begin with simple examples that demonstrate the connection between monotone Boolean models and their more complex multilevel and DSGRN counterparts (Section~\ref{sec:results}). We continue with the rigorous mathematical details (Section~\ref{sec:math}), a Discussion (Section~\ref{sec:discuss}), additional information on lattices of monotone Boolean models (Section~\ref{sec:lattices}), and software guidance (Section~\ref{sec:code}).

\section{Results}
\label{sec:results}

\subsection{Exposition through examples}
\label{sec:examples_intro}

As indicated in the introduction, the goal of this paper is to explain the relationship between asynchronous monotone Boolean models, multilevel combinatorial models, and DSGRN.
However, the mathematics underlying these relationships is rather technical, and thus the goal of this section is to conceptually describe  the progression through examples. We provide two one-dimensional examples and one two-dimensional example to illustrate the steps from Boolean dynamics to DSGRN dynamics in Sections~\ref{subsec:snn} and \ref{subsec:TS}. In these examples, the multilevel combinatorial model is trivially the Boolean dynamics. In Sections~\ref{sec:multilevel_models} and \ref{sec:orders}, we introduce more complex networks in which the multilevel framework is distinctly non-Boolean. The DSGRN parameter graph is introduced in Section~\ref{subsec:DSGRNPGBLattices}, and further examples focused solely on DSGRN output are given in Section~\ref{sec:more_examples}.

\subsubsection{Single Node Networks}
\label{subsec:snn}

We begin our exposition by considering the simplest networks, those consisting of one node.
Figure~\ref{fig:singleDown}(a) shows a one-node network with a repressing edge.
Let us consider the associated \textit{monotone Boolean model} (Definition~\ref{defn:MBM}).
At any point in time the single node in the network, labeled $v_1$ in Figure~\ref{fig:singleDown}(a),  can take one of two states in $\B = \setof{0, 1}$, with $0 < 1$.
The dynamics is generated by a \textit{monotone Boolean function} (Definition~\ref{defn:MBF}) $f = (f_1)$ with $f_1 \colon \B \to \B$.
The repressing edge implies that $f_1$ is a decreasing MBF.
We choose $f_1 = \neg \bI$, where $\bI$ denotes the identity function and  the symbol $\neg$ denotes logical negation, thus $f_1(0) = 1$ and $f_1(1) = 0$.
Pictorially, this function can be represented by Figure~\ref{fig:singleDown}(b), which represents the \textit{Boolean state transition graph} (STG) with two nodes, $0$ and $1$, and \textit{asynchronous update} (see Definition~\ref{def:Boolean_update}).

\begin{figure}[htp!]
\centering
\includegraphics[width=1.0\linewidth]{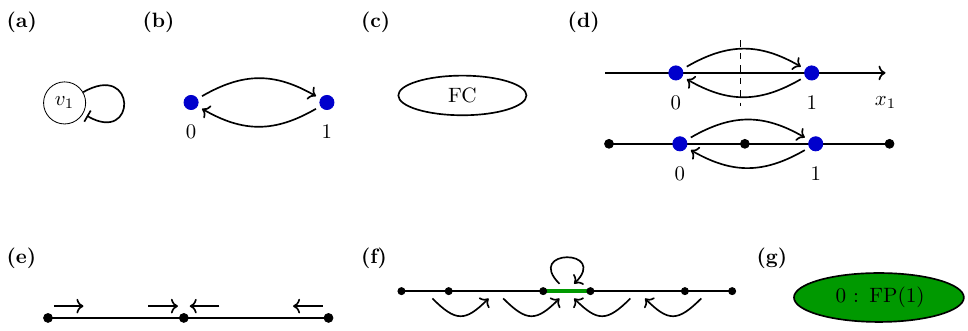}
\caption{(a) A regulatory network $RN$. (b) The asynchronous Boolean state transition graph. The blue dots correspond to the Boolean states $0$ and $1$. (c) The Boolean Morse graph for the Boolean STG consists of one node identified as a Full Cycle (FC). (d) Top: The Boolean STG suggests a division of $[0,\infty)$ into two regions, low and high. Bottom: The Boolean STG defined on a cubical cell complex $\cX$. The three black dots correspond to vertices of the cubical cell complex $\cX$. The black lines are the edges of the cubical cell complex.  (e) The DSGRN wall labeling on the cubical cell complex $\cX$. (f) The blowup cubical complex $\cX_b$ with the DSGRN STG mapping one-dimensional cells to one-dimensional cells. The green edge is the cell in the complex $\cX_b$ corresponding to the node in the DSGRN Morse graph. (g) The DSGRN Morse graph of the DSGRN STG.}
\label{fig:singleDown}
\end{figure}

Figure~\ref{fig:singleDown}(c) shows the \textit{Boolean Morse graph} (see Definition~\ref{def:Boolean_MG}) that characterizes reachability between the maximal recurrent components of the asynchronous Boolean STG.
In this simple example, there is a unique recurrent component consisting of all the states of the STG, and thus the Boolean Morse graph consists of a single node.
In general, a Morse graph has multiple nodes, and the fact that it characterizes reachability between maximal recurrent components implies that it is a partially ordered set.
Nodes that have no edges emanating from them are called \emph{minimal nodes} with respect to this partial order and correspond to stable states.
In this trivial example, this single node is minimal and thus represents a stable state.
Both states of the Boolean STG belong to the unique recurrent component.
Thus, the dynamics is  oscillatory involving all genes
and  we denote it by FC (Full Cycle) (see Definition~\ref{def:Boolean_MG_labels}). 
In summary, the Morse graph of Figure~\ref{fig:singleDown}(c) indicates that the Boolean model characterizes the dynamics of $f$ via a stable oscillation.

For this simple example, the asynchronous Boolean and multilevel models are the same.
Foreshadowing our goal of embedding Boolean and multilevel models into the DSGRN framework, consider Figure~\ref{fig:singleDown}(d) top that provides a conceptual picture of the Boolean states embedded into $[0,\infty)$; low and high concentration values $x_1$  are identified by $0$ and  $1$, respectively. We then superimpose the asynchronous Boolean STG over a \textit{cubical cell complex} (see Definition~\ref{defn:Xcomplex}).
The  cubical cell complex $\cX$ shown in Figure~\ref{fig:singleDown}(d) bottom consists of three zero-dimensional cells (the vertices  denoted by black dots): the leftmost corresponds to $x_1 =0$, the rightmost corresponds to unrealistically large values of $x_1$, and the middle corresponds to an unspecified intermediate value of $x_1$.
The two one-dimensional cells (the edges between the vertices) capture the fact that we hope to be able to relate the DSGRN dynamics with continuous dynamics in $[0,\infty)$.
More precisely, in this identification these edges are associated with connected open intervals in $[0,\infty)$, which are regions that can in principle be identified experimentally.

To proceed with the DSGRN dynamics, there is an important technicality to introduce. The decreasing monotone Boolean function $f_1$ can be represented as a composition $f_1=h_1\circ \beta$ where $h_1= {\bf I}$ is an  increasing monotone Boolean function and called a \textit{DSGRN logic parameter} (see Definition~\ref{def:fi} and Remark~\ref{rem:increasing_MBFs}). Here $\beta=\neg {\bf I}$ is a change-of-variables function (see Equation~\eqref{eq:beta}). The DSGRN dynamics to be presented is associated with $h_1$ rather than with $f_1$.

The construction of the DSGRN model involves several steps.
The first is the identification of the \emph{wall labeling} on $\cX$ (see Definition~\ref{def:wall_labeling}).
As shown in Figure~\ref{fig:singleDown}(e), for this example, it is represented via an arrow pairing each edge with each vertex on its boundary.
By fiat, the vertex associated with $x_1 =0$ and the vertex associated with unrealistically large $x_1$ point towards their respective edges.
In line with the assumption that $f_1(0) =1$, the left edge points towards its right vertex. 
Similarly since $f_1(1) =0$ the right edge points towards its left vertex.
The wall labeling is used to define a state transition graph on $\cX$ (see \cite{rook_field:24}).

As is clear from Figure~\ref{fig:singleDown}(e), the direction of the arrows changes  at the middle vertex.
This suggests that the middle vertex is associated with significant dynamics that should be experimentally observable, directly or indirectly.
In keeping with the discussion above, this suggests that the vertex should be augmented by an edge.
Algorithmically, it is easier to augment all the vertices by edges, leading
 to $\cX_b$, the \textit{blowup cubical cell complex of $\cX$} (see \cite{rook_field:24}). The blowup complex captures both stable and unstable dynamics at the lower-dimensional cells of the cubical cell complex. 

The blowup complex $\cX_b$ has the property that each cell in $\cX$ is represented by an $N$-dimensional cell (in this case one-dimensional cells) in $\cX_b$. Figure~\ref{fig:singleDown}(f) shows the blowup complex $\cX_b$ for the cell complex $\cX$ shown in Figure~\ref{fig:singleDown}(e). There are five one-dimensional cells in $\cX_b$ corresponding to the three zero-dimensional cells (pictorially shown as shorter edges) and the two one-dimensional cells of $\cX$. The DSGRN STG (see \cite{rook_field:24} and Remark~\ref{DSGRN_STG}) is given by the arrows between one-dimensional cells.

The STG on $\cX_b$ in Figure~\ref{fig:singleDown}(f) is determined using the wall labeling on $\cX$.
For this one-dimensional complex we can let monotonicity be our guide.
In particular, to the left of the middle 1-cell the associated arrows of the wall labeling point right, so that each leftward 1-cell maps to its rightward neighbor. 
Similarly, to the right of the middle 1-cell the associated arrows of the wall labeling point left and thus these 1-cells map to the leftward adjacent 1-cell. 
For the middle 1-cell, the associated arrows of the wall labeling both point towards the middle vertex in $\cX$ that is associated to the middle 1-cell in $\cX_b$; the middle 1-cell therefore has no exterior mapping and so we include a self-edge in the DSGRN STG. Self-edges are added to the DSGRN STG whenever there is no permitted transition from a cell to any neighbors.

The STG on $\cX_b$ has a single recurrent set, the single cell in the middle of $\cX_b$ colored green suggesting a fixed point, and thus the \textit{DSGRN Morse graph} (see Definition~\ref{def:DSGRN_MG}) shown in Figure~\ref{fig:singleDown}(g) consists of a single node.
The annotation of the node of the Morse graph takes the form $q: $ FP$(k)$, where $q$ is the index of the Morse node, FP suggesting a fixed point, and $k$ indicates the number of cells in the recurrent set that defines the Morse node (see Definition~\ref{def:DSGRN_MG}).

This simple example shows that the dynamics expressed via the asynchronous Boolean model and the DSGRN model can be fundamentally different.
If one adopts the perspective that ODEs provide an appropriate model for the dynamics associated with the regulatory network  of Figure~\ref{fig:singleDown}(a), then these differences have important implications and demonstrate limitations in purely Boolean modeling. If we assume that $[0,\infty)$ is an appropriate phase space for the dynamics of a single node network, then it is impossible to have oscillatory dynamics as suggested by  Figure~\ref{fig:singleDown}(c).
In contrast the DSGRN Morse graph of Figure~\ref{fig:singleDown}(g) suggests that dynamics converges to a stable fixed point.

\begin{figure}[htp!]
\centering
\includegraphics[width=0.95\linewidth]{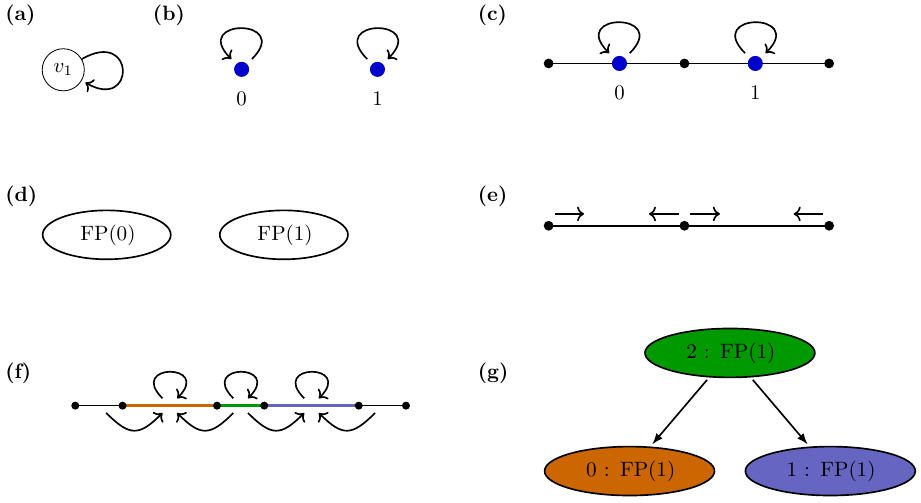}
\caption{(a) A regulatory network $RN$. (b) The asynchronous Boolean state transition graph. The blue dots correspond to the Boolean states $0$ and $1$. (c) The Boolean STG defined on a cubical cell complex $\cX$. The three black dots correspond to vertices of the cubical cell complex $\cX$ ($0$-dimensional cells). The black lines are edges of the cubical cell complex ($1$-dimensional cells). (d) The Boolean Morse graph given by the set of fixed points for Boolean state transition graph. (e) The DSGRN wall labeling on the cubical cell complex $\cX$. (f) The blowup cubical complex $\cX_b$ with the DSGRN STG mapping one-dimensional cells to one-dimensional cells. The green, blue, and orange edges are the cells in the complex $\cX_b$ corresponding to the nodes in the DSGRN Morse graph. (g) The DSGRN Morse graph of the DSGRN STG.}
\label{fig:singleUp}
\end{figure}

We repeat this analysis by considering a one-node network with an activating edge as shown in Figure~\ref{fig:singleUp}(a).
Again, we begin by viewing this as a Boolean network. 
As before, the fact that there is a single node in the network implies that there are two states, $\B = \setof{0, 1}$, that node $v_1$ can assume, and the MBM for the network is given by $f = (f_1)$ with $f_1 \colon \B \to \B$. 
Because the edge in Figure~\ref{fig:singleUp}(a) is activating we  
employ $f_1 = {\bf I}$, i.e., $f_1(0) = 0$ and $f_1(1) = 1$. 
Graphically, this is indicated in Figure~\ref{fig:singleUp}(b), which represents the asynchronous Boolean STG.
Figure~\ref{fig:singleUp}(d) shows the Boolean Morse graph, which indicates that the Boolean model characterizes the dynamics of $f_1$ via two fixed points labeled as FP$(0)$ and FP$(1)$, denoting the states of the fixed points.
Furthermore, the fixed points are minimal and thus represent stable states.

As in the previous example, the cubical complex $\cX$, shown in Figure~\ref{fig:singleUp}(c), consists of three zero-dimensional cells and two one-dimensional cells.
However, the wall labeling is different (see Figure~\ref{fig:singleUp}(e)). 
Again, by fiat the vertex associated with $x_1 =0$ and the vertex associated with unrealistically large $x_1$ point towards their respective edges.
In line with the assumption that $f_1(0) =0$, the left edge points  away from its right vertex. 
Similarly, since $f_1(1) =1$ the right edge points away from its left vertex.

The blowup cubical complex $\cX_b$ and the associated STG is shown in Figure~\ref{fig:singleUp}(f).
Figure~\ref{fig:singleUp}(g) shows the DSGRN Morse graph consisting of three nodes (indexed by 0, 1, and 2) that arise from the self-edges at the one-dimensional cells colored (from left to right) in orange, green, and blue. 
Because the recurrent components consist of single states of the DSGRN STG, we label them as  FP$(1)$.
Recall that $k$ in the DSGRN labels FP$(k)$ indicates the number of cells in the recurrent component and hence it has a fundamentally different meaning than the fixed point states encoded in the Boolean Morse graph (see Definition~\ref{def:Boolean_MG_labels}).
The arrows from the green cell to the orange and blue cells give rise to the ordering of green to blue and orange shown in the DSGRN Morse graph.
The minimality of the blue and orange nodes is interpreted to mean that they represent stable objects, while the green node represents an unstable Morse set that contains  a fixed point.

As in the previous example we see a difference between the Boolean and DSGRN Morse graphs.
In this case the information from the Boolean Morse graph embeds into the DSGRN Morse graph.
The additional information captured by the DSGRN Morse graph allows for consistency with the dynamics of a one-dimensional ODE.
If we assume that $[0,\infty)$ is an appropriate phase space for the dynamics, then it is impossible to have two attractors without a separatrix. The DSGRN Morse graph identifies this separatrix via the green node.

\subsubsection{The toggle switch}
\label{subsec:TS}

We introduce the toggle switch~\cite{toggle}, a well-studied two-node network, to begin to demonstrate that the dynamics of multi-node networks can be organized using products based on single nodes. The network, shown in Figure~\ref{fig:toggle}(a), consists of two nodes with mutually repressive edges, and so the corresponding MBM is $f = (f_1,f_2)$. 

\begin{figure}[htp!]
\centering 
\includegraphics[width=0.8\linewidth]{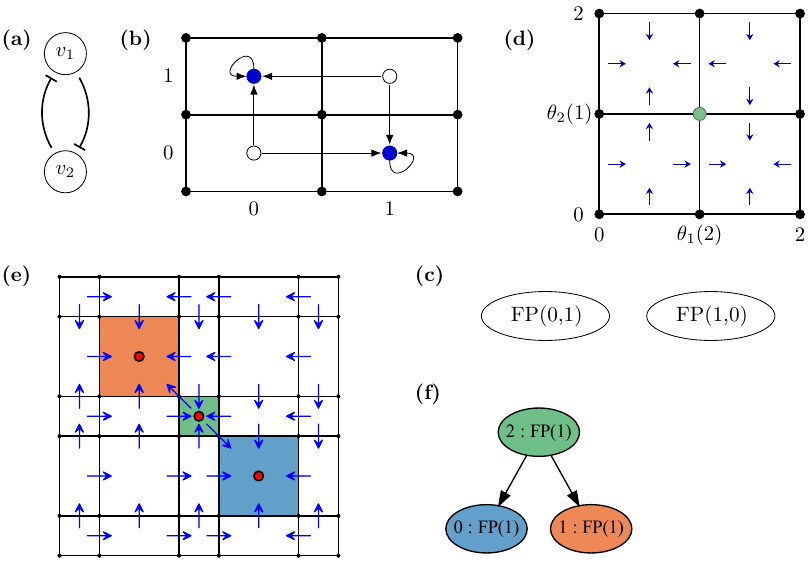}
\caption{(a) A regulatory network $RN$ (the toggle switch). (b) The asynchronous Boolean STG defined on a two-dimensional cubical complex $\cX$. (c) Boolean Morse graph for the asynchronous Boolean STG. (d) The DSGRN wall labeling. (e) The blowup cubical complex $\cX_b$ and the DSGRN STG corresponding to the wall labeling in (d). (f) The Morse graph computed by DSGRN for the STG in (e).
}
\label{fig:toggle}
\end{figure}

As in the examples of Section~\ref{subsec:snn}, each node can assume one of the two states in $\B = \setof{ 0,1 }$ and 
we consider $f_1 = f_2 = \neg {\bf I}$. As before, the repressing edges introduce a change of variables $\beta$, given by Equation \eqref{eq:beta}, that transforms $f_1$ and $f_2$ into increasing MBFs $h_1$ and $h_2$. In this example $h_1 = h_2 = {\bf I}$. The Boolean functions $f_i$ are used to produce the asynchronous Boolean STG represented in Figure~\ref{fig:toggle}(b). The Boolean STG has four nodes $\setof{(0,0),(0,1),(1,0),(1,1)}$ and the asynchronous update edges are shown in Figure~\ref{fig:toggle}(b) within a cubical complex $\cX$. Given that there are two nodes in the network, $v_1$ and $v_2$, the associated concentrations $x_1$ and $x_2$ lie in $[0,\infty)^2$. Thus, $\cX$ is a two-dimensional cubical complex conceptually relating the Boolean states with domains in $[0,\infty)^2$.

Figure~\ref{fig:toggle}(c) shows the Boolean Morse graph. The nodes are the two fixed points that arise from the self-edges at states $(0,1)$ and $(1,0)$ and are denoted by FP(0,1) and FP(1,0), respectively. They are minimal within the Boolean Morse graph and hence represent stable states.

Figure~\ref{fig:toggle}(d) shows the DSGRN wall labeling (Definition~\ref{def:wall_labeling}) corresponding to $h_1$ and $h_2$ on the cubical complex $\cX$.
 Figure~\ref{fig:toggle}(e) shows the blowup cubical complex $\cX_b$, where each zero-, one-, or two-dimensional cell of the cubical complex $\cX$ in Figure~\ref{fig:toggle}(d) is represented by a two-dimensional cell. The arrows indicate the DSGRN STG derived from $h_1$ and $h_2$ and the dots indicate self-edges. 

The DSGRN STG in Figure~\ref{fig:toggle}(e) contains two diagonal arrows, which do not appear in the other examples of this section. The green cell in the blowup cubical complex $\cX_b$ corresponds to the green vertex in the cubical complex $\cX$ (Figure~\ref{fig:toggle}(d)), which is the point where the thresholds of $v_1$ and $v_2$ cross. At this point the two variables regulate each other through a positive feedback loop (two repressing edges). The equilibrium at this point is a saddle and DSGRN adds arrows from it to the two boxes that contain its unstable direction. When $N=2$ and the feedback loop at a threshold crossing is negative, as in the example of Figure~\ref{fig:multi2}, no diagonal arrows are added and the equilibrium corresponding to the crossing is stable. In general, the diagonal arrows depend on the local dynamics at the threshold crossing; see \cite{rook_field:24} for more details.

 Figure~\ref{fig:toggle}(f) shows the DSGRN Morse graph, which consists of three nodes that arise from the self-edges at the two-dimensional cells colored (from upper left to lower right) in orange, green, and blue. Again, because the recurrent components consist of single states of the DSGRN STG we label them as fixed points FP. The arrows from the green cell to the orange and blue cells give rise to the ordering of green to blue and orange shown in the DSGRN Morse graph. The minimality of the blue and orange nodes is interpreted to mean that they represent stable objects. We index the nodes of the Morse graph as nodes $0$, $1$, and $2$ and label each node as $FP(k)$, where $k$ is the number of cells in the complex $\cX_b$ corresponding to the Morse node (in this example $k = 1$ for all nodes).

The  differences between the Boolean Morse graph and the DSGRN Morse graph for the toggle switch are the same as those of the single-node network with an activating edge. Again, from the perspective of ODEs, if we assume that $[0,\infty)^2$ is an appropriate phase space for the dynamics, then it is impossible to have two attractors without a separatrix as is suggested by the Boolean Morse graph. However, the DSGRN Morse graph identifies this separatrix via the green node.

\subsubsection{Multilevel discrete models}
\label{sec:multilevel_models}

To describe the relationship between the set of  monotone Boolean functions and multilevel discrete models we need to consider networks with nodes that have more than one target. A particularly simple example is shown in Figure~\ref{fig:network_boolean}(a). This is a two-node network where node $v_1$ has two targets (itself and node $v_2$) and two inputs, while node $v_2$ has one target (node $v_1$) and one input (also node $v_1$).

The key change from a Boolean model to a multilevel model is the assignment of a distinct ``threshold'' to each edge in a regulatory network $RN$ (see Section~\ref{sec:multilevel_RN}). Each threshold is then associated with a choice of an MBF, permitting different downstream effects from a source node to a target node. For intuition, in a Boolean model, a single threshold of node $v_1$ in Figure~\ref{fig:network_boolean}(a) controls the impact of $v_1$ on both of its outgoing edges, $v_1 \to v_1$ and $v_1 \dashv v_2$, as demonstrated in the Boolean STG shown in Figures~\ref{fig:network_boolean}(b) and \ref{fig:network_boolean}(c). We then ``unfold'' the threshold for $v_1$ into two thresholds, one for each out-edge, plus a choice of threshold ordering, as shown in Figures~\ref{fig:multi1}(a) and \ref{fig:multi2}(a). We initially consider an unfolding in which the MBFs are identical for the two thresholds (Figure~\ref{fig:multi1}), analogous to the Boolean setting, and then proceed to a more general example where the two functions are not identical (Figure~\ref{fig:multi2}).

\begin{figure}[htp!]
\centering
\includegraphics[width=0.85\linewidth]{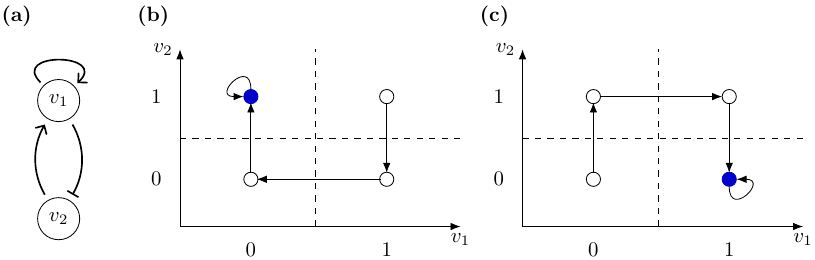}
\caption{(a) A regulatory network $RN$. (b) The asynchronous Boolean STG represented on discrete state space with $f_1 = \wedge$. (c) The asynchronous Boolean STG represented on discrete state space with $f_1= \vee$.}
\label{fig:network_boolean}
\end{figure}

We first consider the dynamics (STG) of the network in Figure~\ref{fig:network_boolean}(a) generated by an asynchronous  monotone Boolean model (MBM) $f \colon \B^2 \to \B^2$, with $f = (f_1, f_2)$. Because $v_1$ has two inputs the update function has a two-dimensional domain: $f_1 \colon \B^2 \to \B$. The fact that both edges from $v_1$ and $v_2$ to $v_1$ are activating implies that $f_1 = h_1$ is an increasing MBF. 
We denote the standard  Boolean functions \textsf{AND} and \textsf{OR} by $\wedge$ and $\vee$, respectively; that is
\[
b_1 \wedge b_2 :=
\begin{cases}
1, & \text{if} ~ b_1 =1 ~\text{and}~ b_2 = 1 \\
0, & \text{otherwise}
\end{cases}
\quad \text{and} \quad
b_1 \vee b_2 :=
\begin{cases}
1, & \text{if} ~ b_1 = 1 ~\text{or}~ b_2 = 1 \\
0, & \text{otherwise}.
\end{cases}
\]
We consider two choices for node $v_1$: $f_1 = h_1 = \wedge$ and $f_1 = h_1 = \vee$. Because node $v_2$ has one input and one target, $f_2 \colon \B \to \B$.  We take $f_2 = \neg {\bf I}$ and $h_2 = {\bf I}$ (after a change of variables) to compute the Boolean and DSGRN STGs as before.

For each choice of $f_1$, we compute the asynchronous Boolean STG given by Definition~\ref{def:Boolean_update}. When the choice of MBF for node $v_1$ is $f_1 = \wedge$, then the resulting asynchronous Boolean STG is the one in Figure~\ref{fig:network_boolean}(b). When $f_1 = \vee$, the resulting asynchronous Boolean STG is the one in Figure~\ref{fig:network_boolean}(c).

We now ``unfold'' the threshold for $v_1$ into two thresholds, that is, we consider one MBF for each edge in the network. Hence we have the increasing MBFs (see Definition~\ref{def:fi}) $h_1^1$ and $h_1^2$, corresponding to the out-edges of node $v_1$, and $h_2^1 = \mathbf I$ corresponding to the out-edge of node $v_2$. The associated MBM $h = ((h_1^1, h_1^2), h_2^1)$ is called a \textit{DSGRN logic parameter} (see Definition~\ref{def:fi} and Remark~\ref{rem:increasing_MBFs}), where $h_i^j$ denotes an MBF associated to the edge from $v_i$ to $v_j$. We first consider the case where $h_1^1 = h_1^2 = \wedge$ are identical. Since now we have two out-edges for node $v_1$ and an increasing MBF corresponding to each edge, we need to decide the order in which the MBFs $h_1^1$ and $h_1^2$ are turned on. The chosen ordering is called an order parameter (see Definition~\ref{defn:order_parameter}). For this example, we allow $h_1^2$ to be turned on before $h_1^1$; i.e., node $v_1$ affects $v_2$ at a lower concentration than it affects itself. We use the notation $\theta_1(2) < \theta_1(1)$ to denote this ordering, where $\theta_i(j)$ denotes the threshold associated to the edge from $v_i$ to $v_j$. The logic and order parameters together are called a \textit{DSGRN parameter} (see Definition~\ref{def:DSGRN_multilevel_parameter_graph}), and define the multilevel  STG (see Definition~\ref{def:update}) on the state space $\cD = \setof{(0,0),(1,0),(2,0), (0,1), (1,1), (2,1)}$ shown in Figure~\ref{fig:multi1}(a). Figure~\ref{fig:multi1}(b) shows the multilevel Morse graph that consists of a single node. Since the associated recurrent component consists of the single state $(0, 1)$ of the multilevel STG, we denote it by FP(0,1). Figure~\ref{fig:multi1}(c) shows the blowup cubical complex $\cX_b$ and the DSGRN STG. The same fixed point is identified by DSGRN as shown in the DSGRN Morse graph in Figure~\ref{fig:multi1}(d).

\begin{figure}[htp!]
\centering
\includegraphics[width=0.8\linewidth]{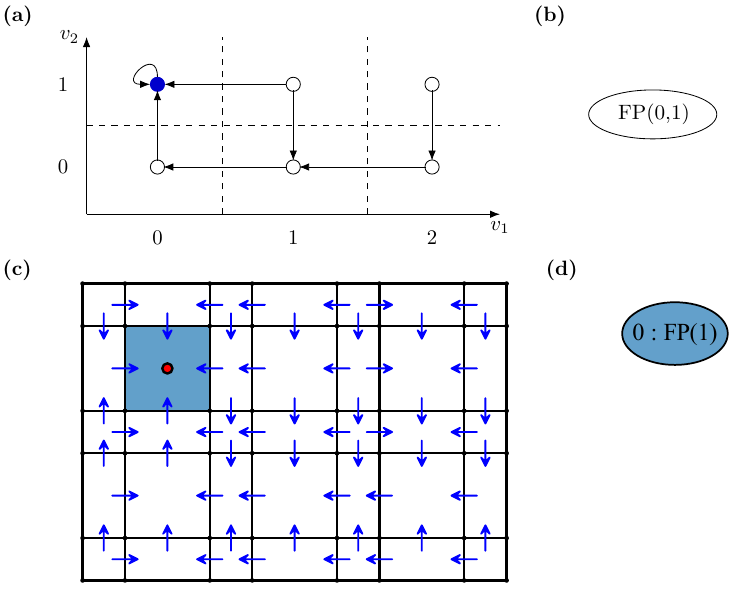}
\caption{(a) The ``unfolded'' multilevel STG for the choice $h_1^1 = h_1^2 = \wedge$ and $\theta_1(2) < \theta_1(1)$ corresponding to the asynchronous Boolean STG in Figure~\ref{fig:network_boolean}(b). (b) Multilevel Morse graph with a single fixed point. (c) The DSGRN STG also detects the same fixed point. (d) The DSGRN Morse graph is isomorphic to the multilevel Morse graph.}
\label{fig:multi1}
\end{figure}

Comparing Figures~\ref{fig:network_boolean}(b) and~\ref{fig:multi1}(a), we see that the same fixed point is recovered by the asynchronous Boolean model and the  multilevel model. 
Based on this example, it may seem that choosing the same MBF for both edges does not fully utilize the freedom of the larger phase space to produce more dynamics. This is actually false, as we will show later that threshold order can nontrivially impact dynamics (Section~\ref{sec:orders}).

We now modify the Boolean representation so that the targets of $v_1$ can be turned on by different Boolean functions. We consider the same network as in Figure~\ref{fig:network_boolean}(a) and choose $h_1^1= \vee$, $h_1^2 = \wedge$, and $\theta_1(1) < \theta_1(2)$ (see Figure~\ref{fig:multi2} for the results). In this scenario, there is no clear comparison to any asynchronous Boolean STG, since in the Boolean setting there is a single Boolean function $f_1$ associated to node $v_1$ and we must choose either $f_1 = \wedge$ or $f_1 = \vee$. For the former choice, see Figure~\ref{fig:network_boolean}(b), and for the latter choice, see Figure~\ref{fig:network_boolean}(c).

\begin{figure}[htp!]
\centering
\includegraphics[width=0.8\linewidth]{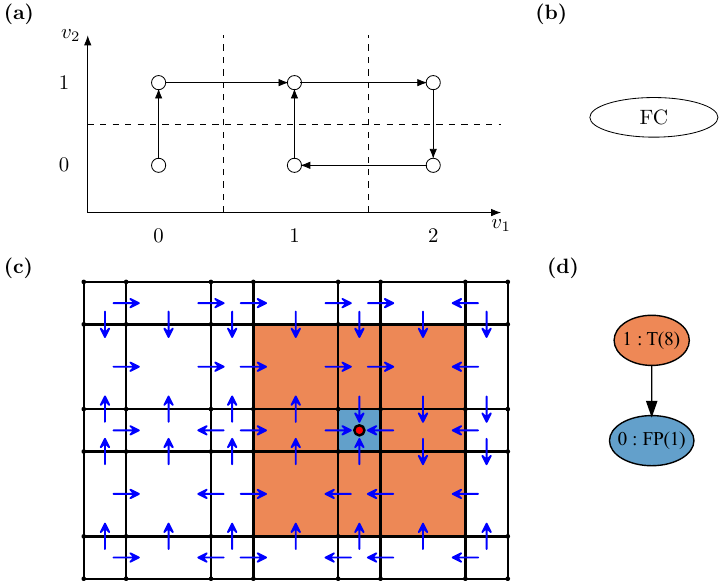}
\caption{(a) The ``unfolded'' multilevel STG for the choice $h_1^1= \vee$ and $h_1^2 = \wedge$ and $\theta_1(1) < \theta_1(2)$ associated to the network shown in Figure~\ref{fig:network_boolean}(a). (b) The multilevel STG identifies a full cycle (FC) through the states $(1,0)$, $(1,1)$, $(2,1)$, and $(2,0)$. (c) The DSGRN STG detects existence of both cyclic behavior and a fixed point. (d) The DSGRN Morse graph for the dynamics computed by DSGRN is compatible with ODE dynamics.}
\label{fig:multi2}
\end{figure}

The choice $h_1^1= \vee$ and $h_1^2 = \wedge$ substantially changes the state transition graph in Figure~\ref{fig:multi2}(a) from the previous choice of $h_1^1 = h_1^2 = \wedge$ in Figure~\ref{fig:multi1}(a). Additionally, the multilevel STG in Figure~\ref{fig:multi2}(a) is substantially different from both of the Boolean STGs in Figures~\ref{fig:network_boolean}(b) and (c).
The multilevel STG now indicates the existence of a cyclic strongly connected component denoted by FC (Full Cycle) shown in Figure~\ref{fig:multi2}(b).

The multilevel update function $g = (g_1, g_2)$ defined by this parameter (see Equation~\eqref{Psi}) can be written in the piecewise form used by modeling tools such as GINsim~\cite{Thomas95,Richard2006} as follows
\[
g_1(d) =
\begin{cases}
2, & \text{if } d_1 \geq 1 \text{ and } d_2 \geq 1, \\
1, & \text{if exactly one of } d_1 \geq 1, \ d_2 \geq 1 \text{ holds}, \\
0, & \text{otherwise},
\end{cases}
\qquad
g_2(d) =
\begin{cases}
1, & \text{if } d_1 < 2, \\
0, & \text{otherwise}.
\end{cases}
\]
The first case of $g_1$ is given by $h_1^2 = \wedge$ and the first two cases together are given by $h_1^1 = \vee$. The condition $d_1 < 2$ in $g_2$ expresses the repression of $v_2$ by $v_1$. The value $g_i(d)$ is the \emph{target value} of node $v_i$ in state $d$, as in the logical formalism of \cite{Thomas95,snoussi_1989}, and the asynchronous dynamics moves $v_i$ one level towards $g_i(d)$.

\begin{rem}
This also explains a feature of Figure~\ref{fig:multi2}(a) that may at first seem counterintuitive. In the state $(1, 0)$ the value of $v_1$ does not change, while in the state $(2, 0)$ it decreases, which might suggest a negative self-regulation of $v_1$. However, $g_1$ is nondecreasing in both $d_1$ and $d_2$. So the self-edge acts as an activation, as required by the network. At level $d_1 = 1$ the self-activation is functional (its threshold is ``crossed'') and it keeps $d_1$ at its target level $1$. Reaching level $2$ requires also the input from $d_2$. Note that this is the discrete analogue of the ODE $\dot{x}_1 = -\gamma_1 x_1 + F_1(x_1, x_2)$ with $F_1$ increasing in both arguments, where the decrease of $x_1$ at high values is due to degradation and not to self-repression.
\end{rem}

The DSGRN STG in Figure~\ref{fig:multi2}(c) shows the existence of both cyclic behavior and a fixed point. However the DSGRN Morse graph (Figure~\ref{fig:multi2}(d)) identifies the cyclic behavior as possibly representing trivial dynamics (see Definition~\ref{def:DSGRN_MG}) and hence labels it as $T(8)$ since it is composed of $8$ cells. The reason the Morse graph in Figure~\ref{fig:multi2}(d) labels the cycling Morse node as trivial is that the state transition graph in Figure~\ref{fig:multi2}(c) does not provide enough information to decide whether trajectories of the continuous system just pass through the orange region and converge to the fixed point in the blue region, or if there is a periodic orbit within the orange region.

As we briefly explain in Section~\ref{sec:Morse_graph}, DSGRN software computes an algebraic-topological invariant called the \textit{Conley index} for each Morse node.  A nontrivial Conley index for a Morse node provides evidence for nontrivial (i.e., recurrent) dynamics within the cells of the cubical complex associated to that node. A trivial Conley index, as was computed for Morse node T(8) in our example, does not provide enough information to confirm the existence of recurrent dynamics.

\subsubsection{Threshold order affects network dynamics}
\label{sec:orders}

It is perhaps not surprising that having more states in the state transition graph of the multilevel discrete map than in the Boolean map and the freedom of having separate Boolean functions associated to different edges produces richer dynamics. This is illustrated by comparing Figure~\ref{fig:network_boolean}, which shows dynamics of monotone Boolean functions, with Figure~\ref{fig:multi2} that shows a more general scenario.

We now consider the network in Figure~\ref{fig:22network}(a) and compare the asynchronous Boolean and the DSGRN state transition graphs when we change the order parameters. For the asynchronous Boolean STG we consider the following MBFs: $f_1(b_1, b_2) = b_1 \wedge b_2$ and $f_2(b_1, b_2) = \neg b_1 \wedge b_2$. The state transition graph for this Boolean system is in Figure~\ref{fig:22network}(b).

\begin{figure}[htp!]
\centering
\includegraphics[width=0.53\linewidth]{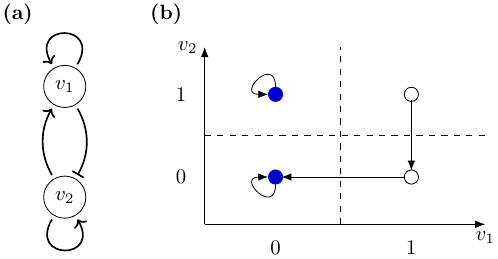}
\caption{(a) Regulatory network $RN$. (b) The asynchronous Boolean STG represented on discrete state space with $f_1$ and $f_2$.}
\label{fig:22network}
\end{figure}

To compute the DSGRN state transition graph we apply the change of variables $\beta$ given by \eqref{eq:beta} which gives $h_1 = h_2 = \wedge$ and consider the logic parameter $h = ((h_1, h_1), (h_2, h_2))$. In other words, we choose matching MBFs for both outgoing edges of $v_1$ and for both outgoing edges of $v_2$, as we did for the example in Figure~\ref{fig:multi1}.

Since both nodes in the network in Figure~\ref{fig:22network}(a) have two targets, we have four ``unfolded'' thresholds to consider: $\{ \theta_1(1), \theta_1(2) \}$ for node $v_1$ and $\{ \theta_2(1), \theta_2(2) \}$ for node $v_2$. There are two ways to order each of these sets and hence there are four possible order parameters:
\begin{equation}
\label{order:4p}
\begin{aligned}
(\theta_1(1) < \theta_1(2), \theta_2(1) < \theta_2(2)) \\
(\theta_1(1) < \theta_1(2), \theta_2(2) < \theta_2(1)) \\
(\theta_1(2) < \theta_1(1), \theta_2(1) < \theta_2(2)) \\
(\theta_1(2) < \theta_1(1), \theta_2(2) < \theta_2(1))
\end{aligned}
\end{equation}

The DSGRN state transition graphs and corresponding Morse graphs for \emph{the same} logic parameter $h$ and the four different order parameters given in (\ref{order:4p}) are shown in Figure~\ref{fig:order_flow}. DSGRN STG shows different dynamics for different order parameters;  these dynamics are all refinements of the bistable Boolean dynamics shown in Figure~\ref{fig:22network}(b). The blue and orange cells and Morse nodes in Figure~\ref{fig:order_flow} correspond to the two fixed points in Figure~\ref{fig:22network}(b). While most of the differences lie in unstable Morse sets, notice that Figure~\ref{fig:order_flow}(b) shows tristable rather than bistable dynamics.

\begin{figure}
\centering
\includegraphics[width=1.0\linewidth]{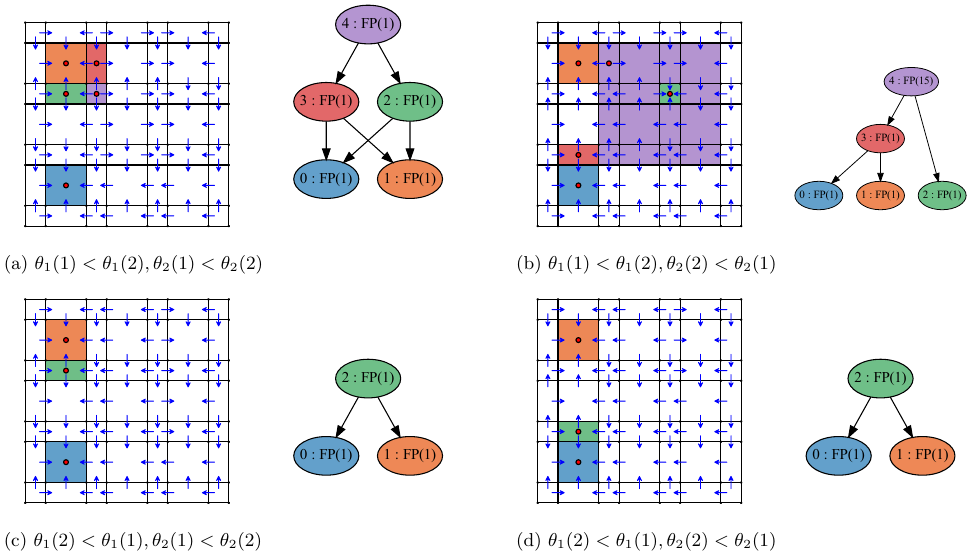}
\caption{DSGRN state transition graphs and corresponding Morse graphs for the four different order parameters given in (\ref{order:4p}) with $h_1 = h_2 = \wedge$.}
\label{fig:order_flow}
\end{figure}

\subsubsection{DSGRN parameter graph and Boolean lattices}
\label{subsec:DSGRNPGBLattices}
Recall that the combination of a logic parameter and an order parameter is called a DSGRN parameter.
The collection of all DSGRN parameters for a given network is organized in the DSGRN parameter graph $\PG$ (see Section~\ref{sec:PG}).

We denote by ${\mathbf 0}$ the constant Boolean function $0$ (i.e., ${\mathbf 0}(0) = {\mathbf 0}(1) = 0$) and by ${\mathbf 1}$ the constant Boolean function $1$ (i.e., ${\mathbf 1}(0) = {\mathbf 1}(1) = 1$). The DSGRN parameter graph for the one-dimensional regulatory networks in Figures~\ref{fig:singleDown}(a) and \ref{fig:singleUp}(a) is the same and is shown in Figure~\ref{fig:PG_single_toggle}(a). The reason both of those networks have the same parameter graph is because the edge signs are encoded by the network via the change of variables described earlier and not by the parameters in the parameter graph (see Remark~\ref{rem:increasing_MBFs}). The parameter graph for the toggle switch in Figure~\ref{fig:toggle}(a) is shown in Figure~\ref{fig:PG_single_toggle}(b). Notice that the order parameter is not present in the parameter graphs in Figure~\ref{fig:PG_single_toggle}. That is because the corresponding network has a single out-edge per node and hence the order parameter is trivial. The order parameter is nontrivial for networks having at least one node with at least two out-edges as in Figures~\ref{fig:network_boolean}(a) and~\ref{fig:22network}(a). The parameter graphs for the networks in Figures~\ref{fig:network_boolean}(a) and~\ref{fig:22network}(a) have $120$ and $1{,}600$ nodes, respectively. See Section~\ref{sec:more_examples} for additional examples of  parameter graphs.

\begin{figure}[htp!]
\centering
\includegraphics[width=0.58\linewidth]{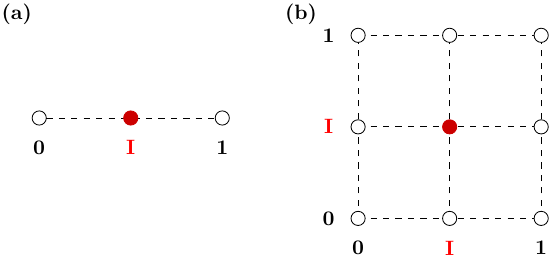}
\caption{(a) Parameter graph for the networks in Figures~\ref{fig:singleDown}(a) and \ref{fig:singleUp}(a). (b) Parameter graph for the network in Figure~\ref{fig:toggle}(a).}
\label{fig:PG_single_toggle}
\end{figure}

We remark that the Boolean functions ${\mathbf 0}$ and ${\mathbf 1}$ are usually not used in Boolean models since their output value does not depend on the input value.  We call these types of MBFs \textit{inessential}, and any monotone Boolean model with at least one inessential MBF is also called inessential. An example of another inessential MBF with two inputs would be the logical function $b_1$:
\begin{center}
\begin{tabular}{cc|c}
$b_1$ & $b_2$ & output  \\
\hline
0 & 0 &  0 \\
1 & 0 & 1 \\
0 & 1 & 0 \\
1 & 1 & 1
\end{tabular}
\end{center}
In the above example, the second input $b_2$ is irrelevant to determining output, and therefore the MBF is classified as inessential. \textit{Essential parameters} are those for which every edge in a network has an influence on the state of its target node (see Section~\ref{sec:MBFs} for rigorous definitions).

Inessential MBFs and MBMs may seem irrelevant to the modeling process. However, they are critical to the definition of neighboring MBFs and neighboring dynamics.
To illustrate this point, consider the 2-node, 4-edge network in Figure~\ref{fig:22network}(a), where both $v_1$ and $v_2$ have two in-edges and two out-edges. Further consider two essential monotone Boolean models and Boolean STGs for that network,
\begin{align*}
f (b_1, b_2) & = (b_1 \wedge b_2, \neg b_1 \wedge b_2) \\
f''(b_1, b_2) & = (b_1 \vee b_2, \neg b_1 \wedge b_2).
\end{align*}
In other words, node $v_1$ expresses the \textsf{AND} function on its out-edges in the first MBM and the \textsf{OR} function in the second MBM, while node $v_2$ expresses the \textsf{AND} function on its out-edges in both MBMs. The Boolean STGs for these two models are given in Figures~\ref{fig:boolean_flow}(a) and (c).

Comparing Figure~\ref{fig:boolean_flow}(a) and (c), we see that although both models produce bistability with shared fixed point $(0,0)$, the second fixed point lies in the state $(0,1)$ in (a) and in $(1,0)$ in (c). The $f$ and $f''$ dynamics are not adjacent in the sense that the second fixed point does not lie in an adjacent cell of the first. However, consider the inessential model, $f'(b_1, b_2) = (b_1, \neg b_1 \wedge b_2)$. The STG for $f'$ is given in Figure~\ref{fig:boolean_flow}(b). We observe that the dynamics of $f$ and $f'$ are adjacent in the sense that their STGs  differ by one fixed point, and similarly for $f'$ and $f''$. 

\begin{figure}[htp!]
\centering
\includegraphics[width=1.0\linewidth]{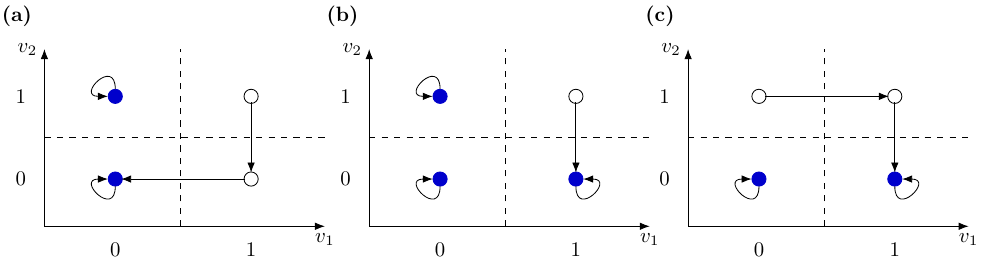}
\caption{(a) The asynchronous Boolean STG of $f$. (b) The asynchronous Boolean STG of $f'$. (c) The asynchronous Boolean STG of ${f}''$.}
\label{fig:boolean_flow}
\end{figure}

This idea of neighboring dynamics is consistent with an idea of neighboring parameters. Recall that the lattice of MBFs requires increasing MBFs only, so parameter adjacency must be in terms of functions $h$. Consider the corresponding MBMs $h$, $h'$, and $h''$ to $f$, $f'$, and $f''$ after the change of variables, $\beta$, in the second component. The MBM $h''$ is the closest \textit{essential} MBM to $h$, in the sense that there is no other essential MBM that can be viewed as a bridge between them. Although $h$ and $h''$ are close essential MBMs, as we noticed the corresponding $f$ and $f''$ do not have adjacent dynamics.

The increasing MBM $h'$ is inessential because the edge from $v_2$ to $v_1$ plays no role in whether $v_1$ is in state $0$ or state $1$. The MBF $b_1$ differs from both functions $\wedge$ and $\vee$ by a value on a single input and therefore the MBM $h'$ is a neighbor of both $h$ and $h''$ in the lattice of MBFs (see Figure~\ref{fig:MBF_PG_multi}(left) in Section~\ref{sec:lattices}). Therefore we can interpret the inessential MBM $h'$ as providing an interpolation between the essential MBMs $h$ and $h''$ in the sense of adjacent asynchronous dynamics.

Note that, although $h'$ is a neighbor of both $h$ and $h''$ in the lattice of MBFs, the corresponding DSGRN parameters are not adjacent. Since each node of the network has $2$ out-edges, the parameters in DSGRN are in fact: $((\wedge, \wedge), (\wedge, \wedge))$, $((X, X), (\wedge, \wedge))$, and $((\vee, \vee), (\wedge, \wedge))$. The MBM lattice corresponding to the logic DSGRN factor graph for node $v_1$ is shown in Figure~\ref{fig:MBF_PG_multi}(right). The logic parameter nodes in the factor graph of node $v_1$ are $(\wedge, \wedge)$, $(X, X)$, and $(\vee, \vee)$ which, as shown in Figure~\ref{fig:MBF_PG_multi}(right), are not adjacent. Additional interpolation of dynamics is required for the greater complexity of DSGRN parameters.

\subsection{Additional examples}
\label{sec:more_examples}

In this section we discuss additional examples that highlight the capabilities of DSGRN. 

\subsubsection{A two node network}

Consider the two-node network in Figure~\ref{fig:22network}(a), where both nodes have positive self-loops, and $v_1 \dashv v_2$ and $v_2 \to v_1$ are the remaining two edges. We computed the Morse graphs for all nodes in the essential parameter graph of this network and classified the nodes into Morse graph isomorphism classes. The resulting parameter graph is presented in Figure~\ref{fig:parameter_graph_two_nodes}, where nodes producing isomorphic Morse graphs have the same color. Nodes that correspond to Boolean parameters (see Definition~\ref{def:DSGRN_Boolean_parameter_graph} to see the connection to monotone Boolean models) are square shaped and the non-Boolean nodes are circles. In Figure~\ref{fig:mg_two_nodes_classes} we present one Morse graph for each isomorphism class.

Notice that all Boolean parameters belong to isomorphism classes $0$, $3$, and $10$. Hence Boolean parameters cannot produce the dynamics in the remaining isomorphism classes. In particular Boolean parameters cannot produce monostability and $4$-stability.

\begin{figure}[!htpb]
\centering
\includegraphics[width=1.0\linewidth]{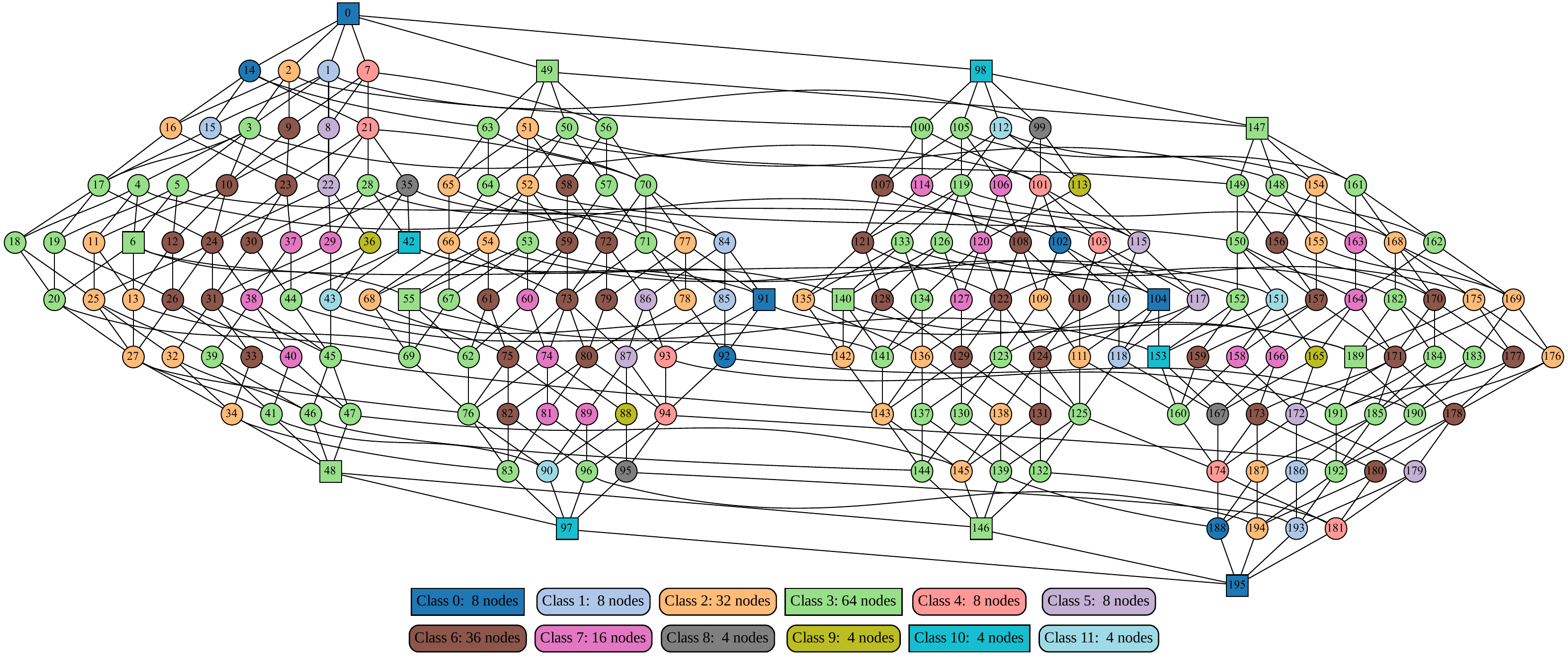}
\caption{Essential parameter graph, with $196$ nodes, for the network in Figure~\ref{fig:22network}(a). Nodes with the same color have isomorphic Morse graphs. The number of nodes in each class is given in the legend. Square nodes are Boolean parameter nodes and round nodes are non-Boolean parameters.}
\label{fig:parameter_graph_two_nodes}
\end{figure}

\begin{figure}
\centering
\includegraphics[width=1.0\linewidth]{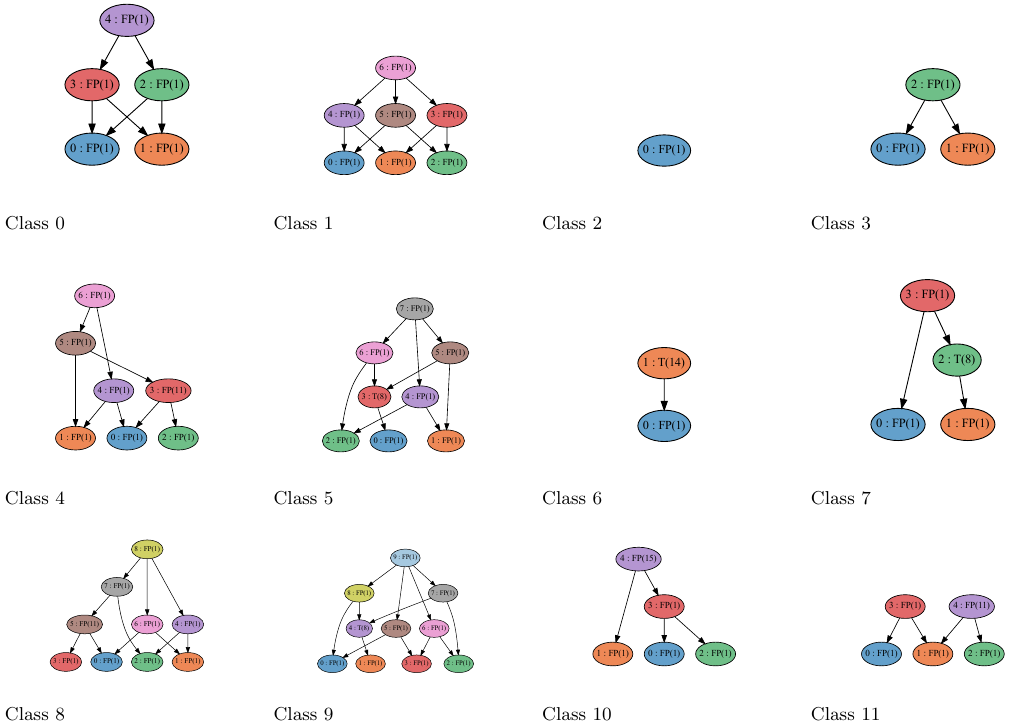}
\caption{Morse graphs for each of the isomorphism classes in Figure~\ref{fig:parameter_graph_two_nodes}.}
\label{fig:mg_two_nodes_classes}
\end{figure}

\subsubsection{A three node network}

Consider the three node network in Figure~\ref{fig:3node_network}. The essential parameter graph (see Definition~\ref{def:DSGRN_essential_parameter_graph}) of this network has $2{,}352$ nodes. To have a complete description of the types of dynamics which can be expressed by this network we computed the Morse graphs for all nodes in the essential parameter graph. We classified each Morse graph according to the number of stable nodes in the Morse graph and whether it contains the Conley index of a stable periodic orbit (the Conley index is a homological invariant that can be used to classify dynamics~\cite{rook_field:24}). For each type of dynamics we counted how many of the parameter nodes are Boolean parameter nodes. The resulting number of parameters for each type of stability encountered is presented in Table~\ref{tab:three_node_dynamics}. As we can see from the results, the Boolean parameters are not expressive enough to capture the periodic behavior and the more complicated dynamics. Figures~\ref{fig:three_node_mg_6_stable} and \ref{fig:three_node_mg_periodic} present two examples of Morse graphs, one exhibiting $6$-stability and one with a Conley index of a stable periodic orbit.

\begin{figure}[!htb]
\centering
\includegraphics[width=0.2\linewidth]{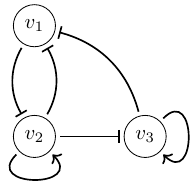}
\caption{A three node network.}
\label{fig:3node_network}
\end{figure}

\begin{table}[!htb]
\centering
\renewcommand{\arraystretch}{1.2}
\textbf{
\begin{tabular}{@{}lll@{}}
\toprule
Type of dynamics & Parameter nodes & Boolean parameters \\ \midrule
monostability & $296$ & $0$ \\
bistability  & $928$ & $48$ \\
$3$-stability & $728$ & $48$ \\
$4$-stability & $330$ & $0$ \\
$5$-stability & $64$ & $0$ \\
$6$-stability & $6$ & $0$ \\
\midrule
Periodic orbit & $96$ & $0$ \\
\bottomrule
\end{tabular}
}
\caption{Number of parameter nodes for each type of dynamics for the network in Figure~\ref{fig:3node_network}. The last column shows how many of the parameter nodes in the second column are Boolean parameter nodes for each type of dynamics.}
\label{tab:three_node_dynamics}
\end{table}

\begin{figure}[!htpb]
\centering
\includegraphics[width=0.9\linewidth]{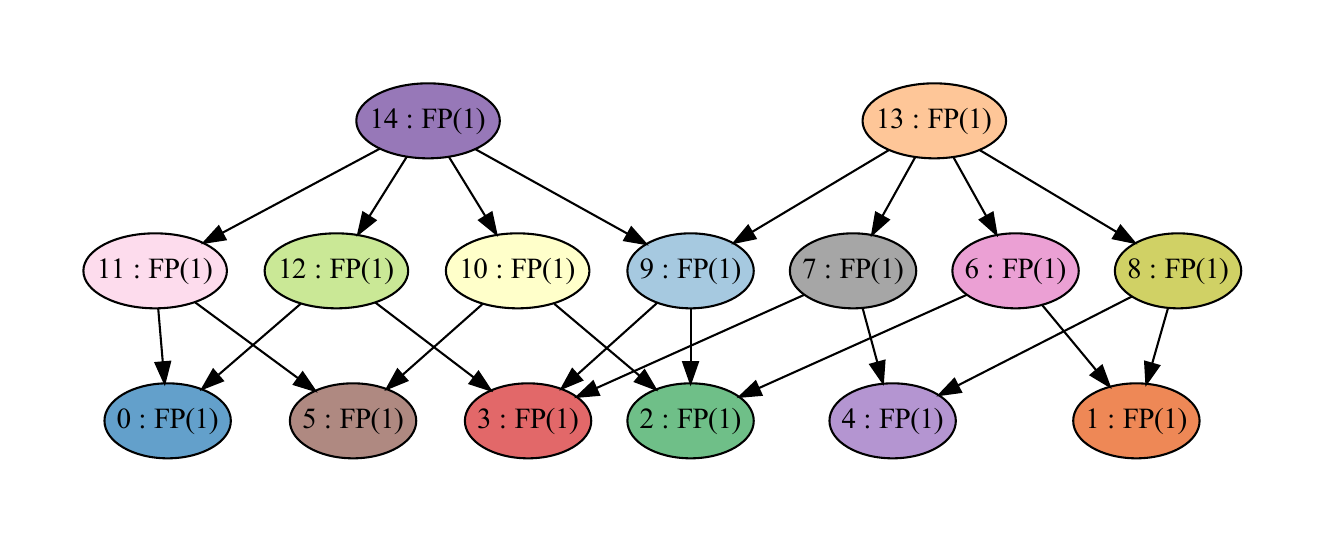}
\caption{Morse graph for parameter node $163$ of the network in Figure~\ref{fig:3node_network} exhibiting $6$-stability ($6$ minimum nodes).}
\label{fig:three_node_mg_6_stable}
\end{figure}

\begin{figure}[!htpb]
\centering
\includegraphics[width=0.8\linewidth]{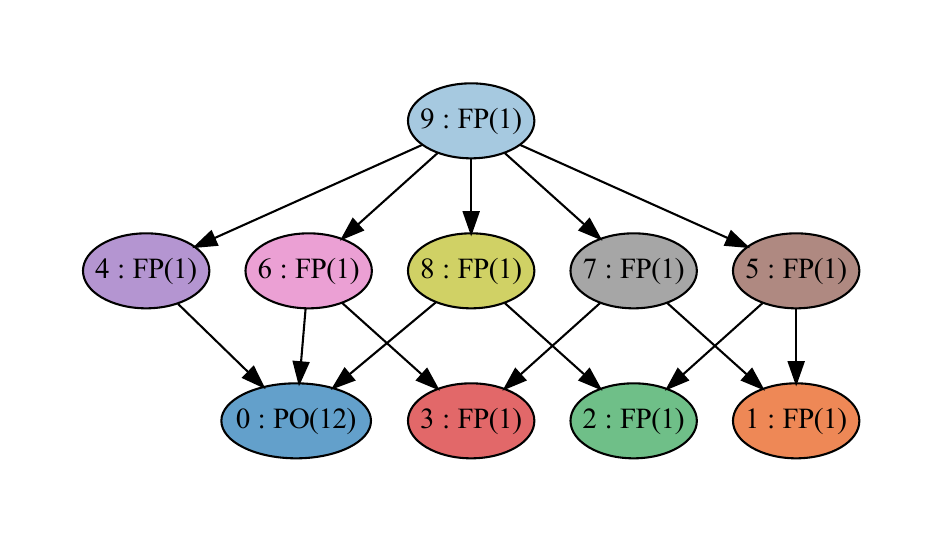}
\caption{Morse graph for parameter node $1251$ of the network in Figure~\ref{fig:3node_network}. Morse node $0$ has the Conley index of a stable periodic orbit.}
\label{fig:three_node_mg_periodic}
\end{figure}

\subsubsection{A five node network}

Consider the five node network in Figure~\ref{fig:5node_network}(Left). The goal of this example is to find parameter nodes producing oscillations. The parameter graph for this network contains $151{,}974{,}144{,}000$ parameter nodes. A careful search of the parameter graph produced parameter nodes exhibiting oscillations. One such parameter node is $56{,}996{,}313{,}179$. The Morse graph for this parameter node is shown in Figure~\ref{fig:5node_network}(Right).

\begin{figure}[!htb]
\centering
\includegraphics[width=0.93\linewidth]{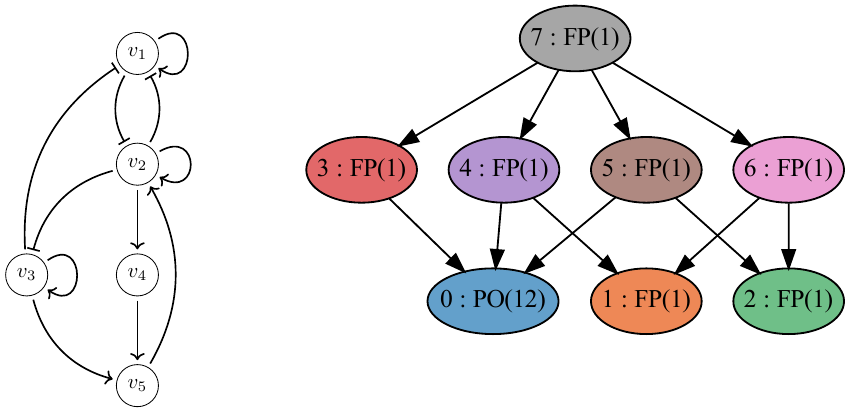}
\caption{(Left) A five node network. (Right) Morse graph for parameter node $56{,}996{,}313{,}179$ of the network. Morse node $0$ has the Conley index of a stable periodic orbit.}
\label{fig:5node_network}
\end{figure}

\subsubsection{The EMT network}
\label{subsec:EMT}

The transition of epithelial cells into a mesenchymal state is an important process implicated in several normal physiological processes, but also  implicated in pathological processes like tumor metastasis~\cite{Kalluri2009}. Therefore the epithelial-to-mesenchymal transition (EMT) and the associated mesenchymal-to-epithelial transition (MET) have been extensively studied both experimentally and theoretically using mathematical models over several decades~\cite{Kalluri2009,hong2015ovol2,andriani2016conversion,bierie2017integrin,yu2013circulating,grosse2015stemness,jolly2016stability,brown2022_EMT,Sha-Nie2019,Jolly2017,Sounders2014,zhang2014tgf,huang2013emt,maclean2018exploring,huang2017interrogating,font2018topography}.

The initial model of direct transition from epithelial to mesenchymal states has been replaced by a more nuanced understanding where cells transition through a series of intermediate cell states (ICS)~\cite{Sha-Nie2019}.
 While in these states, cells can reverse a transition under the influence of new extracellular signals, and then later restart the transition again. Therefore this process resembles navigation through a rugged landscape rather than a linear progression from an E to an M state~\cite{Jolly2017,Sounders2014}.

The ICS have been experimentally characterized in several cancer cell lines~\cite{hong2015ovol2,andriani2016conversion,bierie2017integrin,yu2013circulating,grosse2015stemness,jolly2016stability,brown2022_EMT} and investigated in computational models. Random circuit perturbation (RACIPE) has identified two
different types of intermediate phenotypes when applied to a proposed 22-gene network of
EMT~\cite{huang2017interrogating}, while  numerical simulations of a Boolean network model and the analysis of gene expression data revealed
multiple ICS separating stable epithelial and mesenchymal states~\cite{font2018topography}. 

In spite of these and additional experimental and computational studies, the number of intermediate phenotypes is still a matter of debate~\cite{grigore2016tumor}. Recent
experiments have provided evidence for one, two and three intermediate phenotypes~\cite{zhang2014tgf,huang2013emt,maclean2018exploring} in different tissues.
We apply the DSGRN approach and compare the number of ICS it predicts to the number of ICS that are predicted by asynchronous monotone Boolean models. 

Consider the epithelial-mesenchymal transition (EMT) network in Figure~\ref{fig:EMT_network}.
Its multistability was analyzed with the classical DSGRN parameter graph in~\cite{xin:cummins:gedeon}.
The DSGRN parameter graph for this network has $7{,}702{,}918{,}041{,}600$ nodes and the essential parameter graph has $15{,}371{,}790{,}336$ nodes.
In the multilevel model the variables take the values $\{0,\ldots,m_i\}$, where $m_i$ is the number of out-edges from node $i$. So the levels of $(\textsf{TGFb}, \textsf{miR200}, \textsf{Snail1}, \textsf{Ovol2}, \textsf{Zeb1}, \textsf{miR34a})$ range over
\begin{equation}\label{eq:EMTstates}  
\{0,1\} \times \{0,1,2\} \times \{0,1,2,3,4\} \times \{0,1,2\} \times \{0,1,2,3\} \times \{0,1\}.
\end{equation}
It remains to characterize the epithelial and mesenchymal states in this discrete framework. Following~\cite{xin:cummins:gedeon}, there are biomarkers E-cadherin (Ecad) and Vimentin  that characterize the epithelial and mesenchymal states respectively~\cite{hong2015ovol2}, but do not appear in the network due to the fact that they are pendant nodes: they have no feedback into the rest of the network. Instead of directly tracking Ecad and Vimentin, we track the expression levels of their regulators Zeb1, Snail1 and Ovol2. Since Vimentin is up-regulated by Zeb1 and Snail1 and downregulated by Ovol2, the mesenchymal state is represented as Zeb1 and Snail1 high and Ovol2 low. The opposite pattern with Zeb1 and Snail1 low and Ovol2 high indicates the epithelial state, where Ecad is high. Given~\eqref{eq:EMTstates}, we therefore choose $(*, *, k, 0, 3, *)$ for $k \geq 3$ for the mesenchymal state, and $(*, *, 0, 2, 0, *)$ for the epithelial state. The mesenchymal state is represented by both $k=3$ and $k=4$ to reflect not only the general complexity of EMT landscape~\cite{Jolly2017}, but also the identification of two mesenchymal states in breast cancer cell lines~\cite{brown2022_EMT}.

Figures~\ref{fig:EMT_network_mg_1} and \ref{fig:EMT_network_mg_2} show Morse graphs for different essential parameters exhibiting multistability.
These particular parameter nodes were selected by searching for nodes of the essential parameter graph exhibiting a high degree of multistability.
The cells in the cubical cell complex $\cX$ corresponding to each of the stable nodes on the Morse graphs in Figures~\ref{fig:EMT_network_mg_1} and \ref{fig:EMT_network_mg_2} are given in Table~\ref{tab:EMT_mg_nodes}. All the stable nodes on both Morse graphs correspond to a single cell in $\cX$ of the form $[v, \bone]$ (see Section~\ref{section:wall}). This indicates that they are top-dimensional cells ($6$-dimensional cubes) with minimum vertex $v$. Both Morse graphs exhibit 7-stability with both epithelial and mesenchymal fixed points (see the caption of Table~\ref{tab:EMT_mg_nodes}), the signature of a potential transition between epithelial and mesenchymal states separated by as many as four intermediate states. 

\begin{figure*}[!htb]
\centering
\includegraphics[width=0.25\linewidth]{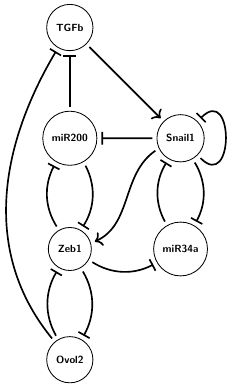}
\caption{Regulatory network for reversible epithelial-to-mesenchymal transition. The nodes of the network are indexed as follows: node 1: \textsf{TGFb}, node 2: \textsf{miR200}, node 3: \textsf{Snail1}, node 4: \textsf{Ovol2}, node 5: \textsf{Zeb1}, and node 6: \textsf{miR34a}.}
\label{fig:EMT_network}
\end{figure*}

\begin{figure}[!htpb]
\centering
\includegraphics[width=1.0\linewidth]{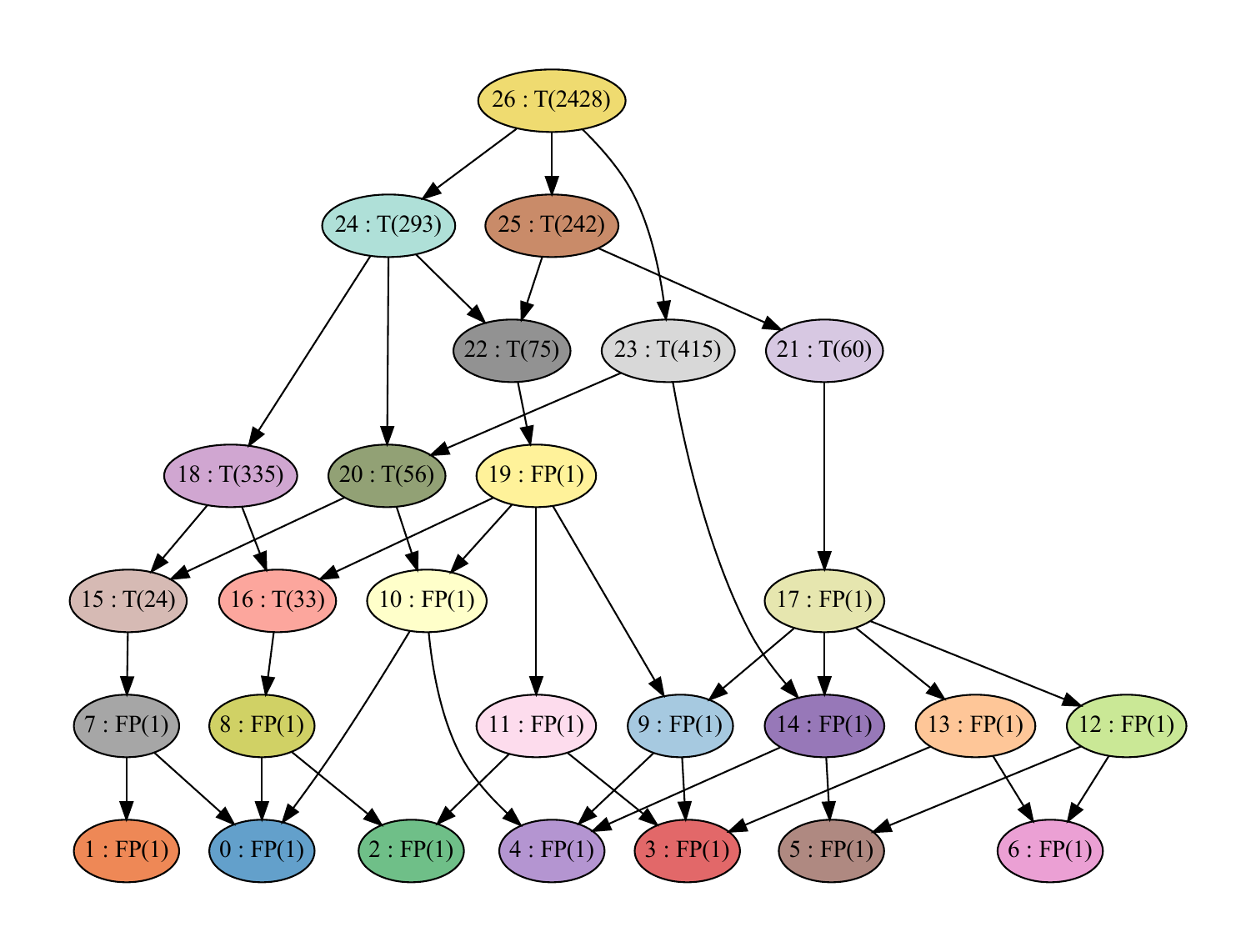}
\caption{Morse graph for parameter node $7{,}314{,}075{,}446$ of the essential parameter graph of the EMT network in Figure~\ref{fig:EMT_network}.}
\label{fig:EMT_network_mg_1}
\end{figure}

\begin{figure}[!htpb]
\centering
\includegraphics[width=1.0\linewidth]{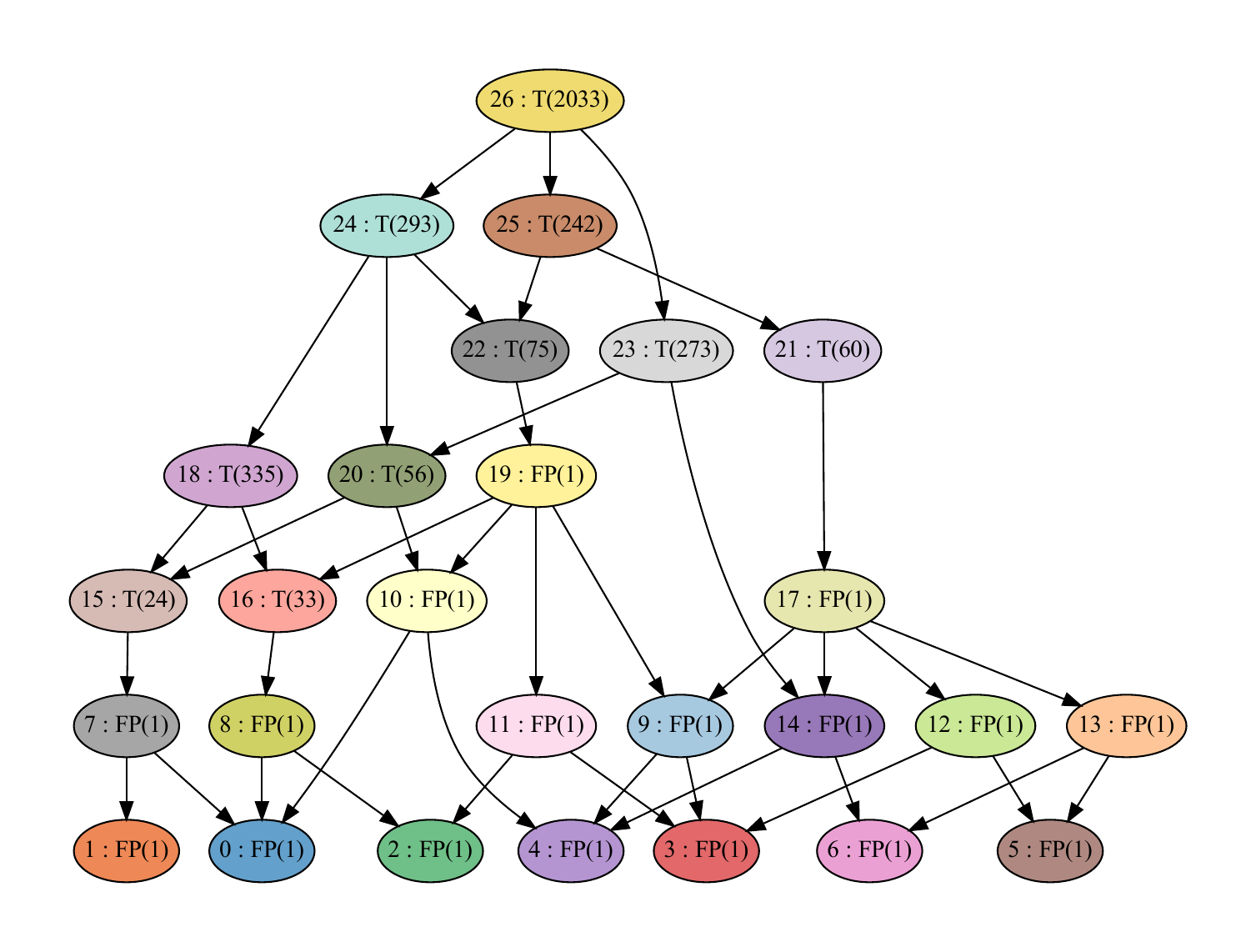}
\caption{Morse graph for parameter node $6{,}033{,}092{,}918$ of the essential parameter graph of the EMT network in Figure~\ref{fig:EMT_network}.}
\label{fig:EMT_network_mg_2}
\end{figure}

\begin{table}[!htbp]
\centering
\renewcommand{\arraystretch}{1.2}
\begin{minipage}{0.45\textwidth}
\centering
\textbf{
\begin{tabular}{@{}ll@{}}
\toprule
Node & Cell in $\cX$ \\ \midrule
$0$ & $[(0, 1, 3, 0, 3, 0), \bone]$ \\
$1$ & $[(1, 0, 4, 0, 3, 0), \bone]$ \\
$2$ & $[(0, 1, 0, 0, 2, 1), \bone]$ \\
$3$ & $[(0, 2, 0, 0, 1, 1), \bone]$ \\
$4$ & $[(0, 2, 3, 0, 1, 0), \bone]$ \\
$5$ & $[(0, 2, 3, 2, 0, 0), \bone]$ \\
$6$ & $[(0, 2, 0, 2, 0, 1), \bone]$ \\
\bottomrule
\end{tabular}
}
\end{minipage}
\begin{minipage}{0.45\textwidth}
\centering
\textbf{
\begin{tabular}{@{}ll@{}}
\toprule
Node & Cell in $\cX$ \\ \midrule
$0$ & $[(0, 1, 3, 0, 3, 0), \bone]$ \\
$1$ & $[(1, 0, 4, 0, 3, 0), \bone]$ \\
$2$ & $[(0, 1, 0, 0, 2, 1), \bone]$ \\
$3$ & $[(0, 2, 0, 0, 1, 1), \bone]$ \\
$4$ & $[(0, 2, 3, 0, 1, 0), \bone]$ \\
$5$ & $[(0, 2, 0, 2, 0, 1), \bone]$ \\
$6$ & $[(0, 2, 3, 2, 0, 0), \bone]$ \\
\bottomrule
\end{tabular}
}
\end{minipage}
\caption{Cells in the cell complex $\cX$ corresponding to each of the stable nodes on the Morse graphs. \textbf{Left:} For the Morse graph in Figure~\ref{fig:EMT_network_mg_1}. \textbf{Right:} For the Morse graph in Figure~\ref{fig:EMT_network_mg_2}. The nodes of the network are indexed in the following order $(\textsf{TGFb}, \textsf{miR200}, \textsf{Snail1}, \textsf{Ovol2}, \textsf{Zeb1}, \textsf{miR34a})$. A Morse graph node corresponds to a Mesenchymal state if the vertex of the corresponding cell in $\cX$ is of the form $(*, *, k, 0, 3, *)$ for $k \geq 3$, and it corresponds to an Epithelial state if the vertex is of the form $(*, *, 0, 2, 0, *)$. Hence we have that for the table on the left nodes $0$ and $1$ are Mesenchymal states and node $6$ is an Epithelial state. For the table on the right nodes $0$ and $1$ are Mesenchymal states and node $5$ is an Epithelial state.}
\label{tab:EMT_mg_nodes}
\end{table}

The DSGRN essential Boolean parameter graph has $373{,}248$ parameter nodes. The nodes in this parameter graph consist of all the $648$ essential monotone Boolean models combined with the order parameters (threshold permutations). As was shown in~\cite{adigwe_characterization_2026} for this network, all stable Morse nodes at any Boolean parameter are fixed points.

To see how typical the multistability of Figures~\ref{fig:EMT_network_mg_1} and~\ref{fig:EMT_network_mg_2} is, and how it depends on the distance from the Boolean parameters, we computed the Morse graphs of three samples of the essential parameter graph: (i)~$10{,}000$ randomly selected  Boolean parameters; (ii)~$1{,}000$ randomly selected  Boolean parameters together with all their neighbors in the essential parameter graph, $17{,}925$ parameters in total; and (iii)~$100$ randomly selected essential Boolean parameters together with all parameters at distance at most two in the essential parameter graph, $17{,}093$ parameters in total. A Boolean parameter has about $17$ neighbors in the essential parameter graph, and exactly $7$ of them are Boolean parameters, which have the same Boolean model and differ from it in the order of two thresholds. Hence the neighbors of Boolean parameters in samples (ii) and (iii) also contain Boolean parameters. We classified the stable Morse nodes, which are all fixed points, as an Epithelial (E) state or Mesenchymal (M) state if their levels are of the form given in the caption of Table~\ref{tab:EMT_mg_nodes}, and as intermediate otherwise.

The results are summarized in Table~\ref{tab:EMT_Boolean_EM_states}. Every Boolean parameter in sample (i) has an epithelial or a mesenchymal state, and $65\%$ of them have both. The number of stable states at a Boolean parameter ranges from $1$ to $4$, and most Boolean parameters have at most one intermediate state: only $67$ ($0.7\%$) have two, and no Boolean parameter has $3$ or more. Since no Boolean model of the network has more than four attractors, the $5$-stable distance $1$ and $2$ neighbors are examples in which the order of the thresholds creates stable states that the Boolean model does not have (a similar situation is seen in Figure~\ref{fig:order_flow}). Within distance one or two of a Boolean parameter, the Morse graphs have up to $5$ stable states, some parameters have three intermediate states ($3$ in sample (ii) and $21$ in sample (iii)), and some have neither an epithelial nor a mesenchymal state. The fraction of parameters with at least two intermediate states increases from $0.7\%$ in sample (i) to $1.6\%$ in sample (ii) and $1.9\%$ in sample (iii). No parameter in the three samples has more than three intermediate states. The $7$-stable parameters whose Morse graphs are depicted in Figures~\ref{fig:EMT_network_mg_1} and \ref{fig:EMT_network_mg_2}, with four intermediate states, are ten steps away from the nearest Boolean parameter.

\begin{table}[!htb]
\centering
\renewcommand{\arraystretch}{1.2}
\begin{tabular}{@{}lrrrrrr@{}}
\toprule
 & \multicolumn{2}{c}{(i) Boolean} & \multicolumn{2}{c}{(ii) Distance $\leq 1$} & \multicolumn{2}{c}{(iii) Distance $\leq 2$} \\
\cmidrule(lr){2-3} \cmidrule(lr){4-5} \cmidrule(l){6-7}
Stable states & Params & Stability & Params & Stability & Params & Stability \\
\midrule
No E and no M state & $0$ & -- & $86$ & $1$--$2$ & $183$ & $1$--$3$ \\
E state only & $1{,}224$ & $1$--$3$ & $2{,}097$ & $1$--$4$ & $2{,}108$ & $1$--$4$ \\
M state only & $2{,}296$ & $1$--$3$ & $4{,}269$ & $1$--$5$ & $3{,}839$ & $1$--$5$ \\
Both E and M states & $6{,}480$ & $2$--$4$ & $11{,}473$ & $2$--$5$ & $10{,}963$ & $2$--$5$ \\
\quad $0$ intermediate states & $4{,}292$ & $2$--$3$ & $7{,}447$ & $2$--$3$ & $6{,}971$ & $2$--$3$ \\
\quad $1$ intermediate state & $2{,}121$ & $3$--$4$ & $3{,}746$ & $3$--$4$ & $3{,}673$ & $3$--$4$ \\
\quad $2$ intermediate states & $67$ & $4$ & $277$ & $4$--$5$ & $298$ & $4$--$5$ \\
\quad $3$ intermediate states & $0$ & -- & $3$ & $5$ & $21$ & $5$ \\
\midrule
Total & $10{,}000$ & $1$--$4$ & $17{,}925$ & $1$--$5$ & $17{,}093$ & $1$--$5$ \\
\bottomrule
\end{tabular}
\caption{\textbf{Stable states at and near Boolean parameters of the EMT network.} Number of parameters (Params) in each class and range (minimum-maximum) of the number of stable Morse nodes (Stability) for three samples of the essential parameter graph: (i)~$10{,}000$ randomly selected Boolean parameters; (ii)~$1{,}000$ randomly selected Boolean parameters and all their neighbors in the essential parameter graph; (iii)~$100$ randomly selected Boolean parameters and all parameters in the essential parameter graph at distance at most two from them. A parameter is classified by whether its Morse graph has an Epithelial (E) or a Mesenchymal (M) stable state, and the parameters with both are further divided by their number of intermediate stable states.}
\label{tab:EMT_Boolean_EM_states}
\end{table}

These show that the number of intermediate states depends on the model that is being used. While the number of intermediate states in Boolean models of the EMT network in Figure~\ref{fig:EMT_network} include one or two intermediate states, multilevel models may support significantly more. Whether these other intermediate states that arise by considering models with varying out-edge activation levels are biologically relevant remains to be seen. However, since there is significant heterogeneity between the tumor cells, some of this diversity may be reflected in the diversity of the activating thresholds and resulting diversity in the ability to activate these edges with the same concentration of the activator/repressor. This supports the view that the EMT process traverses through a complex landscape of intermediate states and this landscape may be different for different cells in a heterogeneous population.

\section{Methods}
\label{sec:math}

In this section we present the necessary mathematical background and results. This section contains definitions and results from earlier work~\cite{Cummins16,Gedeon20,E2,duncan3,rook_field:24} that are used in this paper. The novel results are the contents of Sections~\ref{sec:label_wall_label} and~\ref{sec:multilevel_refinement}, as well as the presentation of monotone Boolean models as an embedding into the DSGRN parameter graph.
Throughout this section, we use the six-node EMT network in Figure~\ref{fig:EMT_network} as a running example, with ordering $(\tgfb,\miRtwo,\snail,\ovol,\zeb,\miRthree)$.

\subsection{Monotone Boolean functions}
\label{sec:MBFs}

We denote by $\bbB := \{ 0, 1 \}$ the totally ordered set  with $0 < 1$. Let $(\bbB^n, \preceq)$ be an $n$-fold product of $\bbB$ with a partial order induced component-wise by $0 < 1$. Depending on the context, the value $0$ represents \textsf{FALSE} or the state \textsf{OFF} and the value $1$ represents \textsf{TRUE} or the state \textsf{ON}. We denote the negation of $b \in \bbB$ by $\neg b$, that is, $\neg 0 = 1$ and $\neg 1 = 0$.

\begin{defn}
\label{def:lattice}
A \emph{Boolean function} is a function $f \colon \bbB^n \to \bbB$. A function $f \colon \bbB^n \to \bbB^n$, $f = (f_1, \ldots, f_n)$, is called a \emph{Boolean model}.
\end{defn}

\begin{defn}
A Boolean function $f \colon \bbB^n \to \bbB$ is \emph{increasing  with respect to input $j$} if for any $b = (b_1, \ldots, b_n) \in \bbB^n$
\[
f(b_1, \ldots, b_{j-1}, 0, b_{j+1}, \ldots, b_n) \leq f(b_1, \ldots, b_{j-1}, 1, b_{j+1}, \ldots, b_n).
\]
It is \emph{strictly increasing  with respect to input $j$} if there is at least one $b \in \bbB^n$ for which the inequality is strict.

A Boolean function $f \colon \bbB^n \to \bbB$ is \emph{decreasing  with respect to input $j$} if for any $b = (b_1, \ldots, b_n) \in \bbB^n$
\[
f(b_1, \ldots, b_{j-1}, 0, b_{j+1}, \ldots, b_n) \geq f(b_1, \ldots, b_{j-1}, 1, b_{j+1}, \ldots, b_n).
\]
It is \emph{strictly decreasing with respect to input $j$} if there is at least one $b \in \bbB^n$ for which the inequality is strict.
\end{defn}

\begin{defn}
\label{defn:MBF}
A Boolean function $f \colon \bbB^n \to \bbB$ is \emph{monotone} if it is increasing or decreasing with respect to each input $j = 1, \ldots, n$. In this case $f$ is called a \emph{monotone Boolean function} (MBF).
\end{defn}

A special type of monotone Boolean function is increasing with respect to all of its inputs.

\begin{defn}
\label{defn:MBF_increasing}
 A Boolean function $f \colon \bbB^n \to \bbB$ is \emph{monotone increasing} if 
it is increasing with respect to all of its inputs. Equivalently,
for all $z, w \in \bbB^n$
\[
z \preceq w \qquad \implies \qquad f(z) \leq f(w).
\]
\end{defn}

\begin{defn}
\label{defn:MBM}
A Boolean model $f \colon \bbB^n \to \bbB^n$ is \emph{monotone} if $f_i$ is monotone for every $i = 1, \ldots, n$. In this case $f$ is called a \textit{monotone Boolean model} (MBM).
\end{defn}

\begin{ex}
\label{ex:EMT_boolean}
Let $b=(b_\tgfb,b_\miRtwo,b_\snail,b_\ovol,b_\zeb,b_\miRthree)\in\bbB^6$ and consider the Boolean functions
\begin{equation}
\label{eq:EMT_boolean}
\begin{aligned}
    f_\tgfb(b)&=\neg b_\miRtwo\wedge\neg b_\ovol,\\
    f_\miRtwo(b)&=\neg b_\snail\wedge\neg b_\zeb,\\
    f_\snail(b)&=(b_\tgfb\wedge \neg b_\miRthree)\vee(b_\tgfb\wedge\neg b_\snail)\vee(\neg b_\miRthree\wedge \neg b_\snail),\\
    f_\ovol(b)&=\neg b_\zeb,\\
    f_\zeb(b)&=\neg b_\miRtwo\wedge b_\snail\wedge\neg b_\ovol,\\
    f_\miRthree(b)&=\neg b_\snail\wedge\neg b_\zeb.
\end{aligned}
\end{equation}
Note that $f_\EMT \colon \bbB^6 \to \bbB^6$ defined by $f_\EMT=(f_\tgfb,f_\miRtwo,f_\snail,f_\ovol,f_\zeb,f_\miRthree)$ is a monotone Boolean model. In particular, $f_\snail$ is strictly increasing with respect to $b_\tgfb$ and strictly decreasing with respect to $b_\miRthree$ and $b_\snail$, so it is an MBF but not an increasing MBF.
The function $\tilde{h}_\snail$ below is strictly increasing with respect to all three, so it is an increasing MBF.
\begin{equation}
\label{eq:h_snail}    
\tilde{h}_\snail(b) = (b_\tgfb \wedge b_\miRthree)\vee(b_\tgfb\wedge b_\snail)\vee(b_\miRthree \wedge b_\snail). 
\end{equation}
\end{ex}

\subsection{Regulatory network}

\begin{defn}
\label{def:RN}
A \textit{regulatory network} $RN = (V, E, \delta)$ is given by
\begin{itemize}
\item A directed graph $D = (V, E)$ with nodes $V = \{ 1, 2, \ldots, N \}$ and  directed edges $E\subset V \times V$;
\item An edge sign function $\delta \colon E \to \{-1, 1\}$. 
\end{itemize}
We denote an edge from node $i$ to node $j$ by $(i, j)$. The edge $(i, j)$ is \textit{activating} if $\delta(i, j) = 1$ and \textit{repressing} if $\delta(i, j) = -1$. Graphically, we denote an edge $(i, j)$ without indicating its sign by
$i \multimap j$, and we denote an activating edge by $i \to j$ and a repressing edge by $i \dashv j$. The \textit{sources} and \textit{targets} of a node $i$ are given by
\[
\Sources(i):= \setdef{k \in V}{(k, i) \in E} \ \text{ and } \ \Targets(i):= \setdef{j \in V}{(i, j) \in E},
\]
respectively. We assume that $\Sources(i) \neq \emptyset$ and $\Targets(i) \neq \emptyset$ for all $i = 1, \ldots, N$.
\end{defn}
For the sake of clarity of exposition, in a slight abuse of notation, we often assign labels $\{ v_1, v_2, \ldots, v_N \}$ to the nodes $V$ and denote a node $i$ and an edge $i \multimap j$ by $v_i$ and $v_i \multimap v_j$, respectively. In this paper we use both notations interchangeably.

\begin{ex}
\label{ex:EMT_network}
Define the EMT regulatory network by $RN_\EMT=(V,E,\delta)$ where the nodes $V$ are given by
\[
    V=\setof{\tgfb,\miRtwo,\snail,\ovol,\zeb,\miRthree},
\]
the edges $E$ are 
\[
    E = \left\{
    \begin{aligned}
    & (\tgfb,\snail), (\miRtwo,\tgfb), (\miRtwo,\zeb),(\miRthree,\snail), \\
    & (\snail,\miRtwo), (\snail,\snail), (\snail,\zeb), (\snail,\miRthree), \\
    & (\ovol,\tgfb), (\ovol,\zeb), (\zeb,\miRtwo), (\zeb,\ovol), (\zeb,\miRthree)
    \end{aligned}
    \right\}
\]
and the edge sign function $\delta : E \to \setof{-1,1}$ is given by 
\[
    \delta(i,j) = \begin{cases}
        1 & \text{if } (i,j) \in \setof{(\tgfb,\snail),(\snail,\zeb)} \\ 
        -1 & \text{otherwise.}
    \end{cases}
\]
The sources and targets of the node $\snail$ are
\[
\Sources(\snail)=\setof{\tgfb,\snail,\miRthree} 
\text{ and } 
\Targets(\snail)=\setof{\miRtwo,\zeb,\snail,\miRthree}.
\]
\end{ex}

Monotone Boolean models are closely linked to regulatory networks.

\begin{defn}
\label{def:regulatory_network}
Given a monotone Boolean model $f \colon \bbB^n \to \bbB^n$, its associated \emph{influence network} is the regulatory network $RN(f) = (V, E, \delta)$ defined by
\begin{enumerate}[(i)]
\item $V = \{1, \ldots, n \}$,
\item $(i, j) \in E$ if and only if $f_j$ is not constant with respect to the input $i$,
\item $\delta \colon E \to \{1, -1\}$ is given by
\[
\delta(i, j) =
\begin{cases}
\;\;\, 1  & \text{if } f_j \text{ is strictly increasing with respect to input } i \\
-1 & \text{if } f_j \text{ is strictly decreasing with respect to input } i
\end{cases}
\]
\end{enumerate}
\end{defn}

\begin{defn}
\label{def:boolean_model_RN}
Let $RN = (V, E, \delta)$ be a regulatory network with $V = \{ 1, 2, \ldots, N \}$. A monotone Boolean model $f \colon \bbB^N \to \bbB^N$ is a \emph{monotone Boolean model} for $RN$ if its influence network is a signed subgraph of $RN$, that is, if it is a subgraph and the sign functions agree on the common edges. If $RN(f)=RN$, then we say that $f$ is an \emph{essential monotone Boolean model} for $RN$.
\end{defn}

\begin{ex}
Consider $f_\EMT \colon \bbB^6\to \bbB^6$ as in Example~\ref{ex:EMT_boolean}.
Note that the component $f_\snail$ is increasing with respect to $b_\tgfb$, decreasing with respect to $b_\miRthree$ and $b_\snail$, and independent of the other coordinates. In particular, 
\begin{equation*}
    f_\snail(b) = \begin{cases} 
        b_\tgfb & \text{if } b_\miRthree=0,\ b_\snail=1, \\ 
        \neg b_\miRthree & \text{if } b_\tgfb=0,\ b_\snail=0, \\ 
        \neg b_\snail & \text{if } b_\tgfb=0,\ b_\miRthree=0.
    \end{cases}
\end{equation*}
Thus, the following edges belong to the influence network of $f$: 
\[
    \tgfb\to\snail,\quad \miRthree\dashv\snail, \quad\snail\dashv\snail. 
\]
Checking each component, one can verify that $RN(f_\EMT)=RN_\EMT$ as in Example~\ref{ex:EMT_network}. 
\end{ex}

Given a monotone Boolean model $f = (f_1, \ldots, f_n)$, the Boolean functions $f_j \colon \bbB^n \to \bbB$ can be written in terms of monotone increasing Boolean functions $h_j \colon \bbB^n \to \bbB$ as a composition
\begin{equation}
\label{eq:f_hbeta}
f_j(b) = h_j(\beta^j(b)),
\end{equation}
where the component $\beta^j_i \colon \bbB \to \bbB$ of the map $\beta^j \colon \bbB^n \to \bbB^n$ is defined by
\begin{equation}
\label{eq:beta}
\beta^j_i(b) =
\begin{cases}
\;\;\, b_i, & \text{if } f_j \text{ is increasing with respect to input } i \\
\neg b_i,   & \text{if } f_j \text{ is decreasing with respect to input } i
\end{cases}
\end{equation}

It follows that any monotone Boolean model $f = (f_1, \ldots, f_n)$ for a regulatory network $RN$ corresponds to an increasing  monotone Boolean model $h = (h_1, \ldots, h_n)$ where the change of variables $\beta^j$ such that $f_j(b) = h_j(\beta^j(b))$ is given by \eqref{eq:beta}.

Therefore from now on we describe the structure of increasing monotone Boolean functions. When considering Boolean models of networks with negative edges we construct the change of variable function $\beta$ in \eqref{eq:beta}.

\begin{ex}
Consider $f_\EMT : \bbB^6\to\bbB^6$ as in Example~\ref{ex:EMT_boolean} and recall that
\begin{align*}
    f_\snail(b) & = (b_\tgfb\wedge \neg b_\miRthree)\vee(b_\tgfb\wedge\neg b_\snail)\vee(\neg b_\miRthree\wedge \neg b_\snail), \\ 
    \tilde{h}_\snail(b) & = (b_\tgfb \wedge b_\miRthree)\vee(b_\tgfb\wedge b_\snail)\vee(b_\miRthree \wedge b_\snail).
\end{align*}
Note that $f_\snail(b)=\tilde{h}_\snail(\beta^\snail(b))$ where $\tilde{h}_\snail$ is a monotone increasing Boolean function and
\[
\beta^\snail(b) = (b_\tgfb,b_\miRtwo,\neg b_\snail,b_\ovol,b_\zeb,\neg b_\miRthree).
\]
\end{ex}

The set of monotone increasing Boolean functions with $k$ inputs is large and its size is the $k$-th Dedekind number \cite{Kleitman1969}, which is known for $k \leq 9$. Since we can represent any monotone Boolean function $f$ in terms of an increasing monotone Boolean function $h$ as in \eqref{eq:f_hbeta}, the Dedekind number is also the number of monotone Boolean functions with any fixed monotonicity with respect to its inputs.

\subsection{Boolean models}

Given a regulatory network $RN$ and a monotone Boolean model $f \colon \bbB^N \to \bbB^N$ we define a multivalued map $\cF \colon \bbB^N \rightrightarrows \bbB^N$, called the \textit{asynchronous update multivalued map}, which is used to compute the \emph{Boolean dynamics} of the network.

\begin{defn}
\label{def:Boolean_update}
The \textit{asynchronous update multivalued map} is the map $\cF \colon \bbB^N \rightrightarrows \bbB^N$ generated by $f$ and defined by:
\begin{itemize}
\item If $f(b) = b$, then $\cF(b) = \{ b \}$.
\item If $f(b) \neq b$, then $\ol{b} \in \cF(b)$ if and only if there exists $i \in V$ such that
\[
\ol{b}_i = f_i(b) \neq b_i \quad \text{and} \quad \ol{b}_j = b_j \mbox{ for } j\neq i.
\]
\end{itemize}
The multivalued map $\cF$ on $\bbB^N$ is also called the \textit{asynchronous Boolean state transition graph}.
\end{defn}

\begin{ex}
    Consider the EMT model in Example~\ref{ex:EMT_boolean}. Note that $f_\snail(0,0,0,0,0,0)=1$. Therefore, the asynchronous update multivalued map $\cF \colon \bbB^6 \mvmap \bbB^6$ satisfies 
    \[
    (0,0,1,0,0,0) \in \cF(0,0,0,0,0,0).
    \]
\end{ex}

\subsection{Finite multilevel models compatible with \texorpdfstring{$RN$}{RN}}
\label{sec:multilevel_RN}

Given a regulatory network $RN = (V, E, \delta)$ with $V = \{ 1, 2, \ldots, N \}$, we and others have proposed \cite{Ironi2011,Cummins16,Gedeon20,Richard2006} a larger class of finite multilevel models compatible with the network structure. For each node $j$ we consider all possible ways in which the collection of its input edges $\{ (i, j) \in E \ | \ i \in \Sources(j) \}$ can ``turn on'' a subset of its output edges $\{ (j, k) \in E \ | \ k \in \Targets(j) \}$.

The precise notion of an edge being turned on is presented in Definition~\ref{def:turn_on}. Recall that we consider only increasing monotone Boolean models for the network. Informally, if an activating edge $v_i \to v_j$ is turned on it presents the input $1$ (\textsf{ON}) to node $v_j$ and hence it is sending an activating signal to $v_j$. If a repressing edge $v_i \dashv v_j$ is turned on it presents the input $0$ (\textsf{OFF}) to node $v_j$ and hence it is sending a repressing signal to $v_j$. If an edge $v_i \multimap v_j$ is not turned on it presents the opposite input to $v_j$.

Note that in a traditional Boolean model~\cite{glass:kaufman:73,Thomas95,Thomas1973} each node $v_i$ has associated to it a single Boolean variable which can be in one of the two states $0$ or $1$, which is determined by a Boolean function of its inputs. The state $1$ then turns on all of the output edges of $v_i$ and the state $0$ does not turn on any of the output edges of $v_i$. Hence the output edges of a node $v_i$ are either all turned on or none of them are turned on.

If we want a discrete network model that represents a continuous process, the pattern of the subsets of the output edges of a node $v_i$ that are turned on must be representable as a response to varying levels of expression of the continuous variable $x_i$, which may represent concentration, associated to node $v_i$. Hence we need a model where the edge $v_i \multimap v_j$ is turned on by the node $v_i$ only when the value of $x_i$ passes a certain level, or threshold, associated with the edge $v_i \multimap v_j$. Therefore we need a discrete network model where, unlike in a traditional Boolean model, a particular input $b \in \bbB^N$ to a node $v_i$ may turn on only some, but not all, of the edges $v_i \multimap v_j$ with $v_j \in \Targets(v_i)$.

In order to describe such a discrete model we need a notion of a ``discrete threshold'' associated to an edge, corresponding to the continuous threshold associated with the continuous variable $x_i$. In principle, there can be more than one discrete threshold associated to an edge $v_i \multimap v_j$. However, the simplest model satisfying these conditions has a single threshold for each edge. The crucial constraint imposed by this view of the process of turning on some, but not all, of the edges is that the discrete thresholds at which individual edges are turned on must be ordered. Hence our notion of a ``discrete threshold'' is simply an ordering of the output edges of each node $v_i$ and is formalized as an \emph{order parameter} defined below.

We next describe the set of \emph{DSGRN parameters} associated with the regulatory network $RN$ \cite{E2,duncan3}, each of which uniquely defines a finite multilevel discrete model. We note that the construction of the DSGRN parameters only depends on the directed graph structure of the regulatory network, the multilevel discrete models do depend on the edge sign structure $\delta$. We follow the exposition in~\cite{E2}.

We first define an order parameter for a node $i$ of a regulatory network. Given a node $i \in V$ let $m_i:= | \Targets(i) |$ denote the number of out-edges of $i$.

\begin{defn}
\label{defn:order_parameter}
An \emph{order parameter} for a node $i \in V$ is an ordering  of the set $\Targets(i)$, that is, it is a bijection $\theta_i \colon \Targets(i) \to \{1, \ldots, m_i \}$.
\end{defn}

Note that an order parameter is just a total ordering of the out-edges of $i$. This total order represents the order in which the out-edges $\{ (i, j) \ | \ j \in \Targets(i) \}$ of $i$ are turned on.

The set of order parameters for $i$ is denoted by $\Theta(i)$. The set of all order parameters is given by $\Theta := \prod_{i=1}^{N} \Theta(i)$.

\begin{ex}
\label{ex:EMT_order_parameter}
    Consider the regulatory network $RN_\EMT$ of Example~\ref{ex:EMT_network}. Since $m_\snail = 4$, an order parameter for $\snail$ is a bijection $\theta_\snail \colon \Targets(\snail)\to\setof{1,2,3,4}$. In Section~\ref{subsec:EMT}, we consider 
    \[
    \theta_\snail(\zeb) = 1,\ \theta_\snail(\miRthree) = 2,\ \theta_\snail(\snail) = 3,\ \theta_\snail(\miRtwo)=4.
    \]
\end{ex}

Next, we define the \textit{logic parameters} for the regulatory network $RN = (V,E,\delta)$.

\begin{defn}
\label{def:fi}
A \textit{logic parameter} for a node $i$ of the regulatory network $RN = (V,E,\delta)$ is a function
\begin{equation}
\label{eq:logical}
h_i \colon \bbB^{|\Sources(i)|} \to \bbB^{|\Targets(i)|}, \quad h_i = (h_i^1, \ldots, h_i^{m_i}),
\end{equation}
whose components $h_i^k \colon \bbB^{|\Sources(i)|} \to \bbB$ are \emph{increasing} monotone Boolean functions satisfying the \textit{ordering conditions}
\begin{equation}
\label{eq:ordering condition}
h_i^1(b) \succeq \cdots \succeq h_i^{m_i}(b)
\end{equation}
for all $b \in \bbB^{|\Sources(i)|}$.
\end{defn}

\begin{defn}
\label{def:ess_logic_parameter}
The logic parameter $h_i$ for a node $i$ is called an \emph{essential logic parameter} if none of the functions $h_i^1, \ldots, h_i^{m_i}$ are constant, and for every input $k \in \Sources(i)$, there exists $j \in \setof{1, \ldots, m_i}$ such that $h_i^j$ is strictly increasing with respect to input $k$ \cite{duncan3, Gedeon24b}.
Otherwise it is called an \emph{inessential logic parameter}.

A logic parameter $h = (h_1, \ldots, h_N)$ is called an \emph{essential logic parameter} if $h_1, \ldots, h_N$ are all essential logic parameters. Otherwise it is called an \emph{inessential logic parameter}.
\end{defn}

The functions $h_i^j$ are used in conjunction with an order parameter $\theta_i$ and their purpose is to describe when an input $b$ turns on the out-edges of $v_i$. More precisely, the function $h_i^{\theta_i(j)}$ describes when an input $b$ turns on the edge $v_i \multimap v_j$ associated with $\theta_i(j)$ as follows:
\begin{itemize}
\item If $h_i^{\theta_i(j)}(b) = 1$, then $b$ turns on the edge $v_i \multimap v_j$;
\item If $h_i^{\theta_i(j)}(b) = 0$, then $b$ does not turn on the edge $v_i \multimap v_j$.
\end{itemize}

The set of all logic parameters for node $i$ is denoted $\cL(i)$. The set of all logic parameters is $\cL := \prod_{i=1}^{N} \cL(i)$.

\begin{ex}
\label{ex:EMT_logic_parameter}
Consider $RN_\EMT$ as in Example~\ref{ex:EMT_network}. A logic parameter for the node $\snail$ is a function
\[
    h_\snail = (h_\snail^1,h_\snail^2,h_\snail^3,h_\snail^4) \colon \bbB^3 \to \bbB^4,
\]
with $(b_\tgfb,b_\snail,b_\miRthree)\in \bbB^3$ that satisfies the ordering conditions. 
In Section~\ref{subsec:EMT}, the following logic parameter is used
\begin{equation}
\label{eq:EMT_logic_parameter}
\begin{aligned}
h_\tgfb^1(b)&=b_\miRtwo\wedge b_\ovol,\\
h_\miRtwo^1(b)&=b_\snail\vee b_\zeb,\\
h_\miRtwo^2(b)&=b_\snail\wedge b_\zeb,\\
h_\snail^1(b)=h_\snail^2(b)=h_\snail^3(b)&=b_\tgfb\vee b_\miRthree,\\
h_\snail^4(b)&=(b_\tgfb\wedge b_\miRthree)\vee(b_\tgfb\wedge b_\snail)\vee(b_\miRthree\wedge b_\snail),\\
h_\ovol^1(b)=h_\ovol^2(b)&=b_\zeb,\\
h_\zeb^1(b)&=b_\ovol,\\
h_\zeb^2(b)&=b_\miRtwo\wedge b_\ovol,\\
h_\zeb^3(b)&=b_\miRtwo\wedge b_\snail\wedge b_\ovol,\\
h_\miRthree^1(b)&=b_\snail\wedge b_\zeb.
\end{aligned}
\end{equation}
Note that $h_\snail^1(b)=h_\snail^2(b)=h_\snail^3(b) \succeq h_\snail^4(b)$ since $h_\snail^4(b)=1$ requires $b_\tgfb=1$ or $b_\miRthree=1$, that is, $h_\snail^i(b)=1$ for $i=1,2,3$. 
\end{ex}

The set of parameters for node $i$ is $\cP(i) := \cL(i) \times \Theta(i)$. The \textit{set of parameters} for the regulatory network is given by the product $\cP := \prod_{i=1}^{N} \cP(i)$. In a slight abuse of notation we identify the set of parameters $\cP$ with the product of logic and order parameters $\cP = \cL \times \Theta$ and denote a parameter $p \in \cP$ by $p = (h, \theta)$.

\begin{defn}
\label{def:essential_parameter}
A parameter $p = (h, \theta)$ is called an \emph{essential parameter} if $h$ is an essential logic parameter. Otherwise it is called an \emph{inessential parameter}. The set of essential parameters for node $i$ is denoted by $\cPE(i)$.
\end{defn}

In the next section we endow the set of parameters $\cP$ with a graph structure by defining an adjacency relation between the elements of $\cP$.

\subsubsection{The Parameter Graph}
\label{sec:PG}

The parameters in $\cP$ are related by an adjacency relationship which we now define.

\begin{defn}
\label{def:adjacency_DSGRN}
Two parameters $(h_i, \theta_i), (\hat{h}_i, \hat{\theta}_i) \in \cP(i)$ for a node $i$ are \textit{adjacent} if exactly one of the following conditions is satisfied.

\begin{enumerate}
\item[(i)] \textit{Logic adjacency}: $\theta_i = \hat\theta_i$ and a single component of the Boolean functions $h_i$ and $\hat h_i$ differ on a single input, i.e., there exist a unique $j$ and a unique $b\in\bbB^{|\Sources(i)|}$ such that $h^j_i(b) \neq \hat{h}^j_i(b)$.

\item[(ii)] \textit{Order adjacency}: $h_i=\hat h_i$ and the values of the order parameters $\theta_i$ and $\hat\theta_i$ are exchanged on a single pair of neighboring entries on which the logic parameters agree, i.e., there is an index $k$ such that $h^k_i= h^{k+1}_i$ (and hence also $\hat h^k_i= \hat h^{k+1}_i$) and  there exist $j_1, j_2 \in \Targets(i)$ such that
\[
\theta_i(j_1) = k, \; \theta_i(j_2) = k+1 \quad \mbox{ and } \quad \hat{\theta}_i(j_2) = k, \; \hat{\theta}_i(j_1) = k+1,
\]
and $\theta_i(j) = \hat{\theta}_i(j)$ for all $j \not\in \{j_1, j_2\}$.
\end{enumerate}
\end{defn}

\begin{defn}
The \textit{factor graph} for node $i$ is the undirected graph $\PG(i) := (\cP(i), \cE(i))$ whose nodes $\cP(i)$ are the parameter nodes for $i$ and whose edges $\cE(i)$ are given by the adjacency relation defined above.
\end{defn}

\begin{defn}
\label{def:DSGRN_multilevel_parameter_graph}
The \textit{DSGRN multi Boolean parameter graph} $\PG := (\cP, \cE)$ of the regulatory network $RN$ is the Cartesian product
\[
\PG := \prod_{i=1}^{N} \PG(i).
\]
That is, given $p, \hat{p} \in \cP$ there is an edge $(p, \hat{p}) \in \cE$ if and only if there is a unique $i \in V$ such that $(p_i, \hat{p}_i) \in \cE(i)$ and $p_j = \hat{p}_j$ for $j \neq i$. We often refer to $\PG$ just as the \textit{DSGRN parameter graph} or simply as the \textit{parameter graph}.
\end{defn}

\begin{defn}
\label{def:DSGRN_essential_parameter_graph}
The \textit{essential factor graph} for node $i$ is the undirected graph $\PGE(i) := (\cPE(i), \cEE(i))$ whose nodes $\cPE(i)$ are the essential parameter nodes for $i$ and whose edges $\cEE(i)$ are given by the adjacency relation defined above.
The \textit{DSGRN multi Boolean essential parameter graph} of the regulatory network $RN$ is the Cartesian product
\[
\PGE := \prod_{i=1}^{N} \PGE(i).
\]
We often refer to $\PGE$ just as the \textit{DSGRN essential parameter graph} or simply as the \textit{essential parameter graph}.
\end{defn}

While we have introduced this definition before \cite{Gedeon18}, we note that this definition of adjacency, induced by the adjacency in the DSGRN parameter graph, is different from that in \cite{MBF}, where adjacency within the subset of  essential monotone Boolean functions is studied. For instance, essential monotone Boolean functions $\wedge$ and $\vee$ are not logically adjacent in our definition since their values differ on two inputs, but since they are the closest essential functions to each other they would be adjacent in the approach of~\cite{MBF}. See also Section~\ref{subsec:DSGRNPGBLattices} for a discussion of Boolean parameter neighbors within the DSGRN parameter graph.

\begin{ex}
\label{ex:EMT_order_adjacency}
Consider the logic parameter $h$ of $RN_\EMT$ as in Example~\ref{ex:EMT_logic_parameter} and the order parameters $\theta,\hat\theta\in \Theta$ with $\theta_i=\hat\theta_i$ for $i\neq\ovol$ and
\[
\begin{aligned}
\theta_\tgfb(\snail)&=1,\\
\theta_\miRtwo(\tgfb)&=1,\ \theta_\miRtwo(\zeb)=2,\\
\theta_\snail(\zeb)&=1,\ \theta_\snail(\miRthree)=2,\ \theta_\snail(\snail)=3,\ \theta_\snail(\miRtwo)=4,\\
\theta_\ovol(\zeb)&=1,\ \theta_\ovol(\tgfb)=2,\\
\theta_\zeb(\ovol)&=1,\ \theta_\zeb(\miRtwo)=2,\ \theta_\zeb(\miRthree)=3,\\
\theta_\miRthree(\snail)&=1,\\
\hat\theta_\ovol(\zeb)&=2,\ \hat\theta_\ovol(\tgfb)=1.
\end{aligned}
\]
Let $p=(h,\theta)$ and $\hat p=(h,\hat\theta)$. We claim that $(p,\hat p)\in\cE$. 
Replacing the values $k=1$, $j_1=\zeb$, and $j_2=\tgfb$ in Definition~\ref{def:adjacency_DSGRN}, we get
\[
h_\ovol^k(b)=h_\ovol^{k+1}(b)=b_\zeb,
\]
and the order parameters satisfy
\[
\theta_\ovol(j_1)=\hat\theta_\ovol(j_2)=k,\quad
\theta_\ovol(j_2)=\hat\theta_\ovol(j_1)=k+1.
\]
Since $\Targets(\ovol)=\setof{j_1,j_2}$, condition~(ii) holds, so $(h_\ovol,\theta_\ovol)$ and $(h_\ovol,\hat\theta_\ovol)$ are order adjacent in $\PG(\ovol)$. Since $p_j=\hat p_j$ for $j\neq\ovol$, it follows that $(p,\hat p)\in\cE$.
\end{ex}

We emphasize again,  that the size and structure of the parameter graph $\PG$ depends only on the number of input and output edges for each $i \in V$, but not on the sign structure of edges given by $\delta$. The sign structure $\delta$ affects only the dynamics at each $p \in \PG$ which we describe next.

\subsubsection{Dynamics}
\label{sec:dynamics}

Given a parameter $p = (h, \theta) \in \PG$, the \emph{finite multilevel dynamics} occurs on the discrete state space
\[
\mathcal{D} := \prod_{i = 1}^N X_i, \qquad X_i:= \{0, 1, \ldots, m_i\}.
\]

\begin{ex}
\label{ex:EMT_state_space}
    Consider $RN_\EMT$ with the node ordering $(\tgfb,\miRtwo,\snail,\ovol,\zeb,\miRthree)$. Since the number of out-edges is $(1,2,4,2,3,1)$, the discrete state space is
    \[
        \mathcal{D} = \setof{0,1}\times\setof{0,1,2}\times\setof{0,1,2,3,4}\times\setof{0,1,2}\times\setof{0,1,2,3}\times\setof{0,1}. 
    \]
    Note that $|\mathcal{D}|=2\cdot3\cdot5\cdot3\cdot4\cdot2=720$ and $d_\snail\in\setof{0,1,2,3,4}$ for every $d\in\cD$.
\end{ex}

We call $d \in \mathcal D$ a \textit{state} of the $RN$, and construct the \textit{finite multilevel update function} $g \colon \mathcal{D} \to \mathcal{D}$ from the information contained in the parameter $(h, \theta) \in \PG$ in three steps.

First, for each state vector $d = (d_1, \ldots, d_N) \in \mathcal{D}$ and every edge $(i, j)$ the function $\mu_i^j \colon X_i \to \bbB$ given by 
\begin{equation}
\label{eq:mu}
\mu_i^j(d_i) :=
\begin{cases}
0, & \text{if } d_i < \theta_i(j) \\
1, & \text{if } d_i \geq \theta_i(j)
\end{cases}
\end{equation}
indicates the relationship of $d$ with respect to the order parameter $\theta_i(j)$ that corresponds to the edge $(i,j)$. Thus the map
\[
\mu \colon \cD \to \Pi_{i=1}^n \bbB^{|\Sources(j)|}
\]
depends only on the directed graph structure $D=(V,E)$ of the network, but not the edge sign structure $\delta$. To implement the effect of $\delta$ we define the map $\tilde{\beta}^j_i \colon \bbB \to \bbB$ given by
\begin{equation}
\label{eq:beta_tilde}
\tilde{\beta}^j_i(b) =
\begin{cases}
\;\;\, b, & \text{if } \delta(i, j) = 1 \\
\neg b,   & \text{if } \delta(i, j) = -1
\end{cases}
\end{equation}
and define the \textit{encoding function}
\begin{equation}
\label{def:B}
B^j \colon \mathcal{D} \to \bbB^{|\Sources(j)|},
\end{equation}
component-wise by
\[
B_i^j := \tilde{\beta}^j_i \circ \mu_i^j.
\]
The composition $B_i^j$ determines the input from node $v_i$ to node $v_j$ for a given $d \in \mathcal{D}$. For an activating edge $v_i \to v_j$, if $d_i$ is below $\theta_i(j)$ then the input from $v_i$ to $v_j$ is $B_i^j(d) = 0$, while if $d_i$ is above $\theta_i(j)$ then the input from $v_i$ to $v_j$ is $B_i^j(d) = 1$. This assignment is reversed if the edge $v_i \dashv v_j$ is repressing.

\begin{rem}
Notice that if $p = (h, \theta) \in \PG$ is a Boolean parameter (see Definition~\ref{def:DSGRN_Boolean_parameter_graph}) and $f$ is a model for the regulatory network $RN$ corresponding to $h$, then the maps $\beta^j_i$ defined by \eqref{eq:beta} and $\tilde{\beta}^j_i$ coincide.
\end{rem}

Note that the codomain of the map $B := (B^1, B^2, \ldots, B^N)$ is $\bbB^{|E|}$ since $E$ is the disjoint union $E = \bigsqcup_{j = 1}^{N}{\{(i, j) \mid i \in \Sources(j)\}}$. Therefore the function $B$ describes assignments of $0$ or $1$ to every edge in the network for any $d \in \mathcal D$. The function $B$ depends on the order parameter $\theta$ and the sign structure function $\delta$, but not on the logic parameter $h$. It is only via the function $B$ that the sign structure $\delta$ and the order parameter $\theta$ affect the dynamics. Note that given two states $c, d \in \cD$ when the edge $(i, j)$ is activating, i.e. $\delta(i, j) = 1$, and $c_i < d_i$ then $B^j_i(c) \leq B^j_i(d)$, and when the edge $(i, j)$ is repressing, i.e. $\delta(i, j) = -1$, then  $c_i < d_i$ implies that $B^j_i(c) \geq B^j_i(d)$. Therefore the function $B^j$ acts like the change of variables function $\beta^j$ for $f_j$ defined in \eqref{eq:beta}. This allows us to only consider increasing monotone Boolean functions in Definition~\ref{def:fi}.

\begin{defn}
\label{def:turn_on}
Given a state $d \in \cD$ we say that the edge $(i, j)$ is \emph{turned on} at the state $d$ if the following conditions are satisfied:
\begin{itemize}
\item If $v_i \to v_j$, then $B_i^j(d) = 1$.
\item If $v_i \dashv v_j$, then $B_i^j(d) = 0$.
\end{itemize}
\end{defn}

\begin{ex}
\label{ex:EMT_encoding}
Consider the regulatory network $RN_\EMT$ and the order parameter $\theta_\snail$ in Example~\ref{ex:EMT_order_parameter}. Since $\theta_\snail(\snail)=3$ and $\delta(\snail,\snail)=-1$, the map $\mu$ and the encoding function $B$ satisfy
\[
\mu_\snail^\snail(d_\snail)=
\begin{cases}
0, & d_\snail<3,\\
1, & d_\snail\geq3,
\end{cases}
\quad\text{and}\quad
B_\snail^\snail(d)=\neg\mu_\snail^\snail(d_\snail).
\]
By Definition~\ref{def:turn_on}, the edge $(\snail,\snail)$ is turned on if and only if $B_\snail^\snail(d)=0$, or equivalently $d_\snail\geq 3$. Similarly, for the edge $(\snail,\zeb)$, $\theta_\snail(\zeb)=1$ and $\delta(\snail,\zeb)=1$ so the edge is turned on if and only if $B_\snail^\zeb(d)=1$ where 
\[
\mu_\snail^\zeb(d_\snail)=
\begin{cases}
0, & d_\snail<1,\\
1, & d_\snail\geq1,
\end{cases}
\quad\text{and}\quad
B_\snail^\zeb(d)=\mu_\snail^\zeb(d_\snail).
\]
\end{ex}

In the second step, we take into consideration the logic parameter $h$ from the parameter $p = (h, \theta) \in \PG$. We apply the logic parameter map $h_i$ (see Definition~\ref{def:fi}) associated to node $i$ to the output of the function $B^i$ and define $g_i \colon \mathcal{D} \to X_i$ by
\begin{equation}
\label{Psi}
g_i(d) := \sum_{j=1}^{m_i} h_i^j(B^i(d)),
\end{equation}
where, by a slight abuse of notation, we are interpreting the Boolean values $0$ and $1$ as integers in the above sum.

\begin{rem}
Note that if $p = (h, \theta) \in \PG$ is a Boolean parameter (see Definition~\ref{def:DSGRN_Boolean_parameter_graph}) and $f$ is a model for the regulatory network $RN$ corresponding to $h$, then
\[
g_i(d) = \sum_{j=1}^{m_i} h_i^j(B^i(d)) = \sum_{j=1}^{m_i} h_i^j(\beta^i(\mu^i(d))) = \sum_{j=1}^{m_i} f_i(\mu^i(d)) = m_i \, f_i(\mu^i(d)),
\] 
where the vectors $\beta^i=(\beta^i_j)_{j\in \Sources(i)}, \mu^i=(\mu^i_j)_{j\in \Sources(i)}$ aggregate inputs to node $i$.

This expression shows that the multilevel update function $g_i(d)$ depends on the monotone Boolean function $f_i$ that both incorporates the directed graph structure $D=(V,E)$ through the monotone increasing function $h$ and the sign structure $\delta$ through function $\beta$. The order parameter $\theta$ only affects  the function $\mu$.
\end{rem}

\begin{rem}
The function $g_i \colon \mathcal{D} \to X_i$ is just counting how many times $h_i^j(B^i(d))$ attains the value $1$. So an alternative definition, without the abuse of notation above, is
\[
g_i(d) := \left| \left\{ j \mid j \in \setof{1, \ldots, m_i} ~\text{and}~ h_i^j(B^i(d)) = 1 \right\} \right|.
\]
\end{rem}

The finite multilevel update function $g = (g_1, g_2,\ldots, g_N)$ is used to define the multilevel state transition graph.

\begin{defn}
\label{def:update}
The \textit{nearest neighbor asynchronous update} (NNU) is the multivalued map $\cG \colon \cD \rightrightarrows \cD$ generated by $g$ and defined by
\begin{itemize}
\item If $g(d) = d$, then $\cG(d) = \{ d \}$.
\item If $g(d) \neq d$, then $\ol{d} \in \cG(d)$ if and only if there exist $i \in V$ and $\eta \in \{-1, 1\}$ such that $\eta g_i(d) > \eta d_i$ and
\[
\ol{d}_i = d_i + \eta, \quad \ol{d}_j = d_j \mbox{ for } j\neq i.
\]
\end{itemize}
The multivalued map $\cG$ on $\mathcal D$ is also called a \textit{state transition graph}.
\end{defn}

\begin{ex}
\label{ex:EMT_multilevel_update}
Consider $p=(h,\theta)$ as in Example~\ref{ex:EMT_order_adjacency}. The multilevel update at $\snail$ is
\[
    g_\snail(d)=\sum_{k=1}^4 h_\snail^k( B^\snail(d) ),
\]
where $B^\snail(d) \in \bbB^3$ has input order $(\tgfb,\snail,\miRthree)$. For the logic in~\eqref{eq:EMT_logic_parameter},
\[
    g_\snail(d) = 3 \left( B_\tgfb^\snail(d)\vee B_\miRthree^\snail(d) \right) +h_\snail^4 \left( B^\snail(d)\right),
\]
with Boolean values interpreted as integers. Thus, $g_\snail(d)\in\setof{0,3,4}$. For the state $d=(1,0,2,0,0,0)$ in the coordinate order of Example~\ref{ex:EMT_state_space},
\[
B^\snail(d)=(1,1,1),\quad h_\snail(B^\snail(d))=(1,1,1,1),\quad g_\snail(d)=4>2=d_\snail.
\]
The state $\ol d=(1,0,3,0,0,0)$ satisfies $\ol d_\snail=d_\snail+1$ and $\ol d_j=d_j$ for $j\neq\snail$, so $\ol d\in\cG(d)$ by Definition~\ref{def:update} with $i=\snail$ and $\eta=1$.
\end{ex}

\begin{rem}
\label{rem:increasing_MBFs}
The asynchronous update multivalued map $\cG \colon \cD \rightrightarrows \cD$ is defined for every DSGRN parameter $p = (h, \theta) \in \PG$, which is given by a collection of \emph{increasing} monotone Boolean functions $h_i \colon \bbB^{|\Sources(i)|} \to \bbB^{|\Targets(i)|}$. The effect of the negative edges of the network $RN$ that give rise to the non-increasing responses of the model are  encoded by the edge sign function $\delta$, which affects the dynamics via the encoding map $B$.
\end{rem}

\subsubsection{Correspondence between MBF and parameters}
\label{sec:MBF_parameters}

We describe how monotone Boolean functions naturally correspond to particular nodes of the parameter graph $\PG$ and hence to a strict subset of the finite multilevel models. Consider a regulatory network $RN = (V, E, \delta)$ with $|V| = N$. A monotone Boolean model $f = (f_1, \ldots, f_N)$ is a finite multilevel model where all edges $(i, j)$ with $j \in \Targets(i)$ are turned on at the same level of the input in $\bbB^N$. This corresponds to the situation where for each $i = 1, \ldots, N$ the functions $h_i^j$, for $j \in \{ 1, \ldots, |\Targets(i)| \}$, are all equal to the same function $h_i$, where $f_i = h_i \circ \beta^i$. Therefore the dynamics of a monotone Boolean model $f = (f_1, \ldots, f_N)$ correspond to
\[
q:= \prod_{i=1}^N |\Targets(i)|!
\]
parameter nodes $(h, \theta)$ where the logic parameter is given by $h_i^j = h_i$ for all $j = 1, \ldots, m_i$, that is,
\begin{equation}
\label{eq:Boolean_par}
h = \left( (h_1, \ldots, h_1), (h_2, \ldots, h_2), \ldots, (h_N, \ldots, h_N) \right).
\end{equation}

\begin{defn}
\label{def:DSGRN_Boolean_parameter_graph}
We call the subset of DSGRN parameters corresponding to monotone Boolean functions \textit{DSGRN Boolean parameters} to distinguish them from the set of all \textit{DSGRN multi Boolean parameters}. We refer to the corresponding parameter graph as the \textit{DSGRN Boolean parameter graph}.
\end{defn}

\begin{ex}
\label{ex:EMT_Boolean_parameters}
Consider $f_\EMT$ as in Example~\ref{ex:EMT_boolean}. For each node $i$, let $\tilde h_i\colon\bbB^{|\Sources(i)|}\to\bbB$ be the increasing MBF associated with $f_i$ via the change of variables in~\eqref{eq:beta}. The logic parameter corresponding to $f_\EMT$ is
\[
\hat h=\left((\tilde{h}_\tgfb),(\tilde{h}_\miRtwo,\tilde{h}_\miRtwo),(\tilde{h}_\snail,\tilde{h}_\snail,\tilde{h}_\snail,\tilde{h}_\snail),(\tilde{h}_\ovol,\tilde{h}_\ovol),(\tilde{h}_\zeb,\tilde{h}_\zeb,\tilde{h}_\zeb),(\tilde{h}_\miRthree)\right).
\]
Note that $\hat h_i^1=\cdots=\hat h_i^{m_i}=\tilde h_i$ for every node $i$, so $(\hat h,\theta)$ is a Boolean parameter for every $\theta\in\Theta$. Since
\[
|\Theta|=1!\,2!\,4!\,2!\,3!\,1!=576,
\]
$f_\EMT$ corresponds to $576$ Boolean parameters, all with logic $\hat h$. In contrast, for $i\in\setof{\miRtwo,\snail,\zeb}$, the components of $h_i$ in Example~\ref{ex:EMT_logic_parameter} are not all identical, so $(h,\theta)$ is not a Boolean parameter for any $\theta\in\Theta$.
\end{ex}

\subsection{Labeling}

We would like to view the nearest neighbor update map $\cG$ generated by the multilevel map $g$ on the set $\cD$ as a representation of a continuous vector field on a compact subset of $\R^N$. To this end, we graphically represent $\cG$ in a way that facilitates the comparison with such a vector field in the form of a \textit{labeling}.

Given $d \in \cD$ and a direction $j \in V$ we assign two labels to the state $d$ in direction $j$ as follows.

\begin{defn}
\label{def:multilevel_labeling}
Let $\Omega = \cD \times V \times \{\pm 1\}$. The map $g$ defines a \emph{labeling} $\alpha \colon \Omega \to \setof{\pm 1}$ by
\[
\alpha(d, j, \lambda) :=
\begin{cases}
\;\;\, 1, & \text{if } g_j(d) > d_j \\
-1,       & \text{if } g_j(d) < d_j \\
-\lambda, & \text{if } g_j(d) = d_j.
\end{cases}
\]
We refer to $\alpha(d, j, -1)$ as the \emph{left label} of $d$ in direction $j$ and to $\alpha(d, j, 1)$ as the \emph{right label} of $d$ in direction $j$.
\end{defn}

The labeling $\alpha$ is equivalent to the nearest neighbor asynchronous update multivalued map $\cG \colon \cD \rightrightarrows \cD$ in the sense that we can build one from the other. The labeling function $\alpha$ is illustrated in Figure~\ref{fig:labeling}. It is used in Section~\ref{sec:label_wall_label} to produce a DSGRN wall labeling (defined in Section~\ref{DSGRN_wall_labeling}).

\begin{figure}
\centering
\includegraphics[width=0.7\linewidth]{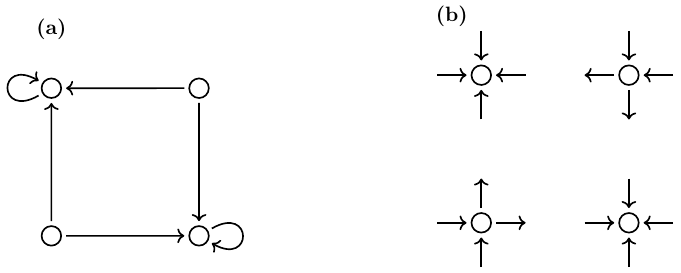}
\caption{(a) State transition graph generated by asynchronous update $\cG \colon \cD \rightrightarrows \cD$ of the monotone Boolean function $h = (h_1, h_2) = ({\bf I}, {\bf I})$ for the network in Figure~\ref{fig:toggle}(a). (b) Labeling of $\cD$ generated by $\cG$. The open circles represent the states of $\cD$ and the arrows represent the labeling.}
\label{fig:labeling}
\end{figure}

\begin{ex}
\label{ex:EMT_labeling}
Consider $p=(h,\theta)$ as in Example~\ref{ex:EMT_order_adjacency} and let $d=(1,0,2,0,0,0)$. Since $g_\snail(d)=4>2=d_\snail$, the left and right labels of $d$ in direction $\snail$ satisfy
\[
\alpha(d,\snail,-1)=\alpha(d,\snail,1)=1.
\]
For $\tgfb$, the encoding functions satisfy $B_\miRtwo^\tgfb(d)=B_\ovol^\tgfb(d)=1$, so $g_\tgfb(d)=h_\tgfb^1(1,1)=1=d_\tgfb$. Thus
\[
\alpha(d,\tgfb,-1)=1\quad\text{and}\quad\alpha(d,\tgfb,1)=-1.
\]
The left and right labels of $d$ in the direction $\snail$ are both $1$, while in the direction $\tgfb$ are $1$ and $-1$, respectively.
\end{ex}

\subsection{Cubical complex}
\label{section:wall}

For $n=1,\ldots, N$, let $K(n)$ be a positive integer, set
\[
I_n := \setof{0,\ldots, K(n)+1} \subset \Z
\]
and define 
\begin{equation}
\label{eq:II}
\I := \prod_{n=1}^N I_n\subset \Z^N.
\end{equation}

For the remainder of this section, let $v \in \I$ and $\bq,\bw\in \setof{0,1}^N$.
Set
\[
\dim(\bw) = \sum_{n=1}^N \bw_n.
\]
Particular elements of interest in $\setof{0,1}^N$ are
\[
{\bf 1}=(1,\ldots, 1), \quad {\bf 0}=(0,\ldots, 0), \quad {\bf 1}^{(n)}=(1,\ldots,1,0,1,\ldots, 1),
\]
where  the $n$-th coordinate is $0$,  and ${\bf 0}^{(n)} = {\bf 1}-{\bf 1}^{(n)}$.

\begin{defn}
\label{defn:Xcomplex}
The set of cells of the $N$-dimensional \emph{cubical cell complex} generated by $\I$ is
defined by
\[
\cX = \setof{ \xi = [v, \bw] \mid v \in \I, \bw\in \setof{0,1}^N, \text{ and } v + \bw \in \I}.
\]
The \emph{dimension} of a cell is given by
\[
\dim([v, \bw]) = \dim(\bw)
\]
and we set $\cX^{(n)} := \setof{\xi\in\cX \mid \dim(\xi) = n}$.
\end{defn}

The partial order $\preceq$, called the \emph{face relation}, indicates which cells are faces of other cells and is defined as
\begin{equation}
\label{eq:is_a_face}
[v, \bw]\preceq[v', \bw']\ \text{if and only if there exists $\bq \in \setof{0,1}^N$ such that $v = v' + \bq$ and $\bw +\bq \leq \bw'$}. 
\end{equation}

\begin{defn}
\label{defn:JeJi}
Let $\cX$ be an $N$-dimensional cubical cell complex. The \emph{essential} and \emph{inessential} directions of $\xi = [v, \bw] \in \cX$ are given by
\[
J_e(\xi) = J_e([v, \bw]) \coloneqq \setof{n\mid \bw_n=1} \text{ and } J_i(\xi) = J_i([v, \bw]) \coloneqq \setof{n\mid \bw_n=0},
\]
respectively.
\end{defn}

Our characterization of combinatorial dynamics on an $N$-dimensional cubical cell complex $\cX$ is based on information organized using top cells $\cX^{(N)}$. Given $\xi \in \cX$ we denote the set of \emph{top cells of $\xi$} by $\Top_\cX(\xi) := \setof{\mu \in \cX^{(N)} \mid \xi \preceq \mu}$.

\begin{defn}
\label{defn:wall}
Given an $N$-dimensional cubical cell complex $\cX$ the set of \emph{top pairs} is given by
\[
\TP(\cX) := \setof{(\xi, \mu) \in \cX \times \cX^{(N)} \mid \xi \preceq \mu}.
\]
Of particular interest is the subcollection where $\xi \in \cX^{(N-1)}$. The set of \emph{walls} of $\cX$ is 
\[
W(\cX) := \setof{(\xi, \mu) \in \cX^{(N-1)} \times \cX^{(N)} \mid \xi \prec \mu}.
\]
A wall $(\xi, \mu)\in W(\cX)$ is called an \emph{$n$-wall} if $J_i(\xi) = \setof{n}$.
\end{defn}

In a slight misuse of language we say that $\xi\in\cX^{(N-1)}$ is a \textit{wall} of $\mu\in\cX^{(N)}$ to indicate that $(\xi, \mu)$ is a wall.

A top cell $\mu = [v, {\bf 1}] \in \cX^{(N)}$ has two $n$-walls 
for each $n$,
\[
\mu^-_n = [v, {\bf 1}^{(n)}]\quad \text{and}\quad \mu^+_n = [v + {\bf 0}^{(n)}, {\bf 1}^{(n)}],
\] 
called the \emph{left} and \emph{right $n$-walls} of $\mu$, respectively.
Two distinct top cells $\mu,\mu'\in\cX^{(N)}$ are \emph{$n$-adjacent} if $\mu$ and $\mu'$ share an \emph{$n$-wall} $\xi$, i.e., $(\xi, \mu)$ and $(\xi, \mu')$ are $n$-walls.

Let $\xi_0 \in\cX$. 
The \emph{walls} of $\xi_0\in\cX$ are
\begin{equation*}
W(\xi_0) := \setof{(\xi, \mu) \in W(\cX) \mid  \xi_0 \preceq \xi}.
\end{equation*}

\begin{ex}
\label{ex:EMT_cubical_complex}
For $RN_\EMT$, let $K(j)=m_j$, $d=(1,0,2,0,0,0)$ and $\ol d=(1,0,3,0,0,0)$. The cells
\[
\mu=[d,\bone],\quad \mu'=[\ol d,\bone],\quad \xi=[\ol d,\bone^{(3)}]
\]
satisfy $\dim\mu=\dim\mu'=6$, $\dim\xi=5$, and $\xi\prec\mu$, $\xi\prec\mu'$. The essential and inessential directions of $\xi$ satisfy
\[
J_e(\xi)=\setof{1,2,4,5,6}\quad\text{and}\quad J_i(\xi)=\setof{3}.
\]
Since $\snail$ is the third coordinate in Example~\ref{ex:EMT_state_space} and $\xi=\mu_3^+=\mu'^-_3$, the top cells $\mu$ and $\mu'$ are $3$-adjacent, and $\xi$ is the right $3$-wall of $\mu$ and the left $3$-wall of $\mu'$.
\end{ex}

\subsection{Wall labeling of a cubical complex}
\label{DSGRN_wall_labeling}

\begin{defn}
\label{def:wall_labeling}
Let $\cX$ be an $N$-dimensional cubical complex.
A function $\omega \colon W(\cX) \to \setof{\pm 1}$ is a \emph{wall labeling} if for each $\sigma \in \cX^{(0)}$ there exists a map
\[
\tilde{o}_{\sigma}\colon \setof{1, \ldots, N} \to \setof{1,\ldots, N},
\]
called a \emph{local inducement map}, satisfying the following two conditions.
\begin{enumerate}
\item[(i)] Let $\mu, \mu'\in \Top_\cX(\sigma)$  be $n$-adjacent.
If $k \neq n$ and $k \neq \tilde{o}_\sigma(n)$, then
\[
\omega(\mu^-_k, \mu) = \omega(\mu'^-_k, \mu') \quad \text{and} \quad \omega(\mu^+_k, \mu) = \omega(\mu'^+_k, \mu').
\]
\item[(ii)] Let $(\xi, \mu)$ and $(\xi, \mu') \in W(\sigma)$ be $n$-walls. If $n \neq \tilde{o}_\sigma(n)$, then
\[
\omega(\xi, \mu) = \omega(\xi, \mu').
\]
\end{enumerate}
\end{defn}

\subsection{Multilevel labeling induces a wall labeling of the cubical complex}
\label{sec:label_wall_label}

Given a regulatory network $RN = (V, E, \delta)$ with $V = \{ 1, 2, \ldots, N \}$, we constructed a finite multilevel model $g \colon \cD \to \cD$ where $\cD$ is the set of states
\[
\mathcal{D} = \prod_{i = 1}^N \{0, 1, \ldots, m_i\},
\]
and $m_i = | \Targets(i) |$ is the number of out-edges of node $i$. For a given parameter $p = (h, \theta) \in \cP$ let $\alpha \colon \cD \times V \times \{ \pm 1\} \to \{\pm 1\}$ be the labeling on $\cD$ given in Definition~\ref{def:multilevel_labeling}.
Let $K(j) := m_j$ and consider the $N$-dimensional cubical complex $\cX$ generated by
\[
\I = \prod_{j=1}^N \setof{0, \ldots, K(j) + 1}
\]
with set of walls $W(\cX)$. Define $\varphi \colon \cD \to \cX^{(N)}$ by
\[
\varphi(v) := [v, {\bf 1}],
\]
and notice that $\varphi$ is a bijection. Let
\begin{equation}
\label{eq:bijection_Boolean_walls}
\psi \colon \cD \times V \times \{ \pm 1\} \to W(\cX)
\end{equation}
be the map defined by
\begin{align*}
& \psi (v, j, -1) := \left([v, {\bf 1}^{(j)}], [v, {\bf 1}] \right) = (\mu_j^-, \mu) \\
& \psi (v, j, 1) := \left( [v + {\bf 0}^{(j)}, {\bf 1}^{(j)}], [v, {\bf 1}] \right) = (\mu_j^+, \mu),
\end{align*}
where $\mu = \varphi(v) = [v, {\bf 1}] \in \cX^{(N)}$.

\begin{prop}
The map $\psi \colon \cD \times V \times \{ \pm 1\} \to W(\cX)$ defined in \eqref{eq:bijection_Boolean_walls} is a bijection.
\end{prop}

\begin{proof}
The map $\psi$ is clearly injective. To show that $\psi$ is surjective, let $(\xi, \mu) \in W(\cX)$ be a wall of $\mu$. It follows that there exists $j \in \{ 1, \ldots, N \}$ such that $J_i(\xi) = \setof{j}$ and hence that $(\xi, \mu)$ is a $j$-wall of $\mu$. Thus $(\xi, \mu)$ is either the left $j$-wall $\mu_j^-$ of $\mu$ or the right $j$-wall $\mu_j^+$ of $\mu$. Let $v = \varphi^{-1}(\mu)$ and notice that in the first case $\psi (v, j, -1) = (\mu_j^-,\mu)$ and in the second case $\psi (v, j, 1) = (\mu_j^+,\mu)$. Therefore $\psi$ is surjective and hence a bijection.
\end{proof}

We define a labeling $\omega \colon W(\cX) \to \{ \pm 1\}$ by
\begin{equation}
\label{eq:Boolean_wall_labeling}
\omega (\xi, \mu) := \alpha(\psi^{-1}(\xi, \mu)).
\end{equation}

\begin{ex}
\label{ex:EMT_wall_labeling}
Consider $d$, $\mu$, and $\xi$ as in Example~\ref{ex:EMT_cubical_complex}. Since $\xi=\mu_3^+$, $\psi$ satisfies $\psi(d,3,1)=(\xi,\mu)$. By Example~\ref{ex:EMT_labeling}, $\alpha(d,3,1)=1$, and hence
\[
\omega(\xi,\mu)=\alpha\bigl(\psi^{-1}(\xi,\mu)\bigr)=\alpha(d,3,1)=1.
\]
\end{ex}

\begin{thm}
\label{thm:Boolen_wall_label}
The map $\omega \colon W(\cX) \to \{ \pm 1\}$ defined by \eqref{eq:Boolean_wall_labeling} is a wall labeling.
\end{thm}

The following lemma is used in the proof of Theorem~\ref{thm:Boolen_wall_label}. We denote by $e_j$ the unit vector in the $j$-th direction.

\begin{lem}
\label{lem:labeling}
Let $d, d' \in \cD$. If $d - d'= e_j$ for some $j \in \setof{1, \ldots, N}$, then there is an index $k = k(j) \in \setof{1, \ldots, N}$ such that
\begin{enumerate}[(i)]
\item For any index $i \neq k$ and $i \neq j$
\[
\alpha(d', i, -1) = \alpha(d, i, -1) \quad \text{ and } \quad \alpha(d', i, 1) = \alpha(d, i, 1).
\]
\item If $j \neq k$, then
\[
\alpha(d', j, 1) = \alpha(d, j, -1).
\]
\end{enumerate}
\end{lem}

\begin{proof}
Let $d,d' \in \cD$ be such that $d - d' = e_j$. So $d_j = d'_j + 1$ and $d_i = d'_i$ for $i \neq j$.
Let $k \in \Targets(j)$ be such that $\theta_j(k) = d_j$. It follows that $d'_j < \theta_j(k)$ and $d_j \geq \theta_j(k)$ and for any $i \neq k$ either $d'_j < \theta_j(i)$ and $d_j < \theta_j(i)$ or $d'_j \geq \theta_j(i)$ and $d_j \geq \theta_j(i)$. Hence for $i \neq k$, by definition, $B^i(d') = B^i(d)$ and thus
\[
g_i(d') = \sum_{j' \in \Targets(i)} h_i^{j'}(B^i(d')) = \sum_{j' \in \Targets(i)} h_i^{j'}(B^i(d)) = g_i(d).
\]
Therefore if $i \neq k$ and $i \neq j$ then $g_i(d') = g_i(d)$ and $d_i' = d_i$ and hence $\alpha(d', i, \lambda) = \alpha(d, i, \lambda)$ which proves (i).

To prove (ii), notice that since $j \neq k$ we have $g_j(d') = g_j(d)$ as above. Hence there are two cases to consider:
\begin{itemize}
\item If $g_j(d') = g_j(d) \leq d'_j < d_j$, then $\alpha(d', j, +1) = -1$ and $\alpha(d, j, -1) = -1$.
\item If $g_j(d') = g_j(d) \geq d_j > d'_j$, then $\alpha(d', j, +1) = 1$ and $\alpha(d, j, -1) = 1$.
\end{itemize}
Therefore $\alpha(d', j, +1) = \alpha(d, j, -1)$, which proves (ii).
\end{proof}

\begin{proof}[Proof of Theorem~\ref{thm:Boolen_wall_label}]
To prove the theorem we need to define a local inducement map for $\omega$. To that end, let $\sigma = [v, {\bf 0}] \in \cX^{(0)}$, with $v = (v_1, \ldots, v_N)$. Define $\tilde{o}_{\sigma} \colon \setof{1, \ldots, N} \to \setof{1,\ldots, N}$ by
\[
\tilde{o}_{\sigma}(j) :=
\begin{cases}
\theta^{-1}_j(v_j), & \text{if } 1 \leq v_j \leq m_j \\
j,                  & \text{otherwise.}
\end{cases}
\]
Notice that $\tilde{o}_{\sigma}(j)$ coincides with the index $k = k(j)$ from Lemma~\ref{lem:labeling}. Notice also that items (i) and (ii) of Lemma~\ref{lem:labeling} match the requirements (i) and (ii) of a wall labeling in Definition~\ref{def:wall_labeling}. Therefore $\tilde{o}_{\sigma}$ defines a local inducement map for $\omega$. This concludes the proof.
\end{proof}

\begin{ex}
\label{ex:EMT_local_inducement}
Consider $p=(h,\theta)$ as in Example~\ref{ex:EMT_order_adjacency} and let $d=(1,0,0,0,2,0)$ and $d'=(1,0,1,0,2,0)$. Let $\mu=[d,{\bf 1}]$, $\mu'=[d',{\bf 1}]$, $\xi=[d',{\bf 1}^{(3)}]$, and $\sigma=[d',{\bf 0}]$. Since $\theta_\snail(\zeb)=d'_\snail=1$, the local inducement map satisfies $\tilde o_\sigma(3)=5$.

Note that the encoding functions satisfy $B^i(d)=B^i(d')$ for $i\in V\setminus\setof{\zeb}$. For $k\in\setof{1,2,4,6}$, we have $d_k=d'_k$ and $g_k(d)=g_k(d')$, so
\[
\omega(\mu_k^-,\mu)=\omega(\mu'^-_k,\mu')\quad\text{and}\quad\omega(\mu_k^+,\mu)=\omega(\mu'^+_k,\mu'),
\]
which verifies condition~(i) of Definition~\ref{def:wall_labeling}.

In the direction $5$, $g_\zeb(d)=2=d_\zeb$ and $g_\zeb(d')=3>d'_\zeb$, so the labels of the right $5$-walls satisfy $\omega(\mu_5^+,\mu)=-1$ and $\omega(\mu'^+_5,\mu')=1$. Finally, $g_\snail(d)=g_\snail(d')=4$ implies
\[
\omega(\xi,\mu)=\omega(\xi,\mu')=1.
\]
Since $3\neq\tilde o_\sigma(3)$, this verifies condition~(ii).
\end{ex}

\subsection{Morse graph}
\label{sec:Morse_graph}

Let $\cZ$ be a finite set and let $\cF \colon \cZ \rightrightarrows \cZ$ be a multivalued map.

\begin{defn}
The directed graph associated with $\cF$ is $G(\cF) = (\cZ,E)$ where the set of vertices is $\cZ$ and the edges are given by 
\[
E = \setdef{ (v,v') \in \cZ \times \cZ }{ v' \in \cF(v) }. 
\] 
\end{defn}

We define an equivalence relation $\sim$ on $\cZ$ by $v \sim v'$ if and only if $v=v'$ or there exist sequences $v_0, \ldots, v_k \in \cZ$ and $v'_0, \ldots, v'_m \in \cZ$ such that $v_0 = v'_m = v$, $v_k = v'_0 = v'$, $v_i \in \cF(v_{i-1})$ for $i=1,\ldots,k$, and $v'_j \in \cF(v'_{j-1})$ for $j=1,\ldots,m$. The equivalence classes are the \emph{strongly connected components} of $G(\cF)$ and are denoted by $\SCC(\cF)$.

\begin{defn}
The \emph{condensation graph} of $G(\cF)$ is the graph defined by $\CG(\cF)=(\SCC(\cF),E_{\SCC})$ where 
\[
E_{\SCC} = \setdef{ (C, C') \in \SCC(\cF) \times \SCC(\cF)}{C \neq C' \text{ and } \exists v \in C, v'\in C' \text{ s.t. } (v,v') \in E }. 
\]

\end{defn}
We define a partial order $\preceq$ on $\SCC(\cF)$ by $C \preceq C'$ if and only if $C = C'$ or if there exist $k\geq 0$ and $C_0,\ldots,C_k \in \SCC(\cF)$ such that $C=C_0$, $C'=C_k$ and $(C_i,C_{i-1}) \in E_{\SCC}$ for $i=1,\ldots,k$.

\begin{defn}
A strongly connected component $C\in \SCC(\cF)$ is a \emph{recurrent component} if the subgraph of $G(\cF)$ induced by $C$ contains at least one edge, that is,
\[
C \text{ is recurrent} \iff E \cap (C \times C) \neq \emptyset.
\]
Let $\RC(\cF)$ denote the set of recurrent components of $G(\cF)$.
\end{defn}

\begin{defn}
The Morse graph of $\cF$, $\MG(\cF)$, is the partially ordered set $(\RC(\cF), \preceq)$. 
\end{defn}

\begin{rem} The Morse graph is related via Birkhoff's Theorem~\cite{birkhoff} to the lattice of \emph{attracting sets}, which is a collection of positively invariant sets of $\cF$ \cite{lattice_structures_I,lattice_structures_II,lattice_structures_III}. In the Boolean networks literature, positively invariant sets are called \emph{trap sets} \cite{Rozum2021,Trinh2025}. When trap sets are also subspaces of $\bbB^N$ where one or more variables have fixed values, they are called \emph{trap spaces}. Collections of trap spaces have been  organized into \emph{succession diagrams} \cite{zanudo2013,Rozum2021}, which aim to describe the attractor landscape of the Boolean systems. Algorithms to compute trap spaces~\cite{trap_space_complexity} and succession diagrams~\cite{Trinh2025} have been developed and implemented. Since not all the trap sets are trap spaces, and identifying the \textit{a priori} existence of trap sets that are not trap spaces remains an open problem, the precise relationship between the succession diagrams of trap spaces and the lattices of trap sets is elusive.
\end{rem}

We represent the Morse graph as the transitive reduction of the directed acyclic graph defined by the poset $(\RC(\cF), \preceq)$, that is, we define a directed acyclic graph whose vertices are $\RC(\cF)$ and edges are $\setdef{(C, C') \in \RC(\cF) \times \RC(\cF)}{C' \prec C}$ and represent $\MG(\cF)$ as the transitive reduction of this graph. This representation of the Morse graph is called the Hasse diagram of the poset $(\RC(\cF),\preceq)$.

We label the nodes of the Morse graph as follows.

\begin{defn}
\label{def:Boolean_MG_labels}
Let $\cZ = \bbB^N$ or $\cZ = \cD$ and consider a multivalued map $\cF \colon \cZ \rightrightarrows \cZ$. Let $q$ be a node in the Morse graph $\MG(\cF)$ and let $C_q$ be the associated strongly connected component. We label the node $q$ as follows:
\begin{enumerate}
\item[(i)] If $C_q = \setof{v}$ consists of a single state $v \in \cZ$, then we call $q$ a \emph{Fixed Point} and label it as $FP(v)$.
\item[(ii)] If $C_q = \setof{v_1, \ldots, v_k}$ for $k \geq 2$ and there exists $i \in \setof{1, \ldots, N}$ such that $(v_1)_i = (v_2)_i = \cdots = (v_k)_i$, then we call $q$ a \emph{Partial Cycle} and label it as $PC$.
\item[(iii)] Otherwise we call $q$ a \emph{Full Cycle} and label it as $FC$. 
\end{enumerate}
\end{defn}

The Boolean and multilevel Morse graphs are defined as follows.

\begin{defn}
\label{def:Boolean_MG}
The \emph{Boolean Morse graph} of a monotone Boolean model $f \colon \bbB^N \to \bbB^N$ is the Morse graph $\MG(\cF)$ of the multivalued map $\cF \colon \bbB^N \rightrightarrows \bbB^N$ given by Definition~\ref{def:Boolean_update}. The nodes of $\MG(\cF)$ are labeled according to Definition~\ref{def:Boolean_MG_labels}.
\end{defn}

\begin{defn}
\label{def:multi_Boolean_MG}
The \emph{multilevel Morse graph} of a regulatory network $RN$ at a parameter $(h, \theta) \in \PG$ is the Morse graph $\MG(\cG)$ of the multivalued map $\cG \colon \cD \rightrightarrows \cD$ given by Definition~\ref{def:update}. The nodes of $\MG(\cG)$ are labeled according to Definition~\ref{def:Boolean_MG_labels}.
\end{defn}

\begin{ex}
\label{ex:EMT_Morse_graph}
Consider $f_\EMT$ as in Example~\ref{ex:EMT_boolean} and the state $b=(0,1,0,1,0,1)$. Substituting in~\eqref{eq:EMT_boolean} yields $f_\EMT(b)=b$, so Definition~\ref{def:Boolean_update} implies $\cF(b)=\setof{b}$. Thus $\setof{b}$ is a strongly connected component with the edge $(b,b)$ and no outgoing edges to other components. In particular,
\[
\setof{b}\in\RC(\cF),
\]
and the corresponding node of the Boolean Morse graph $\MG(\cF)$ has label $FP(b)$ as in Definition~\ref{def:Boolean_MG_labels}.
\end{ex}

A DSGRN parameter $(h, \theta) \in \PG$ gives rise, via equation \eqref{eq:Boolean_wall_labeling}, to a wall labeling $\omega \colon W(\cX) \to \{ \pm 1\}$ which in turn generates a multivalued map $\cF \colon \cX \rightrightarrows \cX$ on the cubical complex $\cX$ defined in Definition~\ref{defn:Xcomplex} (see \cite{rook_field:24} for details). The multivalued map $\cF \colon \cX \rightrightarrows \cX$ generates the \emph{DSGRN Morse graph} as defined below.

\begin{rem}
\label{DSGRN_STG}
A wall labeling $\omega \colon W(\cX) \to \{ \pm 1\}$ actually generates four multivalued maps $\cF_i \colon \cX \rightrightarrows \cX$, $i = 0, 1, \ldots, 3$, at different levels of refinement (see \cite{rook_field:24}). For the computations in this paper we use the most refined multivalued map $\cF_3 \colon \cX \rightrightarrows \cX$ (see Section~\ref{sec:code}). In this paper we refer to this map simply as $\cF \colon \cX \rightrightarrows \cX$. Under the one-to-one correspondence between the cells of $\cX$ and the $N$-dimensional cells of the blowup cubical complex $\cX_b$, the state transition graph of $\cF_3$ is the DSGRN STG drawn on $\cX_b$ in Section~\ref{sec:examples_intro}.
\end{rem}

As indicated in the introduction, the DSGRN models have two distinguishing features. First, the Morse graphs obtained from the combinatorial models can be mapped to Morse decompositions of ODEs. Second, the fact that $\cX$ can be viewed as a cubical cell complex \cite{lefschetz} allows for algebraic topological, in particular homological computations. Furthermore, the computed homological invariants, called \emph{Conley indices}, can be used to deduce the existence of invariant sets, e.g., fixed points and periodic orbits, for ODEs. We refer the reader to \cite{rook_field:24,mischaikow:mrozek} for further details.

For the purposes of this paper it is sufficient to remark that each node in the Morse graph can be assigned a Conley index. Furthermore, given a regulatory network with $N$ vertices, the Conley index for an associated ODE can be represented as a vector $(c_0,\ldots, c_N)$ of non-negative integers (the Betti numbers of the Conley index). Of particular note are the following representations of Conley indices.
\begin{enumerate}
\item If $(c_0,\ldots,c_N) =\bzero^{(k)}$, then this is the \emph{Conley index of a fixed point} with unstable manifold of dimension $k$.
\item If $(c_0,\ldots,c_N) =\bzero^{(k)}+\bzero^{(k+1)}$, then this is the \emph{Conley index of a periodic orbit} with  $k$  Floquet multipliers with magnitude greater than 1.
\item If $(c_0,\ldots,c_N) =\bzero$, then this is the Conley index of the empty set.
\end{enumerate}

In an ODE, the Conley index of a fixed point implies the existence of a fixed point. However, the Conley index of a periodic orbit only suggests the existence of a periodic orbit. Additional conditions need to be satisfied before the existence of a periodic orbit can be guaranteed \cite{mccord:mischaikow:mrozek}. Of greatest significance is the fact that if the Conley index does not equal $\bzero$, then for an associated ODE the existence of a nontrivial invariant set is guaranteed.

\begin{defn}
\label{def:DSGRN_MG}
The \emph{DSGRN Morse graph} of a regulatory network $RN$ at a parameter $(h, \theta) \in \PG$ is the Morse graph $\MG(\cF)$ of the DSGRN multivalued map $\cF \colon \cX \rightrightarrows \cX$ defined above. Let $q$ be a node in the Morse graph $\MG(\cF)$ and let $k$ denote the number of cells in the corresponding strongly connected component. We assign to node $q$ the label $\ell(q)$ defined by
\[
\ell(q) :=
\begin{cases}
FP(k), & \text{if}~ q ~\text{has the Conley index of a fixed point} \\
PO(k), & \text{if}~ q ~\text{has the Conley index of a periodic orbit} \\
T(k),  & \text{if}~ q ~\text{has the trivial Conley index} \\
M(k),  & \text{otherwise}.
\end{cases}
\]
\end{defn}

\begin{rem}
For the multivalued map $\cF \colon \cX \rightrightarrows \cX$, defined in \cite{rook_field:24}, there is a notion of a gradient direction which can be used to conclude that some recurrent components are trivial (in the sense that they do not represent any dynamics of the model). So we define $\GRC(\cF)$ as the set of recurrent components in $\RC(\cF)$ for which there is a common gradient direction and define the Morse graph as $\MG(\cF) := \RC(\cF) \setminus \GRC(\cF)$ with partial order $\preceq$.
\end{rem}

\subsection{Cost of computing Morse graphs}
\label{sec:cost}

Observe that there are three steps to identifying a Morse graph: the identification of the condensation graph, the identification of the recurrent components, and finally the identification of the partial order $\preceq$. Our perspective on quantifying the computational costs is based on the empirical observation that the number of recurrent components is typically orders of magnitude smaller than the combinatorial state space $\cZ$. Thus, we restrict our attention to the computation of the condensation graph that, using standard algorithms \cite{tarjan}, can be done in $O(|\cZ|+|E|)$ steps where $|\cZ|$ is the number of elements of $\cZ$ and $|E|$ is the number of edges in $G(\cF)$.

Given a monotone Boolean model $f\colon \B^N\to\B^N$, an upper bound on the computational cost of computing the associated Morse graph is
\[
O(N2^N).
\]
This bound follows from the fact that for the associated multivalued map $\cF\colon \B^N\rightrightarrows \B^N$ the number of vertices is $2^N$ and each vertex can have at most $2N+1$ edges.

Given a regulatory network $RN = (V, E, \delta)$ and using an associated finite multilevel model, an upper bound on the computational cost of computing the Morse graph is
\[
O\left( N\prod_{n=1}^N \left( |\Targets(n)| + 1 \right) \right).
\]
This bound follows from the fact that for the associated multivalued map $\cG\colon \cD \rightrightarrows \cD$ the number of vertices is $\prod_{n=1}^N \left( |\Targets(n)| + 1 \right)$ and each vertex can have at most $2N+1$ edges.

Given a regulatory network $RN = (D,\delta) =(V, E, \delta)$ and using a DSGRN multivalued map $\cF_i\colon \cX\rightrightarrows \cX$, an upper bound on the computational cost of computing an associated Morse graph is
\[
O \left( N 2^N\prod_{n=1}^N \left( |\Targets(n)| + 2 \right) \right).
\]
This bound follows from the fact that each element of $\cX$ takes the form $[v, \bw]$ where $v \in \cX^{(0)}$ and $\bw\in \B^N$ and each element of $\cX$ can have at most $2N+1$ edges.

As indicated in Remark~\ref{DSGRN_STG} the various DSGRN multivalued maps $\cF_i$ allow for different levels of expression of dynamics and in general the more expressive the potential dynamics the higher the computational cost. However, the asymptotic bound given above is applicable to all the DSGRN multivalued maps developed at the time of this writing.

We do not claim that these estimates are sharp. However, they do make clear that the expected computational cost increases significantly as one moves from a Boolean to a finite multilevel to a DSGRN model. 

We measured the time needed to compute a DSGRN Morse graph together with the Conley indices, for random essential parameters of every network in this paper. All timings were obtained with DSGRN v1.10.0 on a laptop (Intel Core i7-8750H, 2.2~GHz, 16~GB RAM) using a single core. The results are shown in Table~\ref{tab:benchmarks}. The cost is governed by the size of the cubical complex $\cX_b$; more specifically by the number of $N$-dimensional cells in $\cX_b$, which gives the number of vertices of the DSGRN state transition graph. This number is given by
\[
\prod_{i = 1}^{N}{2 (m_i + 2)},
\]
and it depends only on the number of output edges $m_i$ of each node. The number of input edges to each node only affects the size of the parameter graph and have no impact on the cost per parameter.

A parameter of the EMT network ($276{,}480$ cells) takes between $6$ and $12$ minutes and requires less than $1$~GB of memory. By extrapolation, complexes with about $10^6$ cells, corresponding for example to seven nodes with $1$ or $2$ out edges or five to six nodes with $3$ out edges, would take on the order of $30$ minutes and require about $3$~GB of memory per parameter. A network with ten nodes with $2$ out edges would produce about $10^9$ cells, which is impractical for the current DSGRN implementation.

\begin{table}[!htb]
\centering
\renewcommand{\arraystretch}{1.2}
\begin{tabular}{@{}lrrrr@{}}
\toprule
Network                & $N$ & Cells       & Time per parameter & Essential parameters \\ \midrule
Toggle switch          & $2$ & $36$        & $0.006$~s          & $1$ \\
Two nodes, three edges & $2$ & $48$        & $0.007$~s          & $14$ \\
Two nodes, four edges  & $2$ & $64$        & $0.008$~s          & $196$ \\
Three nodes            & $3$ & $480$       & $0.11$~s           & $2{,}352$ \\
Five nodes             & $5$ & $34{,}560$  & $24$~s             & $1{,}534{,}256{,}640$ \\
EMT                    & $6$ & $276{,}480$ & $543$~s            & $15{,}371{,}790{,}336$ \\
\bottomrule
\end{tabular}
\caption{Cost of the DSGRN computation per parameter node and size of the parameter graph: Number of nodes $N$, number of cells of the cubical complex $\cX_b$, average time in seconds per parameter node, and number of essential parameters.}
\label{tab:benchmarks}
\end{table}

While the computational cost of exhaustive DSGRN characterization of large networks is prohibitive, there are several points to consider. First, a random sampling of network dynamics can estimate the distribution of attracting Morse sets across the parameter graph. Second, we can perform targeted searches, such as in Table~\ref{tab:EMT_Boolean_EM_states}, where network dynamics within a radius of initial Boolean parameters can provide information on the variability of dynamics near a reasonable model. Lastly, it has been suggested that many relevant biological phenomena are controlled by relatively small networks. For instance, large-scale genetic oscillations are controlled by relatively small genetic ``core oscillators'' \cite{cummins2022}, and  small networks may control transitions to steady cellular states. Therefore the study of smaller networks driving large-scale dynamics can provide insight into important cellular processes.

\subsection{Multilevel refinement of Boolean models}
\label{sec:multilevel_refinement}

In this section we establish a formal relationship between asynchronous Boolean models and multilevel models. Specifically, we show that the DSGRN Boolean parameters give rise to \emph{multilevel refinements} of the corresponding asynchronous Boolean models.
We remark that multilevel refinements of  Boolean models were introduced in \cite{Pauleve_MP} in the context of \emph{most permissive update (MP)} of Boolean models. This paper has shown not only that every transition permitted by asynchronous update is also permitted by the MP update, but also that any transition permitted by any multilevel refinement of the Boolean model is also captured by the MP update.

Let $RN$ be a regulatory network with $V = \setof{1, \ldots, N}$ and let $f \colon \bbB^N \to \bbB^N$ be a monotone Boolean model for $RN$ (see Definition~\ref{def:boolean_model_RN}). We assume that $RN$ is a fully connected network. This assumption is made solely for the sake of presentation, the results presented here are valid without this assumption. Recall that a DSGRN parameter is defined in terms of monotone increasing Boolean functions (see Definition~\ref{def:fi}). In particular, associated to $f$, there is an increasing monotone Boolean model $\tilde{h} \colon \bbB^N \to \bbB^N$ via the change of variables $\beta$ in \eqref{eq:beta}. More specifically, for each $i \in V$ there is a change of variables $\beta^i \colon \bbB^N \to \bbB^N$ given by \eqref{eq:beta} such that
\begin{equation}
\label{eq:f_f_tilde}
f_i(b) = \tilde{h}_i(\beta^i(b)).
\end{equation}

There are two types of multivalued maps associated to $f$. One is the asynchronous update multivalued map $\cF \colon \bbB^N \rightrightarrows \bbB^N$ generated by $f$, given by Definition~\ref{def:Boolean_update}, which gives rise to the Boolean dynamics of the network.

The other is the nearest neighbor asynchronous update multivalued map $\cG \colon \cD \rightrightarrows \cD$ (see Definition~\ref{def:update}), which is used to compute the finite multilevel dynamics of the network. The map $\cG$ is defined in terms of a DSGRN parameter as follows. Given an order parameter $\theta$, the map $\tilde{h}$ defines a Boolean DSGRN parameter $(h, \theta) \in \PG$, where
\begin{equation}
\label{eq_f_hat}
h = \left( (\tilde{h}_1, \ldots, \tilde{h}_1), (\tilde{h}_2, \ldots, \tilde{h}_2), \ldots, (\tilde{h}_N, \ldots, \tilde{h}_N) \right),
\end{equation}
as defined by \eqref{eq:Boolean_par}. The parameter $(h, \theta)$ is then used to define the finite multilevel update function $g \colon \cD \to \cD$ given by \eqref{Psi}, which in turn generates the nearest neighbor asynchronous update multivalued map $\cG \colon \cD \rightrightarrows \cD$.

The goal of this section is to show that the multilevel update function $g$ is a multilevel refinement of the Boolean model $f$. The refinement criteria rely on the notion of a binarization of the multilevel configurations $d \in \cD$. Our definition of binarization is a special case of the definition in \cite{Pauleve_MP}, since it requires monotonicity of the binarization function. 

\begin{defn}
\label{defn:binarization}
A \emph{binarization} of the set $X = \setof{0,1,\ldots,m}$ is a monotone increasing function  $\mu \colon X \to \bbB$ satisfying $\mu(0) = 0$ and $\mu(m) = 1$. A \emph{binarization} of $\cD = \prod_{i=1}^N X_i$ is a map $\mu \colon \cD \to \bbB^N$ given by
\[
\mu(d) = (\mu_1(d_1), \mu_2(d_2), \ldots, \mu_N(d_N)),
\]
where $\mu_i \colon X_i \to \bbB$ is a binarization of $X_i$ for $i = 1, \ldots, N$.
\end{defn}

\begin{defn}
\label{defn:multilevel_refinement}
A multilevel update function $g \colon \cD \to \cD$ is a \emph{trivial multilevel refinement} of a Boolean model $f \colon \bbB^N \to \bbB^N$ if for each $i \in \setof{1, \ldots, N}$ there exists a binarization 
$\mu^i \colon \cD \to \bbB^N$ 
such that the following holds for every $d \in \cD$
\begin{enumerate}[(i)]
\item If $g_i(d) < d_i$, then $f_i(\mu^i(d)) = 0$,
\item If $g_i(d) > d_i$, then $f_i(\mu^i(d)) = 1$,
\item If $g_i(d) = d_i$, then $f_i(\mu^i(d)) = b_i$, where $b = \mu^i(d)$.
\end{enumerate}
\end{defn}
For our Boolean DSGRN parameter $(h, \theta) \in \PG$ and for $j \in \setof{1, \ldots, N}$, define $\mu^j \colon \cD \to \bbB^N$ by
\begin{equation}
\label{eq:binaritazion_muj}
\mu^j(d) = (\mu^j_1(d_1), \ldots, \mu^j_N(d_N)),
\end{equation}
where $\mu^j_i \colon X_i \to \bbB$ is defined by \eqref{eq:mu}.

\begin{lem}
\label{lem:binarization}
For each $j \in \setof{1, \ldots, N}$, the map $\mu^j \colon \cD \to \bbB^N$ defined by \eqref{eq:binaritazion_muj} is a binarization of $\cD$ which satisfies $f_j(\mu^j(d)) = \tilde{h}_j(B^j(d))$ for all $d \in \cD$, where $B^j$ is given by \eqref{def:B}.
\end{lem}

\begin{proof}
Let $j \in \setof{1, \ldots, N}$. For each $i = 1, \ldots, N$ it follows directly from \eqref{eq:mu} that $\mu^j_i \colon X_i \to \bbB$ is monotone increasing, $\mu^j_i(0) = 0$, and $\mu^j_i(m_i) = 1$. Hence $\mu^j_i$ is a binarization of $X_i$ and therefore $\mu^j$ is a binarization of $\cD$.

From the definitions of $\beta^j$ given by \eqref{eq:beta} and $B^j$ given by \eqref{def:B} it follows that $\beta^j(\mu^j(d)) = B^j(d)$. Hence by \eqref{eq:f_f_tilde} we have that
\[
f_j(\mu^j(d)) = \tilde{h}_j(\beta^j(\mu^j(d))) = \tilde{h}_j(B^j(d)).
\]
\end{proof}

\begin{thm}
\label{thm:multilevel_refinement}
Let $RN$ be a regulatory network with $V = \{ 1, \ldots, N \}$. Let $f \colon \bbB^N \to \bbB^N$ be a monotone Boolean model for $RN$ and let $(h, \theta) \in \PG$ be a corresponding parameter node as defined above. Let $g \colon \cD \to \cD$ be the multilevel function defined by $(h, \theta)$. Then $g$ is a trivial multilevel refinement of the Boolean model $f$.
\end{thm}

\begin{proof}
For $i \in \setof{1, \ldots, N}$ let $\mu^i \colon \cD \to \bbB^N$ be the map defined by \eqref{eq:binaritazion_muj}. By Lemma~\ref{lem:binarization} it follows that $\mu^i$ is a binarization of $\cD$ and $f_i(\mu^i(d)) = \tilde{h}_i(B^i(d))$ for all $d \in \cD$.

Furthermore notice that, by \eqref{Psi} and \eqref{eq_f_hat}, we have
\[
g_i(d) = \sum_{j \in \Targets(i)} h_i^j(B^i(d)) = \sum_{j \in \Targets(i)} \tilde{h}_i(B^i(d)) = | \Targets(i) | \tilde{h}_i(B^i(d)) = m_i \tilde{h}_i(B^i(d)).
\]
Hence
\begin{equation}
\label{eq:g_i_values}
g_i(d) =
\begin{cases}
0,  & \text{if}~  \tilde{h}_i(B^i(d)) = 0 \\
m_i, & \text{if}~ \tilde{h}_i(B^i(d)) = 1.
\end{cases}
\end{equation}
From this it follows that if $g_i(d) < d_i \leq m_i$, then $g_i(d) = 0$ and so $\tilde{h}_i(B^i(d)) = 0$. Hence
\[
f_i(\mu^i(d)) = \tilde{h}_i(B^i(d)) = 0,
\]
which proves (i). Analogously, if $g_i(d) > d_i \geq 0$, then $g_i(d) = m_i$ and so $\tilde{h}_i(B^i(d)) = 1$. Thus
\[
f_i(\mu^i(d)) = \tilde{h}_i(B^i(d)) = 1,
\]
which proves (ii).

To prove (iii) let $b = \mu^i(d)$. If $g_i(d) = d_i$, then it follows from \eqref{eq:g_i_values} that $d_i = g_i(d) = 0$ or $d_i = g_i(d) = m_i$.

If $d_i = g_i(d) = 0$ then, since $\mu^i$ is a binarization, we have that $b_i = 0$. Additionally, $f_i(\mu^i(d)) = 0$ as above. Hence $f_i(\mu^i(d)) = b_i$.

Similarly, if $d_i = g_i(d) = m_i$, then $b_i = 1$ and $f_i(\mu^i(d)) = 1$. Thus $f_i(\mu^i(d)) = b_i$.

Therefore $g$ is a trivial multilevel refinement of $f$.
\end{proof}

\begin{ex}
\label{ex:EMT_refinement}
Let $\hat g$ be the multilevel update function defined by $(\hat h,\theta)$, with $\hat h$ as in Example~\ref{ex:EMT_Boolean_parameters} and $\theta$ as in Example~\ref{ex:EMT_order_adjacency}. For each $i\in V$, let $\mu^i\colon\cD\to\bbB^6$ be a binarization with $\mu_j^i$ defined by~\eqref{eq:mu} for $j\in\Sources(i)$. Then
\[
\hat g_i(d)=m_i f_i\bigl(\mu^i(d)\bigr)
\quad\forall i\in V,\ \forall d\in\cD.
\]
By Theorem~\ref{thm:multilevel_refinement}, $\hat g$ is a trivial multilevel refinement of $f_\EMT$.
\end{ex}

\section{Discussion}\label{sec:discuss}

We have demonstrated a path leading from monotone Boolean models to ODEs with steep nonlinearities. This was accomplished by embedding Boolean models into multilevel models, which themselves embed into the DSGRN methodology that rigorously captures topological properties of continuous systems. 

As has been demonstrated by studying feed-forward motifs~\cite{AlonFFL1,AlonFFL2} of gene regulation, the dynamics of Boolean (and multilevel) models cannot capture some transient dynamics that may be important for regulatory dynamics. In this paper, we show that asynchronous Boolean dynamics may also miss recurrent dynamics, including attractors, that are consistent with continuous models of the same network. These dynamics can be captured by considering a larger, but specific, class of multilevel discrete models.
The ability to obtain a more complete picture of the potential dynamics of continuous models should be balanced against the need to examine the dynamics of larger numbers of discrete models than monotone Boolean models. This trade-off already exists when Boolean models are used to study large ($>20$ nodes) networks. One typically cannot examine all compatible monotone Boolean models, but focuses on one or a few closely related neighboring models.
 The DSGRN software provides a comprehensive description of state transition graph dynamics possible for a regulatory network while also providing a natural bridge to ODE dynamics.

DSGRN permits a more systematic exploration of regulatory network dynamics under a change of parameters, since it tracks changes in the underlying continuous dynamics via discrete changes in Morse graphs induced by paths in the parameter graph of multilevel systems. This allows  the description of effects of gene overexpression or suppression, rather than its complete knockout, on the attractors of the network. As such, this ability has a potential to enable refined modeling of  genetic disruptions leading to disease states.

\section{Lattices of monotone increasing Boolean functions}
\label{sec:lattices}

\begin{defn}
The set of monotone increasing Boolean functions with $k$ inputs $f \colon \bbB^k \to \bbB$ is denoted by $\cM_k$.
\end{defn}

\begin{defn}
The \emph{truth set} of a monotone Boolean function $f \colon \bbB^k \to \bbB$ is defined as
\[
\cT(f) :=\{ z \in \bbB^k \;|\; f(z) =1 \}.
\]
\end{defn}

The set $\cM_k$ forms a lattice, where the join ($\vee$) and meet ($\wedge$) operations are defined by union and intersection of the corresponding truth sets, i.e.,
\begin{align*}
h = f \vee g   & \iff \cT(h) = \cT(f) \cup \cT(g), \\
h = f \wedge g & \iff \cT(h) = \cT(f) \cap \cT(g).
\end{align*}

The lattice $\cM_2$ is shown in Figure~\ref{fig:MBF_PG_multi} on the left. On the right is a set of logic parameters for a network node with two inputs and two outputs (see \eqref{eq:logical}). Each logic parameter is a pair $(h_1, h_2)$ where $h_1, h_2 \in \cM_2$ and $h_1 \succeq h_2$ (see \eqref{eq:ordering condition}). This is a partially ordered set with adjacency defined in Definition~\ref{def:adjacency_DSGRN}. Finally, the red nodes on the right are embedded Boolean logic parameters (see \eqref{eq:Boolean_par}). We note that the poset on the right is isomorphic to the lattice $\cM_3$, based on the arguments in~\cite{crawford22,Gedeon24b}.

\begin{figure}[htp!]
\centering
\includegraphics[width=0.72\linewidth]{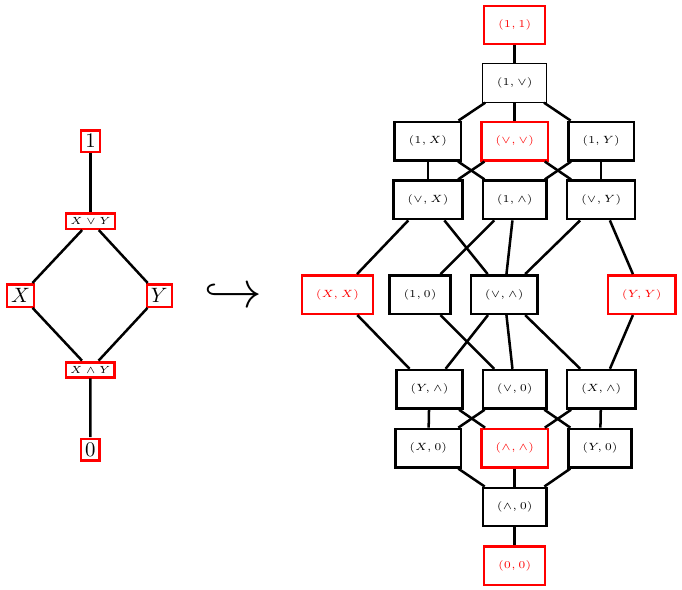}
\caption{\textbf{Left:} Lattice $\cM_2$ of monotone Boolean functions forms a set of logic parameters for a network node with $2$ in-edges and $1$ out-edge. \textbf{Right:} DSGRN logic parameter graph for a 2-in, 2-out network node.  The red nodes are the embedded Boolean logic parameter nodes.}
\label{fig:MBF_PG_multi}
\end{figure}

\section{Code Usage}
\label{sec:code}

In this section we describe via examples how to use the main functionalities of DSGRN. For a more complete description of DSGRN, see the tutorials available at \cite{DSGRN_repo}. All computations in this paper were performed with DSGRN version 1.10.0. Jupyter notebooks with code to run the examples in this paper are available at \url{https://github.com/marciogameiro/DSGRN_Boolean_paper}.

DSGRN is a Python package and can be installed (assuming Python is installed) with:
\begin{center}
\texttt{pip install DSGRN}
\end{center}

\subsection*{Defining a Regulatory Network}

The first step is to specify a regulatory network (Definition~\ref{def:RN}) using DSGRN string based specification format. Each node of the network is defined by a string, which defines the node and a list of its input nodes. The format is
\begin{center}
\texttt{node\_name : algebraic\_expression : optional\_modifiers},
\end{center}
where \texttt{node\_name} specifies the network node, \texttt{algebraic\_expression} specifies the input nodes to \texttt{node\_name}, and \texttt{optional\_modifiers} are the optional modifiers \texttt{M}, \texttt{B}, or \texttt{E} and determine which parameter graph to construct as described in the next section.

For example, to define the network in Figure~\ref{fig:22network}(a) we can use the following strings

{\centering
\begin{align*}
\texttt{v1 :} & \texttt{\ \ v1 + v2 : M} \\
\texttt{v2 :} & \texttt{\ \~{}v1 + v2 : M}
\end{align*}
}%
which specify a network with two nodes $v_1$ and $v_2$, where the input nodes to $v_1$ are $v_1$ and $v_2$ and the input nodes to $v_2$ are also $v_1$ and $v_2$. A positive edge is specified by denoting the corresponding input node by $v_i$ and a negative edge by denoting the input node by $\sim v_i$. Hence in the network given above the edge $(v_1, v_2)$ is negative and the other three edges are positive (see Figure~\ref{fig:22network}(a)). The optional flag \texttt{M} has no effect on the construction of the network and only affects the associated parameter graph, as described in the next section. The meaning of the algebraic expressions \texttt{v1 + v2} and \texttt{\~{}v1 + v2} is also discussed in the next section (for now just think of them as lists of input nodes).

After importing DSGRN with the command
\begin{lstlisting}[language=Python]
import DSGRN
\end{lstlisting}
the network above can be constructed in DSGRN with the commands
\begin{lstlisting}[language=Python]
net_spec = """v1 :  v1 + v2 : M
              v2 : ~v1 + v2 : M"""

network = DSGRN.Network(net_spec)
\end{lstlisting}

Once a DSGRN Network object is constructed as above, the network graph can be plotted with the command below. This will draw the network graph, as in Figure~\ref{fig:22network}(a). 
\begin{lstlisting}[language=Python]
DSGRN.DrawGraph(network)
\end{lstlisting}

\subsection*{Constructing the Parameter Graph}

The DSGRN parameter graph is constructed from a network with the following command
\begin{lstlisting}[language=Python]
parameter_graph = DSGRN.ParameterGraph(network)
\end{lstlisting}
The optional modifier \texttt{M} used in the definition of the network tells DSGRN to construct the \emph{DSGRN multi Boolean parameter graph} (see Definition~\ref{def:DSGRN_multilevel_parameter_graph}). We inquire DSGRN about the size of the parameter graph
\begin{lstlisting}[language=Python]
print(parameter_graph.size())  # Output: 1600
\end{lstlisting}
The output indicates that we have $1{,}600$ parameter nodes (vertices in the parameter graph). We can then interact with the parameter graph and compute dynamics for each node of the parameter graph as described in what follows.

\begin{rem}
In the original DSGRN framework~\cite{Cummins2017b, kepley:mischaikow:zhang} the DSGRN parameter graph is constructed via a decomposition of the parameter space in terms of semi-algebraic sets. We refer to that parameter graph as the \emph{classical DSGRN parameter graph}. For networks with nodes with more than two input edges, the classical parameter graph depends on the algebraic expressions of the input nodes. Hence, for example, the expressions \texttt{v1 : v1 + v2 + v3} and \texttt{v1 : v1 (v2 + v3)} will produce different classical parameter graphs. For the \emph{DSGRN multi Boolean parameter graph} (Definition~\ref{def:DSGRN_multilevel_parameter_graph}) and the \emph{DSGRN Boolean parameter graph} (Definition~\ref{def:DSGRN_Boolean_parameter_graph}) only the input nodes are considered and the specific format of the algebraic expressions is ignored by DSGRN, and hence we can always use sums as in \texttt{v1 : v1 + v2 + v3} (see Remark~\ref{rem:PG_sizes_diffs}).
\end{rem}

If we use the optional modifier \texttt{B} instead of the modifier \texttt{M} when defining the network, then DSGRN will construct the \emph{DSGRN Boolean parameter graph} (see Definition~\ref{def:DSGRN_Boolean_parameter_graph}). If neither the modifier \texttt{M} nor \texttt{B} is used, then DSGRN will construct the \emph{classical DSGRN parameter graph}. The \emph{DSGRN Boolean parameter graph} is a subgraph of the \emph{DSGRN multi Boolean parameter graph}. In all of these cases, if we additionally use the modifier \texttt{E}, then DSGRN will construct the \emph{essential parameter graph} (see Definition~\ref{def:DSGRN_essential_parameter_graph}), which is always a subgraph of the corresponding parameter graph.

Hence, the following commands define the classical DSGRN parameter graph (note the different algebraic expression for node \texttt{v2})
\begin{lstlisting}[language=Python]
net_spec = """v1 : v1 + v2
              v2 : (~v1)(v2)"""

network = DSGRN.Network(net_spec)

parameter_graph = DSGRN.ParameterGraph(network)
print(parameter_graph.size())  # Output: 1600
\end{lstlisting}

The following commands define the DSGRN Boolean parameter graph
\begin{lstlisting}[language=Python]
net_spec = """v1 :  v1 + v2 : B
              v2 : ~v1 + v2 : B"""

network = DSGRN.Network(net_spec)

parameter_graph = DSGRN.ParameterGraph(network)
print(parameter_graph.size())  # Output: 144
\end{lstlisting}

The following commands define the essential DSGRN multi Boolean parameter graph
\begin{lstlisting}[language=Python]
net_spec = """v1 :  v1 + v2 : M E
              v2 : ~v1 + v2 : M E"""

network = DSGRN.Network(net_spec)

parameter_graph = DSGRN.ParameterGraph(network)
print(parameter_graph.size())  # Output: 196
\end{lstlisting}

The following commands plot the parameter graph for the network in Figure~\ref{fig:network_boolean}(a).
\begin{lstlisting}[language=Python]
net_spec = """v1 : v1 + v2
              v2 : ~v1"""

network = DSGRN.Network(net_spec)
parameter_graph = DSGRN.ParameterGraph(network)
DSGRN.draw_parameter_graph(parameter_graph)
\end{lstlisting}
The resulting parameter graph is shown in Figure~\ref{fig:pg_2_nodes_3_edges}. The nodes in the parameter graph are indexed by an integer index starting at zero (from $0$ to $119$ in this example), which are shown as the labels of the parameter nodes.

\begin{figure}[!htpb]
\centering
\includegraphics[width=1.0\linewidth]{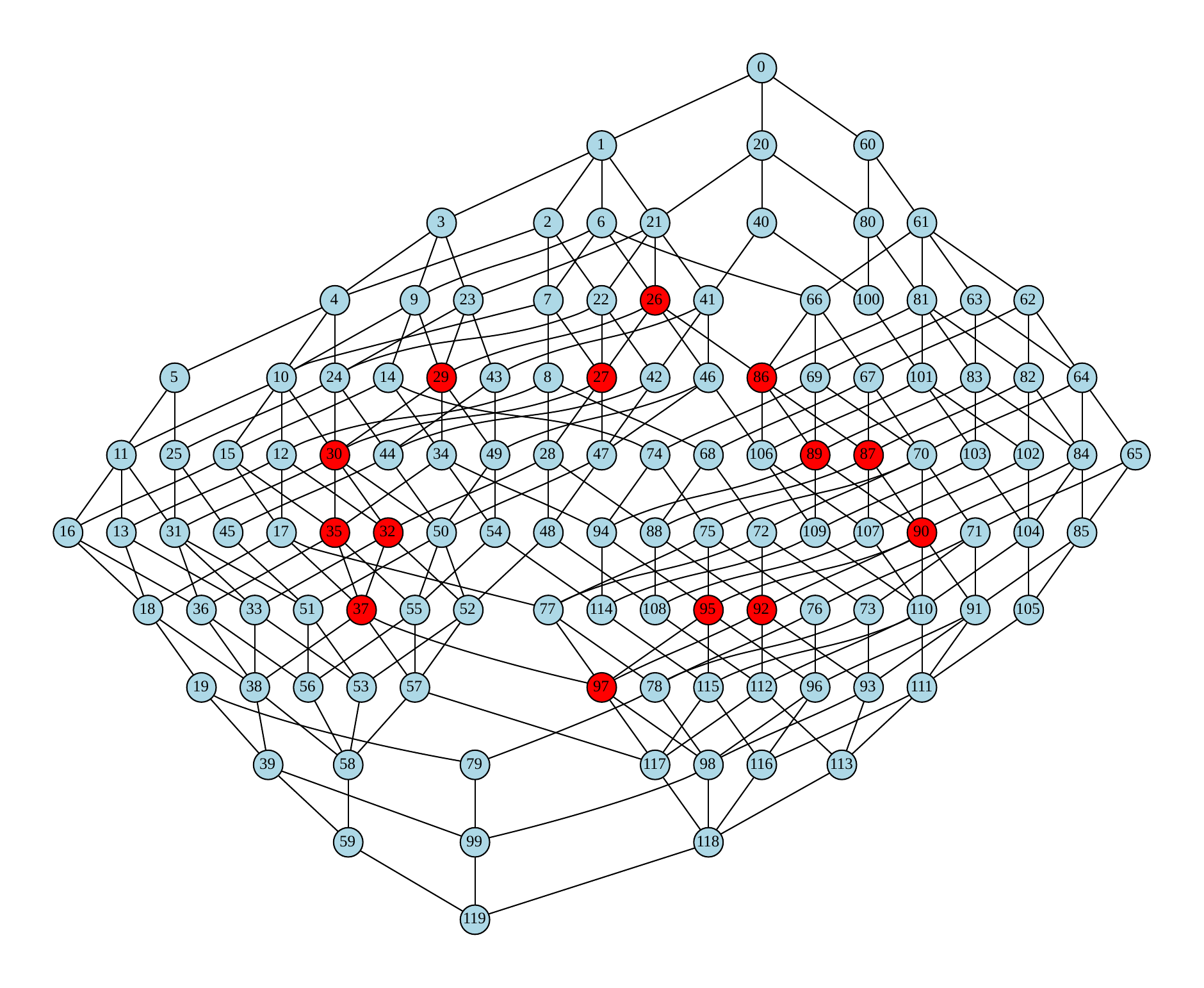}
\caption{Parameter graph for the network in Figure~\ref{fig:network_boolean}(a). The essential nodes are shown in red.}
\label{fig:pg_2_nodes_3_edges}
\end{figure}

Now let us redefine and explore the DSGRN multi Boolean parameter graph for the network in Figure~\ref{fig:22network}(a)
\begin{lstlisting}[language=Python]
net_spec = """v1 :  v1 + v2 : M
              v2 : ~v1 + v2 : M"""

network = DSGRN.Network(net_spec)

parameter_graph = DSGRN.ParameterGraph(network)
\end{lstlisting}

We specify parameters in the parameter graph via their integer indexes. In this example they are indexed from $0$ to $1{,}599$. The command
\begin{lstlisting}[language=Python]
par_index = 1032

parameter = parameter_graph.parameter(par_index)
\end{lstlisting}
selects a particular parameter node from the parameter graph. We can ask DSGRN for information about the parameter with the command
\begin{lstlisting}[language=Python]
print(parameter.partialorders())
\end{lstlisting}
which will display the following output
\begin{lstlisting}[language=Python]
v1 : (p0, t0, p2, t1, p1, p3)
v2 : (t1, p0, p1, p2, t0, p3)
\end{lstlisting}

This output describes the parameter in a compact notation in terms of its order and logic parameters (see Definitions~\ref{defn:order_parameter} and \ref{def:fi}) and should be interpreted as follows:

The \emph{ordering} of the $t_j$ terms in the list gives the order parameter. More specifically, how they are ordered with respect to their indexing represents the order parameter. Hence for node $v_1$ since $t_0$ appears before $t_1$ the order parameter $\theta_1 \colon \setof{1, 2} \to \setof{1, 2}$ is the identity function, that is, $\theta_1 (1) < \theta_1 (2)$. For node $v_2$ since $t_1$ appears before $t_0$ the order parameter $\theta_2 \colon \setof{1, 2} \to \setof{1, 2}$ is given by $\theta_2 (1) = 2$ and $\theta_2 (2) = 1$, that is, $\theta_2 (2) < \theta_2 (1)$.

The \emph{positions} of the $t_j$ terms with respect to the $p_i$ terms give the logic parameter as follows. Each $p_i$ represents the value of a Boolean function $h$ evaluated at the Boolean vector corresponding to the binary digits of $i$, that is, $p_i = h(i)$. Hence we have $p_0 = h(0, 0)$, $p_1 = h(0, 1)$, $p_2 = h(1, 0)$, and $p_3 = h(1, 1)$. Each $t_j$ gives rise to one Boolean function $h$ which is defined by assigning $0$ to the $p_i$ terms before $t_j$ and $1$ to the $p_i$ terms after $t_j$. Hence for node $v_1$ we have the logic parameter $h_1 = (h_1^1, h_1^2)$ where $h_1^1$ is defined by the relative position of $t_0$ and is given by $h_1^1(0, 0) = 0$ and $h_1^1(0, 1) = h_1^1(1, 0) = h_1^1(1, 1) = 1$. Notice that only the relative positions of the $p_i$ with respect to $t_0$ are relevant to the definition of $h_1^1$. The function $h_1^2$ is defined by the position of $t_1$ as $h_1^2(0, 0) = h_1^2(1, 0) = 0$ and $h_1^2(0, 1) = h_1^2(1, 1) = 1$. Hence we have the logic parameter $h_1 = (h_1^1, h_1^2)$, where
\[
h_1^1(b_1, b_2) = b_1 \vee b_2 \quad \text{and} \quad h_1^2(b_1, b_2) = b_2.
\]
Analogously, for node $v_2$ we have the logic parameter $h_2 = (h_2^1, h_2^2)$ where $h_2^1$ is defined by the relative position of $t_1$ and is given by $h_2^1(0, 0) = h_2^1(0, 1) = h_2^1(1, 0) = h_2^1(1, 1) = 1$. The function $h_2^2$ is defined by the position of $t_0$ as $h_2^2(0, 0) = h_2^2(1, 0) = h_2^2(0, 1) = 0$ and $h_2^2(1, 1) = 1$. Notice that the logic parameter $h_2 = (h_2^1, h_2^2)$ is given by
\[
h_2^1(b_1, b_2) = 1 \quad \text{and} \quad h_2^2(b_1, b_2) = b_1 \wedge b_2.
\]

To construct a DSGRN parameter, and hence to compute its dynamics, we need to specify the integer indexing of the parameter node in the parameter graph. However we can ask DSGRN for the parameter index from the partial order information described above as follows
\begin{lstlisting}[language=Python]
partial_orders = ['(p0, t0, p2, t1, p1, p3)',
                  '(t1, p0, p1, p2, t0, p3)']

par_index = DSGRN.index_from_partial_orders(parameter_graph, partial_orders)
print(par_index) # Output: 1032
\end{lstlisting}

Since the partial order information described above represents the order and logic parameters, we can use the above commands to go back and forth between the DSGRN parameter indexing and the order and logic parameters representation. 

\begin{rem}
\label{rem:PG_sizes_diffs}
The classical DSGRN parameter graph is a subgraph of the DSGRN multi Boolean parameter graph. They are identical for networks with nodes having at most two input edges. However, for networks with nodes with more than two input edges the classical DSGRN parameter graph may be a proper subset of the DSGRN multi Boolean parameter graph, and the size of the classical DSGRN parameter graph depends on the algebraic expressions of the input nodes. The examples below illustrate the size differences.
\begin{lstlisting}[language=Python]
net_spec = """v1 : v1 + v2 + ~v3
              v2 : ~v1
              v3 : v1 + v2"""

network = DSGRN.Network(net_spec)
parameter_graph = DSGRN.ParameterGraph(network)
print(parameter_graph.size()) # Output: 305424

net_spec = """v1 : (v1 + v2)(~v3)
              v2 : ~v1
              v3 : v1 + v2"""

network = DSGRN.Network(net_spec)
parameter_graph = DSGRN.ParameterGraph(network)
print(parameter_graph.size()) # Output: 326592
\end{lstlisting}

The size of the DSGRN multi Boolean parameter graph does not depend on the algebraic expressions of the input nodes
\begin{lstlisting}[language=Python]
net_spec = """v1 : v1 + v2 + ~v3 : M
              v2 : ~v1 : M
              v3 : v1 + v2 : M"""

network = DSGRN.Network(net_spec)
parameter_graph = DSGRN.ParameterGraph(network)
print(parameter_graph.size()) # Output: 383184

net_spec = """v1 : (v1 + v2)(~v3) : M
              v2 : ~v1 : M
              v3 : v1 + v2 : M"""

network = DSGRN.Network(net_spec)
parameter_graph = DSGRN.ParameterGraph(network)
print(parameter_graph.size()) # Output: 383184
\end{lstlisting}
\end{rem}

\subsection*{Computing dynamics}

We can compute the dynamics of a DSGRN parameter with the following command 
\begin{lstlisting}[language=Python]
net_spec = """v1 :  v1 + v2 : M
              v2 : ~v1 + v2 : M"""

network = DSGRN.Network(net_spec)

parameter_graph = DSGRN.ParameterGraph(network)

par_index = 926

parameter = parameter_graph.parameter(par_index)
morse_graph, stg, graded_complex = DSGRN.Blowup.ConleyMorseGraph(parameter)
\end{lstlisting}

We can then plot the Morse graph
\begin{lstlisting}[language=Python]
DSGRN.Blowup.PlotMorseGraph(morse_graph)
\end{lstlisting}
and the corresponding color-coded cell complex $\cX_b$ (for dimensions greater than two a projection of $\cX_b$ into two dimensions is plotted) and the state transition graph
\begin{lstlisting}[language=Python]
DSGRN.Blowup.PlotMorseSets(morse_graph, stg, graded_complex)
\end{lstlisting}
This will plot the Morse graph and state transition graphs corresponding to Figure~\ref{fig:order_flow}(b). However, by default DSGRN labels the nodes of the Morse graph with their Conley indices, as shown in Figure~\ref{fig:two_node_mg_conley_index}.

\begin{figure}[!htpb]
\centering
\includegraphics[width=0.35\linewidth]{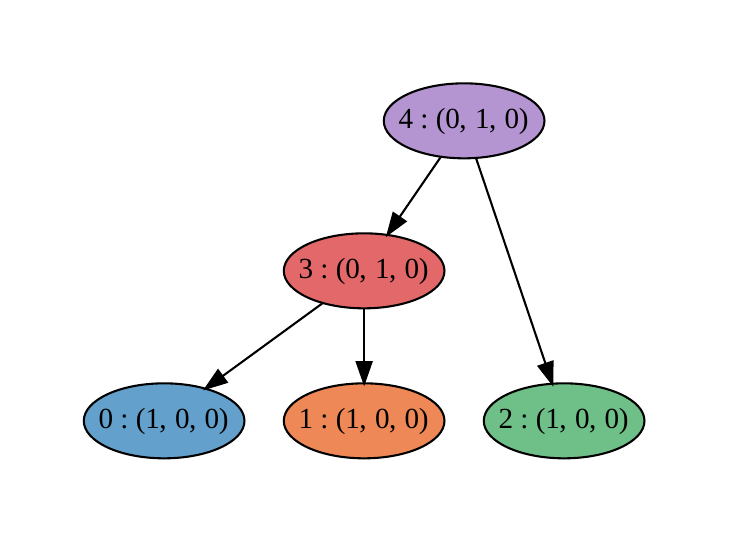}
\caption{Morse graph for parameter node $926$ of the network in Figure~\ref{fig:22network}(a). Each node is labeled with its Conley index.}
\label{fig:two_node_mg_conley_index}
\end{figure}

To plot the Morse graph with the labeling used in Figure~\ref{fig:order_flow} (see Definition~\ref{def:DSGRN_MG}) we can use the option \texttt{label='a'} as follows
\begin{lstlisting}[language=Python]
DSGRN.Blowup.PlotMorseGraph(morse_graph, label='a')
\end{lstlisting}
This will reproduce the Morse graph shown in Figure~\ref{fig:order_flow}(b).

\section*{Acknowledgments}
M.G. was partially supported by the Air Force Office of Scientific Research under awards numbered FA9550-23-1-0011 and FA9550-23-1-0400.
B.C. and T.G. were partially supported by the National Science Foundation under award number 2601860.
K.M. was partially supported by the Air Force Office of Scientific Research under awards numbered FA9550-23-1-0011 and FA9550-23-1-0400.
B.R. was supported by the Air Force Office of Scientific Research under award number FA9550-23-1-0400.

\section*{AI disclosure}
Generative AI was used only for grammar, latex compilation issues, assisting with tikz code, and to write simple Python scripts to check the computations.

\bibliographystyle{plain}
\bibliography{references}

\end{document}